\documentclass[conference]{IEEEtran}

\usepackage{enumitem}
\usepackage{hyperref}

\usepackage{amsmath}
\usepackage{amssymb}
\usepackage{soul}
\usepackage{amsthm}
\hypersetup{
    colorlinks=true,
    linkcolor=black,
    filecolor=black, 
    citecolor=black,
    urlcolor=cyan,
}
\usepackage{graphicx}
\usepackage[framemethod=TikZ]{mdframed}
\newtheorem{definition}{Definition}
\newtheorem{theorem}{Theorem}

\usepackage[ruled, vlined, linesnumbered]{algorithm}
\usepackage[noend]{algpseudocode}
\usepackage{makecell}  
\alglanguage{pseudocode}
\algnewcommand{\Initialize}[1]{%
  \State \textbf{Initialize:}
  \Statex \hspace*{\algorithmicindent}\parbox[t]{.8\linewidth}{\raggedright #1}
}
\usepackage{tcolorbox}

\usepackage{mathtools}
\usepackage{hyperref}
\hypersetup{
    colorlinks=true,
    linkcolor=blue,
    filecolor=magenta,      
    urlcolor=cyan,
    pdfpagemode=FullScreen
    }
\usepackage{pifont}
\usepackage{listings}
\usepackage{xcolor}
\usepackage[framemethod=TikZ]{mdframed}
\newcounter{theory}

\usepackage{cite}

\usepackage{url}
\urldef{\mailsa}\path|{alfred.hofmann, ursula.barth, ingrid.haas, frank.holzwarth,|
	\urldef{\mailsb}\path|anna.kramer, leonie.kunz, christine.reiss, nicole.sator,|
	\urldef{\mailsc}\path|erika.siebert-cole, peter.strasser, lncs}@springer.com|    
\newcommand{\name}{\texttt{RSMixLeak}\xspace}

\usepackage[T1]{fontenc}
\usepackage[utf8]{inputenc}
\usepackage[english]{babel}
\definecolor{rowbg}{RGB}{238,240,248}   % proposed method row highlight
\definecolor{bestval}{RGB}{192,57,43}   % best value color

\usepackage{amsthm}

\usepackage{soul}
\definecolor{mixbg}{RGB}{235,237,245}  

\usepackage{minibox}
\newcounter{observcntr}

\usepackage{minibox}
\newcounter{takeaway}

\definecolor{R}{RGB}{0,0,150}
\usepackage{xspace}

\theoremstyle{remark}
\usepackage{tikz}
\usepackage{amsmath}
\usepackage{multirow}
\usepackage{threeparttable}
\usepackage{diagbox}
\usepackage{caption}
\usepackage{subcaption}
\usepackage{float}
\usepackage{balance}
\usepackage{graphicx}
\usepackage{textcomp}
\usepackage{booktabs}
\usepackage{colortbl}

\usepackage[absolute]{textpos}
\newcommand{\mypara}[1]{\smallskip\noindent\textbf{#1}}

\definecolor{softgreen}{RGB}{0,128,96}

\begin{document}
\title{When T2I Synthetic Data Backfires: Amplified Privacy Risks in Real-Synthetic Mix Training}
\author{
\IEEEauthorblockN{
Na Li\textsuperscript{1},
Boyu Kuang \textsuperscript{1},
Hongsheng Hu\textsuperscript{2},
Liquan Chen\textsuperscript{3},
Hyoungshick Kim\textsuperscript{4},
Yansong Gao\textsuperscript{5},
Anmin Fu\textsuperscript{1}
}

\IEEEauthorblockA{\textsuperscript{1}Nanjing University of Science and Technology, China}
\IEEEauthorblockA{\textsuperscript{2}The University of Newcastle, Australia}

\IEEEauthorblockA{\textsuperscript{3}Southeast University, China}
\IEEEauthorblockA{\textsuperscript{4}Sungkyunkwan University, Korea}

\IEEEauthorblockA{\textsuperscript{5}The University of Western Australia, Australia}}
\IEEEoverridecommandlockouts
\makeatletter\def\@IEEEpubidpullup{6.5\baselineskip}\makeatother

\maketitle

\begin{abstract}
To overcome data scarcity and privacy constraints in data collection, it has become standard practice across academia and industry to augment real training data with text-to-image (T2I)-generated synthetic data, a paradigm
we term \emph{Real-Synthetic Mix-Training} (RSMT). While substituting synthetic data for sensitive real samples is widely regarded as a means to mitigate privacy exposure of the substituted data, the risk to the remaining real samples that actively participate in training has remained
largely unexamined.

This work reveals, for the first time, that RSMT can substantially amplify privacy leakage of these real training samples. We establish a theoretical framework, \emph{RSMT Memorization Amplification}, proving that
incorporating synthetic data displaces real samples toward peripheral regions of the mixed feature space, in turn forcing the model to memorize them more aggressively. Guided by this foundation, we propose \name to
systematically assess this risk through membership inference attacks (MIAs). \name comprises two variants depending on the adversary's capability. The non-adversarial variant audits a benign RSMT pipeline with
an honest T2I provider, establishing a lower bound on the leakage induced by the intrinsic gap between real and T2I-generated data. The adversarial variant considers an adversary who controls the T2I model or contributes
crafted data to the T2I provider, and deliberately enlarges this distributional gap on a target class via either high-level semantic attribute binding or imperceptible pixel-level coating, further amplifying
leakage on real training data while improving downstream model utility.
Motivated by these findings, we further propose a lightweight leakage propensity indicator computable from real data alone that reliably identifies high-risk datasets \textit{unsuitable} for entering RSMT, as a
self-assessable mitigation.

Extensive experiments across five benchmark datasets spanning facial and satellite imagery, together with four state-of-the-art T2I generators --- two open-source models (SD1.5 and Flux-mini) and two commercial APIs
(Google's Nano Banana and OpenAI's ChatGPT Images 2.0) --- provide strong evidence that RSMT jeopardizes the privacy of real training samples.
\end{abstract}

\IEEEpeerreviewmaketitle

\section{Introduction}
\label{sec:intro}
% ======================================================
The rise of deep learning (DL) has entirely reshaped the landscape of computer vision research over the past decades, with the recent surge in Artificial Intelligence Generated Content (AIGC) propelling a new wave of unprecedented revolution~\cite{Sariyildiz2022FakeIT, he2016deep, stokel2023chatgpt}. 
A well-known challenge is that the performance of most contemporary DL technologies is long-standing bottlenecked by the quality and quantity of natural image data. 
%However, 
Manually collecting and annotating large, diverse, and fair datasets is inherently time-consuming and labor-intensive, and imposes practical burdens for storage and transfer~\cite{azizi2023synthetic, yuan2024realfake, singh2024synthetic, fan2024scaling}.
More critically, across a wide range of real-world scenarios, obtaining sufficient training data itself often poses a nontrivial barrier.

Fortunately, the latest breakthroughs in text-to-image (T2I) generation have flourished a large and continuously increasing number of open-source T2I models (e.g., Stable Diffusion (SD)~\cite{rombach2022high}  and Flux~\cite{labs2025flux1kontextflowmatching}) and commercial APIs (e.g., Google's Nano Banana~\cite{comanici2025gemini} and OpenAI's latest ChatGPT Images 2.0~\cite{chatgptimages2026}), exhibiting stunning image generation capabilities and inspiring a surge in leveraging them for data augmentation~\cite{azizi2023synthetic}.
In particular, by using text prompts to control generated image content precisely, this approach can synthesize task-specific datasets that benefit low-resource or privacy-compliant areas (e.g., facial recognition, satellite imagery) and produce categorical labels without human labor. 
Building on these benefits, synthetic data has been rapidly and widely adopted across both academia and industry as a data augmentation tool, typically blended with real training data for mix-training downstream models, a paradigm we term \emph{Real-Synthetic Mix-Training (RSMT)}.
For instance, Google Research reports that augmenting real datasets with ImageNet-style images generated by Imagen~\cite{saharia2022photorealistic} can boost the classification accuracy of ResNet and Transformer~\cite{azizi2023synthetic}. In practice, NVIDIA incorporates synthetic trajectory data produced by Cosmos World Foundation Models~\cite{agarwal2025cosmos} with real-world data to enhance the performance of Isaac GROOT N1.5.

\subsection{Limitation}
Beyond utility gains, synthetic data has also been posited as a privacy safeguard that replaces sensitive real samples unable to enter training due to privacy or licensing constraints and thereby protects them from exposure~\cite{yuan2024realfake, zhao25does}. However, the real data that ultimately enters the training pipeline may still carry sensitive information, and the specific privacy risks that synthetic data poses to these remaining real training samples are overlooked to date and poorly understood.

To quantify privacy risks, MIAs have emerged as the standard auditing tool~\cite{shokri2017membership, carlini2022membership, nasr2019comprehensive, yeom2018privacy}. By leveraging the model’s tendency to memorize its training data, MIAs identify subtle discrepancies in outputs between members and non-members, enabling attackers to infer whether a specific sample was part of the training data, thereby posing a substantial threat to DL models~\cite{WangZCBKY25, chen2020gan,li25enhanced,chen2025method}.
For instance, membership inference against a melanoma classifier trained on dermatoscopic images can effectively disclose a patient's sensitive information, revealing that the individual has been diagnosed with the corresponding condition~\cite{ziller2024reconciling}.
On the other hand, MIAs can also serve as an auditing tool for assessing the extent of privacy leakage in deployed models.

Notably, state-of-the-art T2I generative models unavoidably embed measurable low-level artifacts into synthetic images~\cite{sander2024watermarking, karageorgiou2025any, ojha2023towards}, inducing an irreducible distributional gap between synthetic and real data. When synthetic samples dominate the training set, this gap renders the remaining real samples atypical relative to the dominant synthetic distribution, making them hard examples that compel the model to memorize them more aggressively to minimize the training loss. 
Although prior studies~\cite{yuan2024realfake, zhao25does} have noted the privacy-preserving potential of T2I-generated synthetic data via MIAs, they focus exclusively on scenarios where downstream models are trained on pure synthetic datasets serving as proxies for real samples excluded from training, not RSMT.
The emerging RSMT has not been looked at yet, where the distributional discrepancy between real and synthetic data, may instead expose the real samples that actively participate in training.

We hence ask the following research questions to underscore the urgent need for a comprehensive assessment of the privacy risks arising from the emerging RSMT scenario.

\begin{mdframed}[backgroundcolor=black!10,rightline=false,leftline=false,topline=false,bottomline=false,roundcorner=2mm]
Does the incorporation of T2I-generated synthetic data amplify the privacy leakage of the real training samples that actively participate in training? If so, to what extent is the privacy risk exacerbated?
\end{mdframed}

\subsection{Our Work}
This work, for the first time, unveils and proves that incorporating T2I-generated synthetic data exacerbates the privacy leakage of real samples that actively participate in training.
Firstly, we establish a theoretical framework, \emph{RSMT Memorization Amplification}, proving that the intrinsic distributional gap between real and T2I-generated data pushes real samples toward peripheral regions of the mixed feature space, turning them into an atypical subpopulation that the model is forced to memorize more aggressively.
Guided by this theoretical foundation, we propose \name, a systematic framework that quantifies the privacy leakage of real training samples through MIAs and shows that the intrinsic gap alone suffices to amplify leakage, which an adversary can further enlarge by manipulating the upstream T2I pipeline.
Below, we brief the key findings of \name and its core design.

\mypara{Theoretical Foundation.}
We formalize RSMT Memorization Amplification as a unified theoretical framework, establishing that RSMT provably amplifies the privacy risk of real training samples over real-only training. First, we prove that the irreducible real–synthetic distributional gap displaces real samples toward peripheral regions of the mixed feature space and relegates them to a tail subpopulation that is atypical relative to the dominant synthetic core (Theorem~\ref{thm:atypicality}). 
Building on this, we show that these atypical real samples cannot be captured by shared statistical patterns and thus rely on substantially amplified memorization to be fitted (Theorem~\ref{thm:mem-amp}). Both theorems are substantiated by corresponding empirical evidence.

\mypara{\name.}
The above theorems lie the foundations of \name, a framework that systematically audits RSMT-induced privacy leakage to real training samples through two progressive variants differentiated by the adversary's capability to manipulate the upstream T2I pipeline.

\noindent \textit{Non-adversarial Audit.} We first audit a benign non-adversarial RSMT pipeline with an honest T2I provider, where the adversary only has black-box query access to the downstream victim model.
Specifically, we collect synthetic data from open-source backbones, i.e., pre-trained SD1.5 and Flux-mini (a lightweight TencentARC-released substitute for FLUX.1-dev)\footnote{Flux-mini is a 3.2B distilled variant of FLUX.1-dev released by TencentARC at \url{https://github.com/TencentARC/FluxKits}}, plus their LoRA fine-tuned variants for specialized domains and commercial APIs, i.e., Google's Nano Banana and OpenAI's ChatGPT Images 2.0, mix it with real data at a 1:4 real-to-synthetic ratio, and train the victim classifier.
Four representative MIA methods spanning training-based ~\cite{shokri2017membership,yuan2022membership} and metric-based methods~\cite{song2019privacy,ZarifzadehLS24} then quantify the membership leakage. Our results reveal that the intrinsic gap alone substantially amplifies leakage. For example, when applying RMIA~\cite{ZarifzadehLS24} on synthetic data generated by pre-trained SD1.5 for ImageNet100, the TPR@0.1\% FPR increases from 12\% to 21.2\%.

\noindent \textit{Adversarial Audit.} Adversarial \name considers an adversary who operates upstream as a malicious T2I provider or as a contributor injecting crafted data into the T2I training set. 
Since memorization amplification grows with the real–synthetic distributional gap $\delta$, \name
deliberately enlarges $\delta$ on a chosen victim target class through two schemes.
\textit{Semantic attribute binding scheme} fine-tunes the T2I model on bias images that consistently embed a rare but visually natural attribute underrepresented in the real distribution, shifting the synthetic distribution at the semantic level. \textit{Pixel-level coating scheme} applies imperceptible perturbations directionally optimized away from the real-class centroid; the T2I model absorbs these perturbations as class-intrinsic features and reproduces the induced shift in its generated samples. 
Both adversarial schemes further amplify privacy leakage beyond the non-adversarial \name baseline while improving downstream classification utility over real-only training. For instance, on SD1.5/VGGFace2, semantic attribute binding further lifts TPR@0.1\% FPR from 4.0\% to 26.0\% and AUC from 89.2\% to 92.8\%, while pixel-level coating achieves 12.0\% TPR@0.1\% FPR and 93.4\% AUC.

Our main contributions are summarized as follows:

\noindent$\bullet$ We establish the theoretical foundation that RSMT amplifies the memorization of real training samples, as the large influx of synthetic data renders them an atypical tail subpopulation that the model is compelled to memorize more aggressively.

\noindent$\bullet$ We propose \name, the first framework for systematically auditing the privacy leakage of the RSMT paradigm through the lens of MIAs, showing that benign RSMT alone is sufficient to jeopardize the privacy of real training samples.

\noindent$\bullet$ We further show that an adversary can substantially amplify leakage via enlarging the distributional gap $\delta$ between real and synthetic data on a chosen target class, through semantic or pixel-level manipulation on the upstream T2I pipeline.

\noindent$\bullet$ We propose a lightweight leakage propensity indicator computable from real data alone, which reliably identifies high-risk datasets \textit{unsuitable} for RSMT before synthetic data generation or downstream training, as an \textit{avoidance countermeasure}.

\noindent$\bullet$ We conduct extensive experiments across diverse datasets, T2I backbones, and commercial APIs that validate the effectiveness of \name. We also provide practical mitigation to \name e.g., by computing an indicator upon the real samples \textit{only}.

% ======================================================
\section{Background and Related Work}
\label{sec:related}
% ======================================================
\subsection{T2I Generative Models} T2I generation has witnessed remarkable progress in the past few years, with the emergence of numerous accessible open-source and commercial models capable of producing high-fidelity, photo-realistic images. Specifically, early mainstream T2I models predominantly relied on the UNet-based architecture (e.g., SD1.5~\cite{rombach2022high}, SDXL~\cite{podell2024sdxl}, Imagen~\cite{saharia2022photorealistic}, etc.), enabling image generation by conducting text-conditioned diffusion within a latent space, which has been widely adopted in both academic research and industrial applications.
More recently, Diffusion transformers (DiT)~\cite{peebles2023scalable} burst onto the scene as a notable evolution (i.e., Flux~\cite{labs2025flux1kontextflowmatching}, Sana~\cite{xie2025sana}), replacing the UNet backbone with vision transformers, supporting high-quality image generation.

\subsection{Synthetic Data Generation}
For general-domain tasks such as synthesizing ImageNet-like data, large-scale pre-trained T2I models can reliably generate high-quality samples owing to their extensive training on massive and broadly representative image–text corpora, whereas domain-specific tasks (e.g., satellite imagery) can be effectively supported by applying parameter-efficient fine-tuning (PEFT) methods (e.g., LoRA~\cite{hu2022lora}) to fine-tune the models effectively using a small amount of real data from the target domain~\cite{singh2024synthetic}. 
Moreover, an increasing number of companies, such as Synthesis AI~\cite{synthesisai} (serving clients like Amazon, Apple, Google, and Intel), offer commercial synthetic data services, enabling the creation of customized enterprise synthetic datasets.

\subsection{Learning with Real and T2I-Generated Synthetic Data}
\label{sec:learn}
Synthetic data has been extensively utilized across various computer vision tasks, including robotics~\cite{moreau2022lens, yen2022nerf}, autonomous driving~\cite{abu2018augmented}, and object detection~\cite{abu2018augmented,peng2015learning}.
An earlier line of work explores synthetic data via dataset condensation and distillation, which compresses real training sets into smaller synthetic surrogates of comparable utility, reducing storage and accelerating training~\cite{wang2018dataset,zhaodataset,cazenavette2022dataset}. 
While Dong \textit{et al.}~\cite{dong2022privacy} claim that condensed data inherently resist membership inference since no real sample enters training, Carlini \textit{et al.}~\cite{carlini2022no} refute this by showing the reported gain becomes statistically insignificant against a properly measured baseline. Liu \textit{et al.}~\cite{liu2023backdoor} further exposed backdoor vulnerabilities in distilled datasets.
However, growing attention has shifted to T2I-generated synthetic data, which is also the focus of our work. 
Primarily, existing research can be grouped into two categories: one line of work examines the use of purely synthetic data to evaluate its standalone effectiveness, while the other leverages synthetic data as an augmentation source, integrating it with real data to assess the improvement in performance.

Existing studies investigate training solely on T2I generated synthetic data from multiple perspectives, e.g., model performance, robustness, and scalability. Although purely synthetic training typically falls short of matching real data performance, Sariyildiz \textit{et al.}~\cite{Sariyildiz2022FakeIT} first demonstrate that carefully leveraging prompt engineering can substantially narrow the performance gap between models trained on T2I-generated synthetic versus real images.
Singh \textit{et al.}~\cite{singh2024synthetic} show that self-supervised and multi-modal models trained exclusively on T2I-generated synthetic data can match real image baselines across several robustness metrics.

Cutting-edge works highlight that when T2I-generated synthetic data works collaboratively with real data for model training, a paradigm we term Real-Synthetic Mix Training (RSMT), it is not only more common and effective in practice but also the focus of our work. RSMT can enhance performance across a variety of tasks, particularly in data-scarce scenarios. For example, studies~\cite{he2023is, fan2024scaling,azizi2023synthetic} demonstrate that synthetic images generated by GLIDE, SD1.5, and Imagen, respectively, provide measurable gains when jointly used with real data for model training. Singh \textit{et al.}~\cite{singh2024synthetic} further report that the RSMT paradigm yields models that are less susceptible to both adversarial perturbations and natural distributional noise. In industry, NVIDIA uses synthetic trajectory data generated by WFMs~\cite{agarwal2025cosmos}, co-training it with real-world data to improve the performance of Isaac GROOT N1.5~\cite{bjorck2025gr00t}, an open foundation model for humanoid robot reasoning and skills.

\subsection{Membership Inference Attacks} MIAs are widely adopted to assess the privacy risks of DL models by determining whether a given sample is part of the model’s training set~\cite{li2021membership,carlini2022membership,Wang2025rigging,ye2022enhanced,shang2025defend,HuiYYBGC21,zhang25soft,stevanoski2024querycheetah,chen2024slmia}.
Formally, for a target sample $x$ and a victim model $\mathcal{M}$, the membership inference algorithm $\mathcal{A}$ is defined as:
\begin{equation}
\mathcal{A}: (x, \mathcal{M}) \rightarrow {0,1},
\end{equation}
where the output indicates whether $x$ is a member (1) or a non-member (0) of the training dataset of $\mathcal{M}$.

In practical scenarios, MIAs are typically conducted under the black-box assumption, where the adversary can only access the model’s output probabilities (i.e., confidence scores). Existing approaches can be broadly divided into training-based and metric-based methods, depending on whether an attack meta-classifier is trained. Training-based methods first construct shadow models to approximate the behavior of the victim model, then generate labeled attack metadata to train a binary meta-classifier for membership inference. For instance, Shokri \textit{et al.}~\cite{shokri2017membership} use prediction confidence vectors as attack metadata, while Yuan \textit{et al.}~\cite{yuan2022membership} further incorporate sensitivity, label and posterior to train a transformer-based attack model. In contrast, metric-based approaches avoid training the attack meta-classifier and instead rely on statistical indicators derived from the victim model, i.e., loss~\cite{yeom2018privacy,song2019privacy}, posterior~\cite{yeom2018privacy}, likelihood ratios~\cite{carlini2022membership}, and scaled confidence score~\cite{du2025imitative}. More recently, Du \textit{et al.}~\cite{du2026cascading} improved attack effectiveness by exploiting membership dependencies.

Furthermore, MIAs have also been employed as audit tools to analyze privacy risks in various scenarios, including explainable machine learning~\cite{liu2024please}, visual encoders~\cite{zhu2024unified,liu2021encodermi}, machine unlearning~\cite{chen2021machine}, model compression~\cite{li2025compleak} and RAG datastore~\cite{NasehPSCOH25,GaoMD0G25}.
Recent efforts further extend MIAs to foundation models, including VLMs~\cite{hu2025vlms}, LLMs~\cite{meeus2024did,wen2024membership,he2025llms, tong2025membership,chen2026window}, text-to-video models~\cite{wang2026vidleaks} and diffusion models~\cite{Peng2025diffence,pang2025black,wang2026inference}.

To the best of our knowledge, no prior work has investigated the privacy leakage of real training samples in the RSMT pipeline, where T2I-generated synthetic data is jointly used with real data to train downstream models. Existing efforts treat synthetic data as a proxy for real samples, either to circumvent the prohibitive cost of curating real datasets~\cite{yuan2024realfake,singh2024synthetic} or to substitute for real samples excluded from training and thereby shield them from exposure~\cite{zhao25does}.

In contrast, our work focuses on the prevailing RSMT paradigm, in which T2I-generated synthetic data augments real training data to improve downstream model performance, particularly when real data is scarce. Within this paradigm, we investigate how---and to what extent---the incorporation of (non-adversarial/adversarial) T2I-generated synthetic data affects the privacy leakage of the real samples that actively participate in training.
% ======================================================

\section{Theoretical Foundations: RSMT Memorization Amplification}
\label{sec:theory}

While the RSMT paradigm has been adopted for data augmentation, we posit that it inevitably amplifies the memorization of sensitive real training samples, thereby exacerbating the  privacy leakage of the latter. To formalize this, we establish a theoretical framework for \emph{RSMT Memorization Amplification}, which comprises two progressive claims, each supported by empirical evidence.
%We substantiate this claim through two progressive theoretical propositions, each supported by empirical evidence. 
First, we prove that incorporating a volume of synthetic data pushes real samples toward peripheral regions of the mixed feature space, rendering them relatively atypical (Theorem~\ref{thm:atypicality}). 
Building on this, we demonstrate that such atypical real samples---those from the tail subpopulation---exhibit amplified memorization under mixed-training (Theorem~\ref{thm:mem-amp}).

\begin{tcolorbox}[colback=black!10, colframe=black!50, boxrule=0pt, arc=3pt]
\subsubsection*{\textbf{Key Insight~1}}
\textit{Incorporating a volume of synthetic data into the training set can inadvertently shift the remaining real samples toward off-centroid regions of the mixed feature space, implicitly relegating them to a tail subpopulation and being relatively atypical.}
\end{tcolorbox}

To formalize this claim, we begin by defining class-wise feature centroids under mixed data and then characterize how real samples are displaced toward peripheral regions, thereby forming a tail subpopulation.

\begin{definition}[Feature Centroid of Mixed Data]
\label{def:centroid}
Let $\mathcal{D}_{\mathrm{real}}$ and $\mathcal{D}_{\mathrm{syn}}$ denote the real and synthetic training sets, respectively, with $|\mathcal{D}_{\mathrm{real}}| = N$ and mixing ratio $\lambda = |\mathcal{D}_{\mathrm{syn}}|/|\mathcal{D}_{\mathrm{real}}|>1$.
Denote by $\phi:\mathcal{X}\to\mathbb{R}^d$ a fixed feature extractor (e.g., a pre-trained CLIP encoder). For each class $c$, let
$\boldsymbol{\mu}_c^{\mathrm{real}},
 \boldsymbol{\mu}_c^{\mathrm{syn}}\in\mathbb{R}^d$
be the per-class centroids of the real and synthetic samples, respectively, and let
\[
  \boldsymbol{\mu}_c^{\mathrm{mix}}
  = \frac{1}{1+\lambda}\,\boldsymbol{\mu}_c^{\mathrm{real}}
  + \frac{\lambda}{1+\lambda}\,\boldsymbol{\mu}_c^{\mathrm{syn}}
\]
be the mixed centroid.
\end{definition}

\begin{theorem}[RSMT Push Real Samples to Peripheral Regions of the Mixed Feature Space]
\label{thm:atypicality}
Under Definition~\ref{def:centroid}, assume the per-class
feature distributions satisfy
$\phi(\mathcal{D}_c^{\mathrm{real}})\sim
 \mathcal{N}(\boldsymbol{\mu}_c^{\mathrm{real}},\Sigma_r)$
and
$\phi(\mathcal{D}_c^{\mathrm{syn}})\sim
 \mathcal{N}(\boldsymbol{\mu}_c^{\mathrm{syn}},\Sigma_s)$.
Suppose there exists a constant $\delta>0$ such that
\begin{equation}
  \label{eq:gap}
  \bigl\|\boldsymbol{\mu}_c^{\mathrm{real}}
        - \boldsymbol{\mu}_c^{\mathrm{syn}}\bigr\|_2
  \;\geq\; \delta
\end{equation}
reflecting the irreducible distributional gap induced by T2I
generative artifacts.
Then, for any $\mathbf{x}\in\mathcal{D}_c^{\mathrm{real}}$:
\begin{equation}
  \label{eq:atypicality}
  \mathbb{E}\!\left[\bigl\|\phi(\mathbf{x})
    - \boldsymbol{\mu}_c^{\mathrm{mix}}\bigr\|_2^2\right]
 \;>\;
  \mathbb{E}\!\left[\bigl\|\phi(\mathbf{x})
    - \boldsymbol{\mu}_c^{\mathrm{real}}\bigr\|_2^2\right]
  + \left(\frac{\lambda}{1+\lambda}\right)^{\!2}\delta^2
\end{equation}
\end{theorem}

\begin{figure}[t]
    \centering

    \begin{subfigure}[b]{0.15\linewidth}
        \centering
        \includegraphics[width=\linewidth]{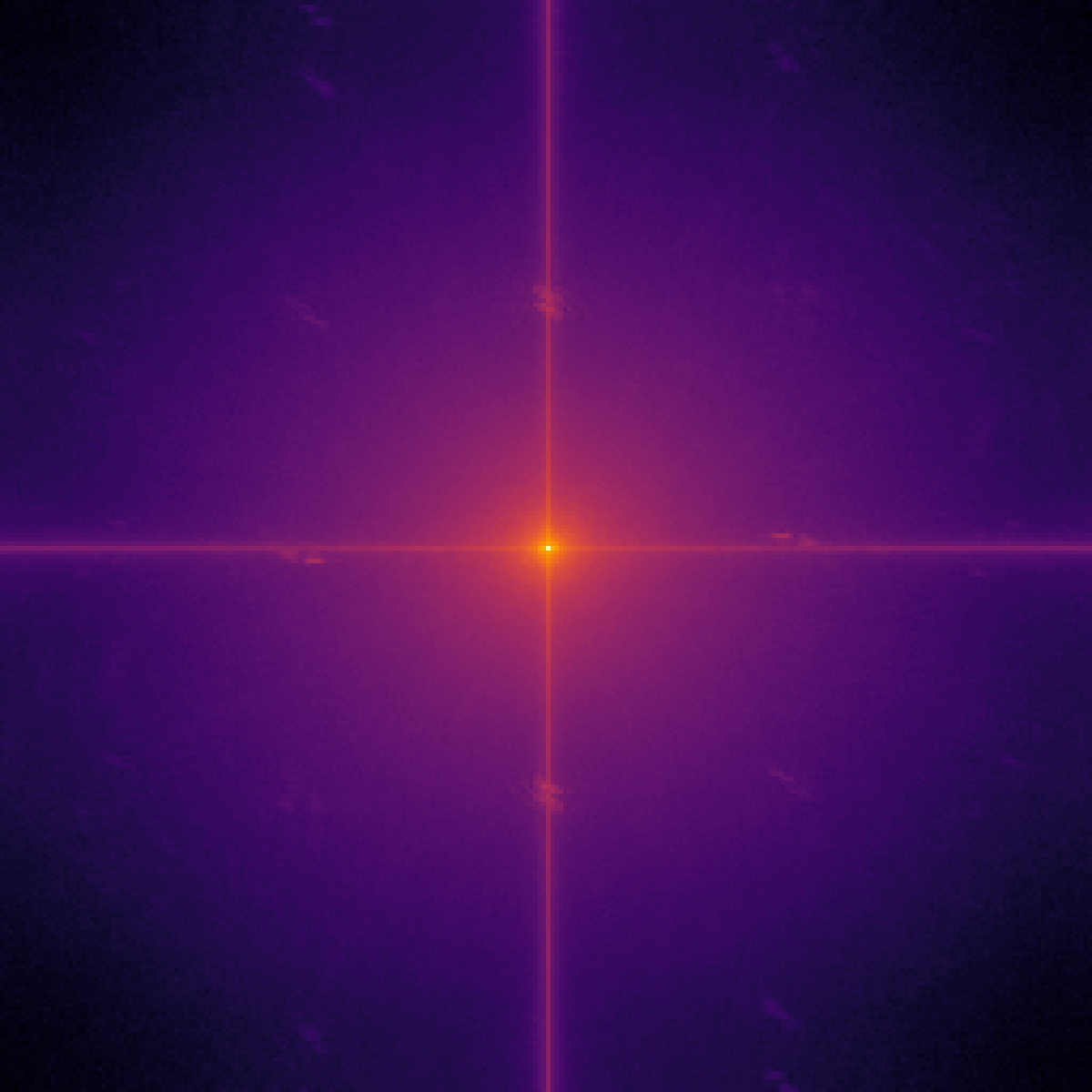}
        \caption{\small Real}
        \label{fig:a}
    \end{subfigure}
    \hfill
    \vspace{4pt}
    \begin{subfigure}[b]{0.15\linewidth}
        \centering
        \includegraphics[width=\linewidth]{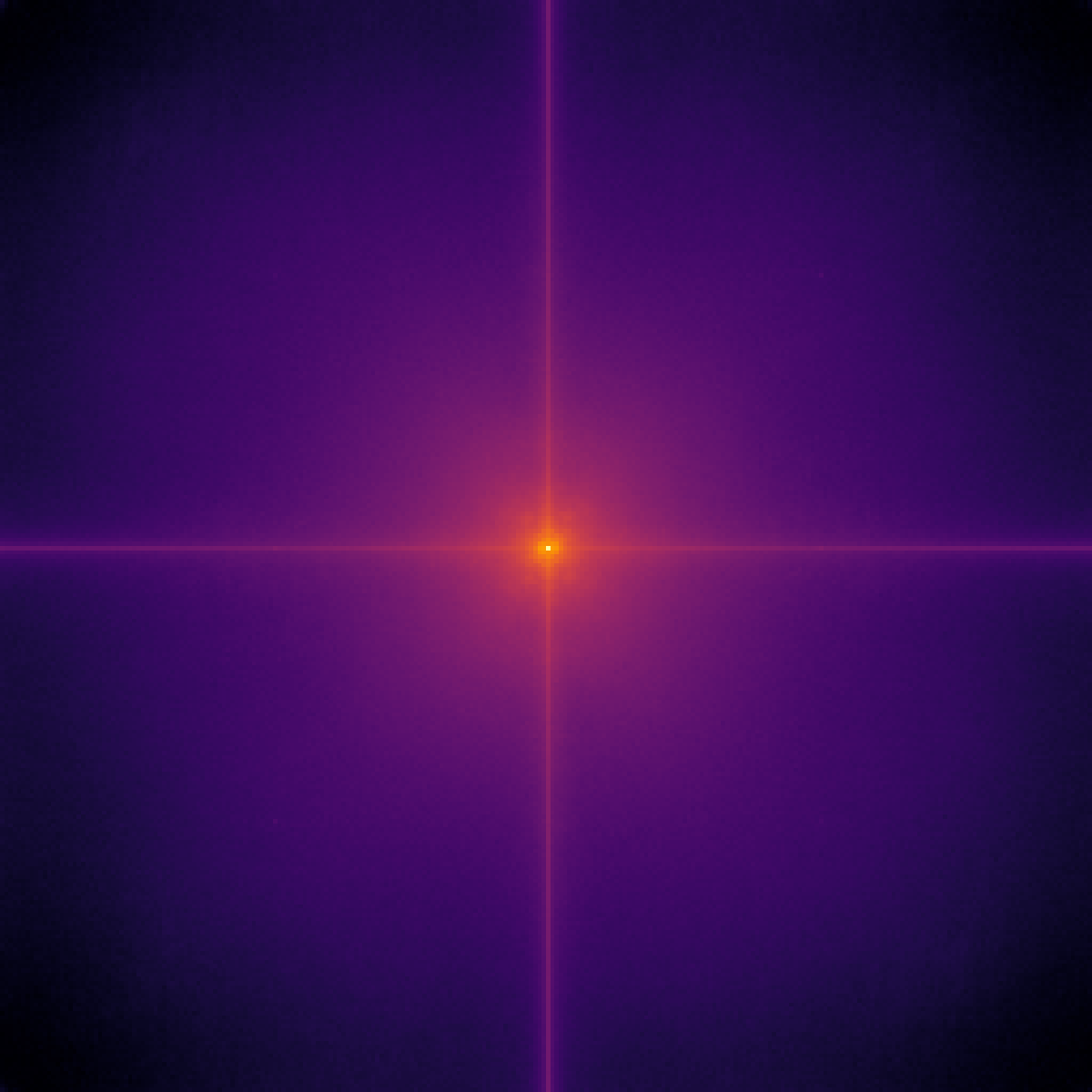}
        \caption{\small SD}
        \label{fig:b}
    \end{subfigure}
    \hfill
    \begin{subfigure}[b]{0.15\linewidth}
        \centering
        \includegraphics[width=\linewidth]{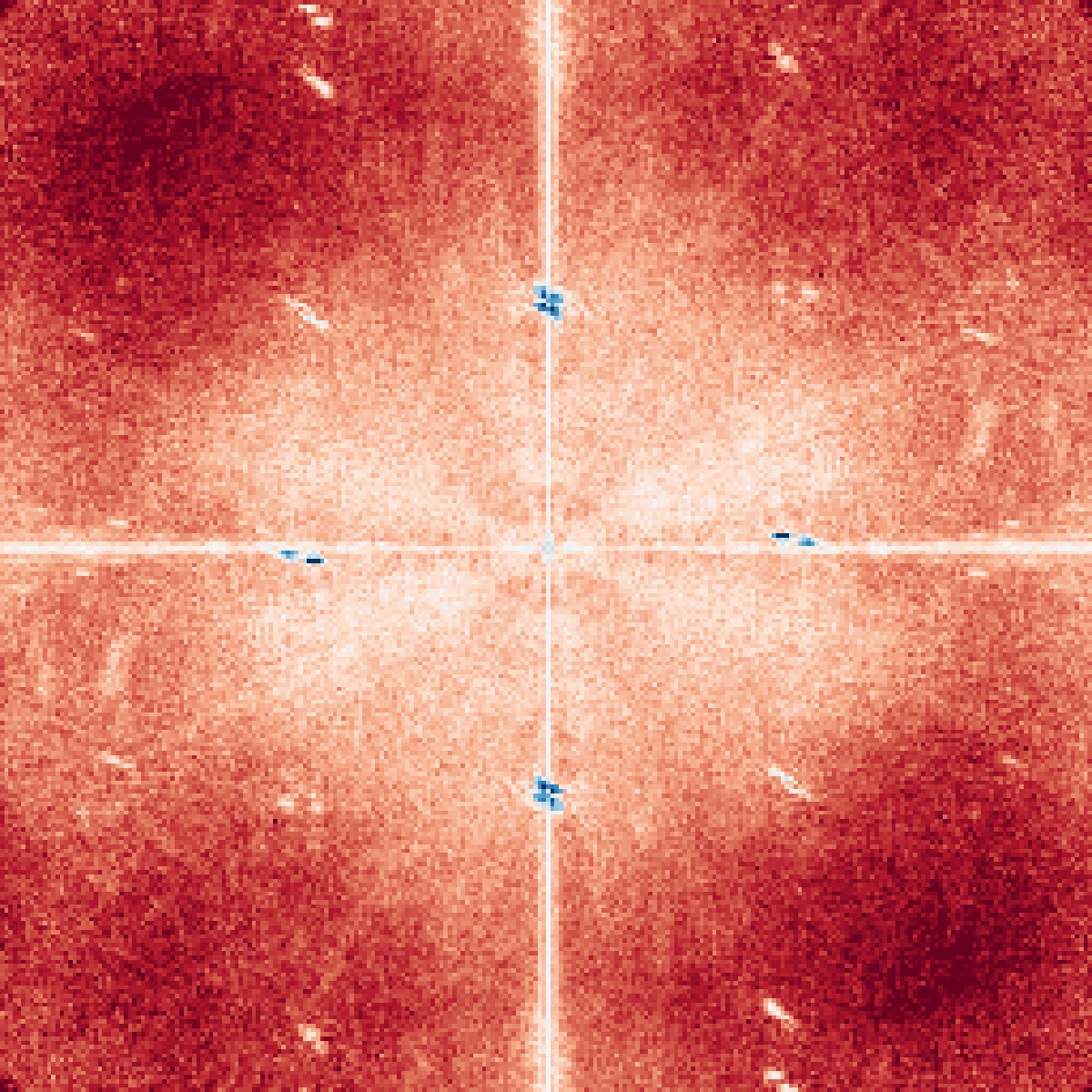}
        \caption{\small $\Delta$SD}
        \label{fig:c}
    \end{subfigure}
    \vspace{6pt}
    \hfill
    \begin{subfigure}[b]{0.15\linewidth}
        \centering
        \includegraphics[width=\linewidth]{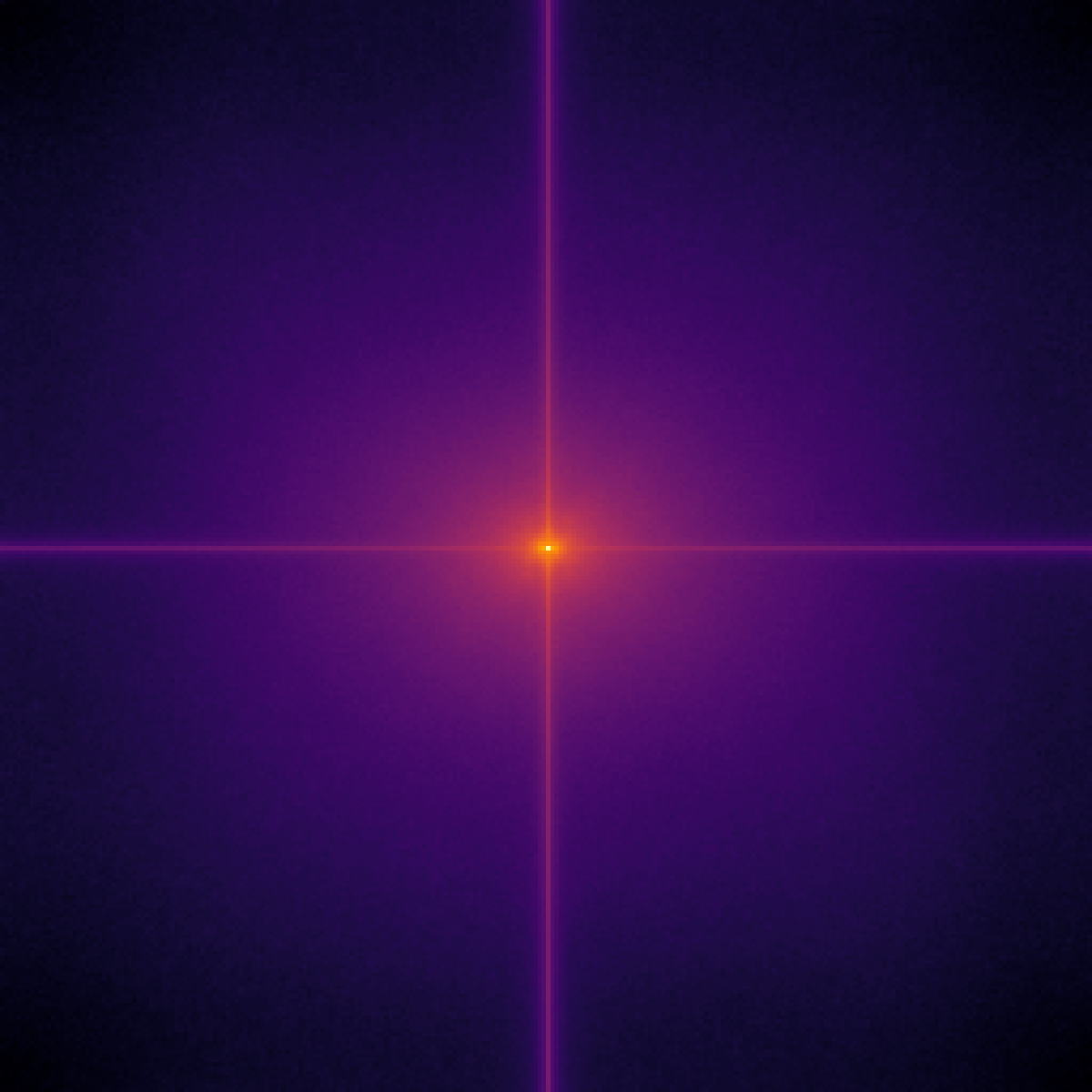}
        \caption{\small Flux}
        \label{fig:d}
    \end{subfigure}
    \hfill
    \begin{subfigure}[b]{0.15\linewidth}
        \centering
        \includegraphics[width=\linewidth]{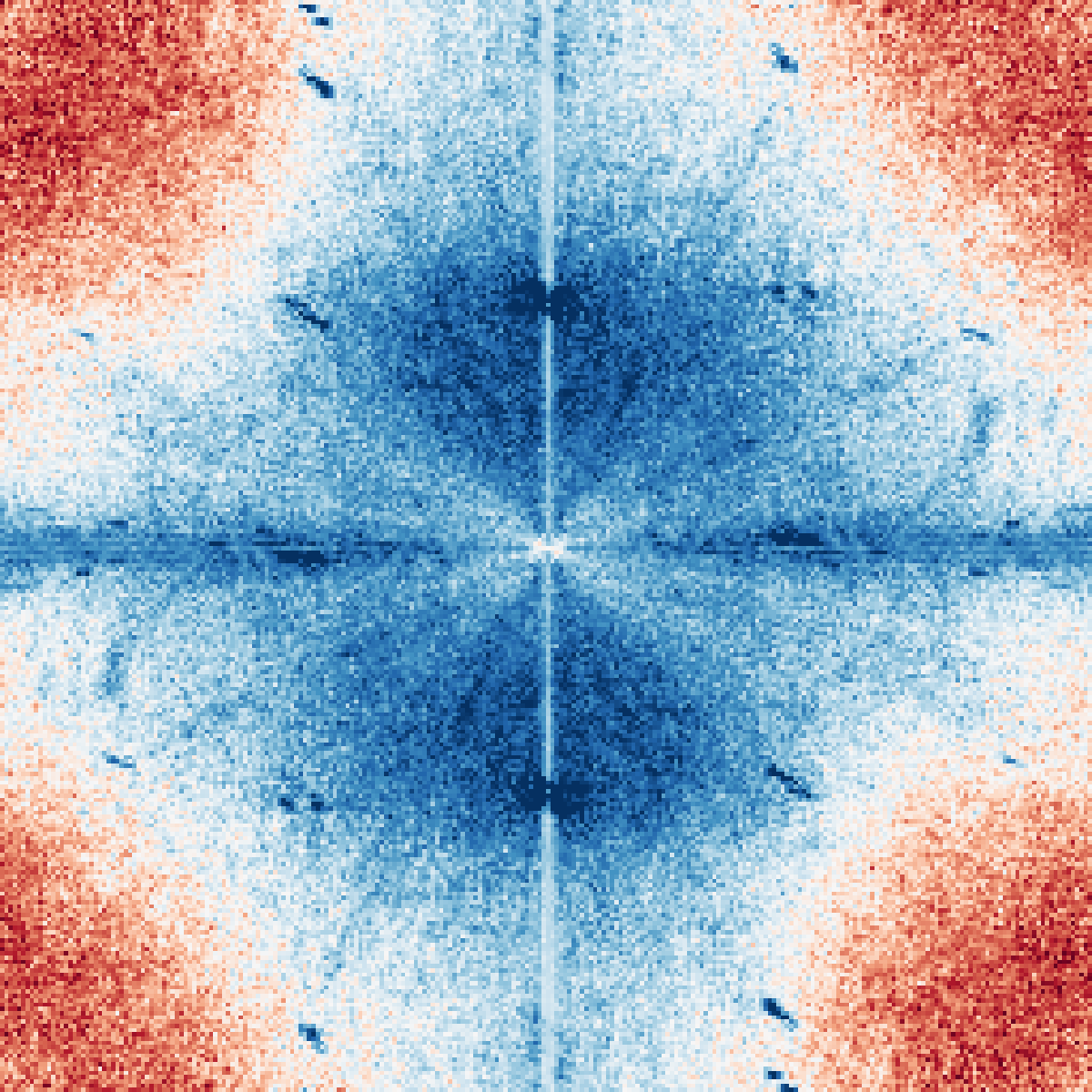}
        \caption{\small $\Delta$Flux}
        \label{fig:e}
    \end{subfigure}

    \caption{Average 2D log power spectra of real and T2I-generated images on the celebrity face dataset, along with their spectral differences. (a)(b): Average power spectra of real, SD1.5 and Flux-mini images, where brighter regions indicate higher spectral energy. (c): Spectral differences relative to real data, where red denotes higher energy in synthetic data, blue denotes lower, and white indicates no difference. 
    }
    \label{fig:spectrum}
\end{figure}
\noindent \emph{Proof.} See Appendix~\ref{app:atypicality} for details.
This phenomenon arises from a fundamental difference in how the two data sources are constructed. On one side, it is a well-established fact that state-of-the-art 
T2I generative models unavoidably embed measurable frequency and spectral artifacts into synthetic images~\cite{sander2024watermarking, karageorgiou2025any, 
ojha2023towards}. On the other side, since T2I-synthetic samples are generated from latent noise $\mathbf{z}$ drawn from a standard distribution (e.g., Gaussian), these artifacts manifest as a consistent, predictable pattern shared across synthetic images, causing them to cluster tightly in feature space. 
However, real samples exhibit non-standard, instance-specific artifacts that no generator trained on a standard latent prior can reproduce, e.g., atypical poses and lighting irregularities, which disperse them across an irregular, sparsely populated region of the feature space.
The two effects compound into an irreducible gap that prevents the 
synthetic distribution from \textit{faithfully enveloping} the real one. 
When real samples constitute a \emph{minority} within the mixed training set, the combined effect of distributional gap and numerical imbalance pushes them towards the periphery of the mixed feature space (Theorem~\ref{thm:atypicality}), placing them in the tail of the mixed distribution that are atypical relative to the dominant synthetic core.

\mypara{Empirical Evidence:} 
We first compute the average 2D log power spectra of real and T2I-synthetic images together with their residuals (Figure~\ref{fig:spectrum}). The contrast is visible at the individual spectra. The real spectrum (Figure~\ref{fig:spectrum}(a)) shows scattered off-axis components and a granular texture reflecting sample-specific irregularities, whereas the synthetic spectra (Figure~\ref{fig:spectrum}(b),(d)) dominated by axis-aligned bands, reflecting a consistent generator-specific artifacts. The residual maps (Figure~\ref{fig:spectrum}(c),(e)) further confirm that this gap is non-vanishing across spatial frequencies.
We further quantify this gap by measuring the per-class CLIP centroid distance between two data sources (Figure~\ref{fig:gap}(a)), which remains strictly positive across all classes.
To verify the downstream effect of this gap as established by Theorem~\ref{thm:atypicality}, we visualize CLIP feature representations via t-SNE~\cite{van2008visualizing} (Figure~\ref{fig:gap}(b)--(c)), comparing RSMT against a real-only baseline that adds the same number of additional real samples as the synthetic portion. In the real-only baseline (b), the original and additional real samples occupy the same regions within each cluster, with the two groups indistinguishable in their spatial distribution. Under RSMT (c), however, real samples are consistently displaced toward cluster peripheries as a tail subpopulation, while synthetic samples dominate the dense cores.

\begin{figure}[t]
    \centering
    \begin{subfigure}[b]{0.32\linewidth}
        \centering
        \includegraphics[width=\linewidth]{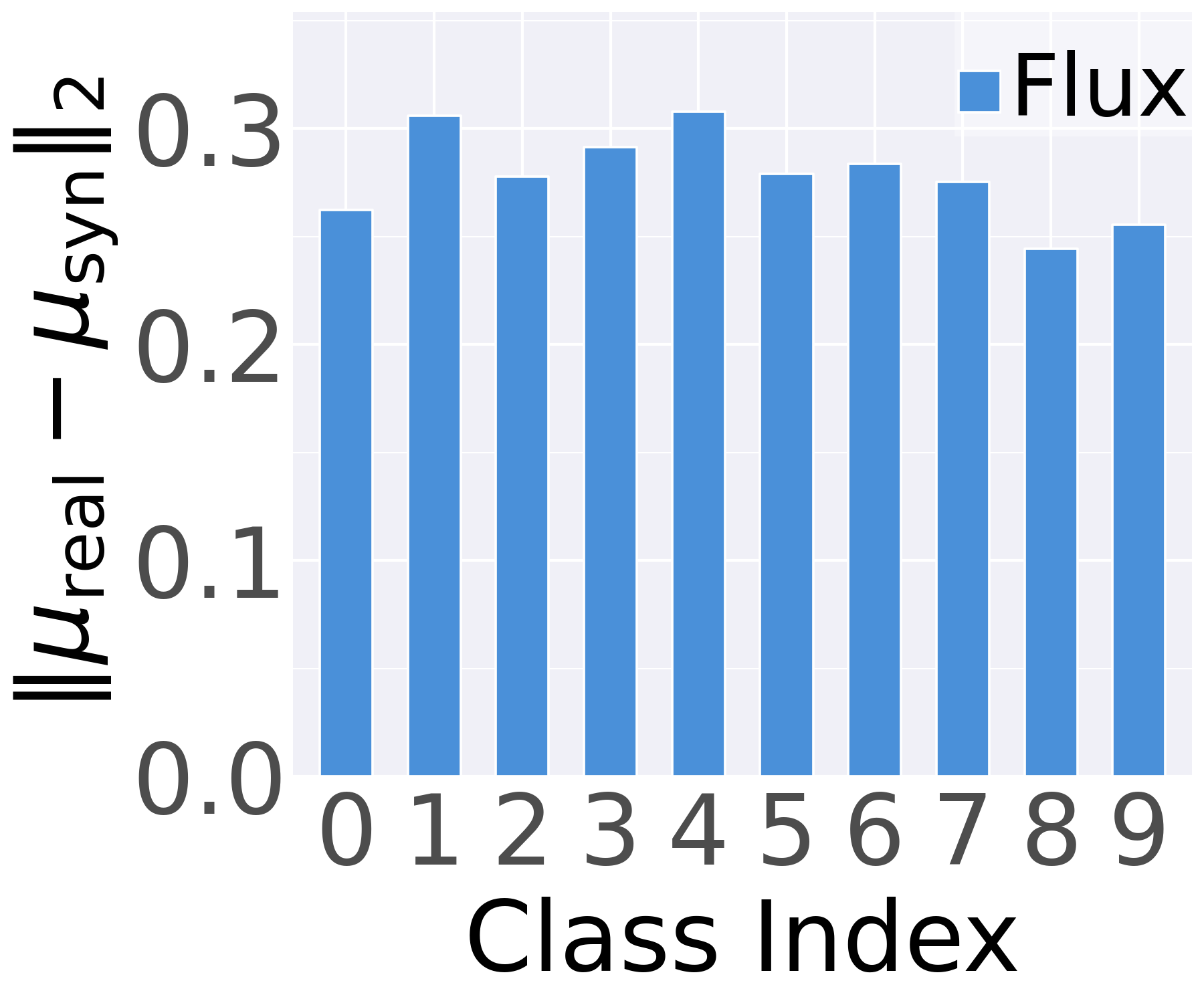}
        \caption{\small Centroid gap}
        \label{fig:centroid}
    \end{subfigure}
    \hfill
    \begin{subfigure}[b]{0.32\linewidth}
        \centering
        \includegraphics[width=\linewidth]{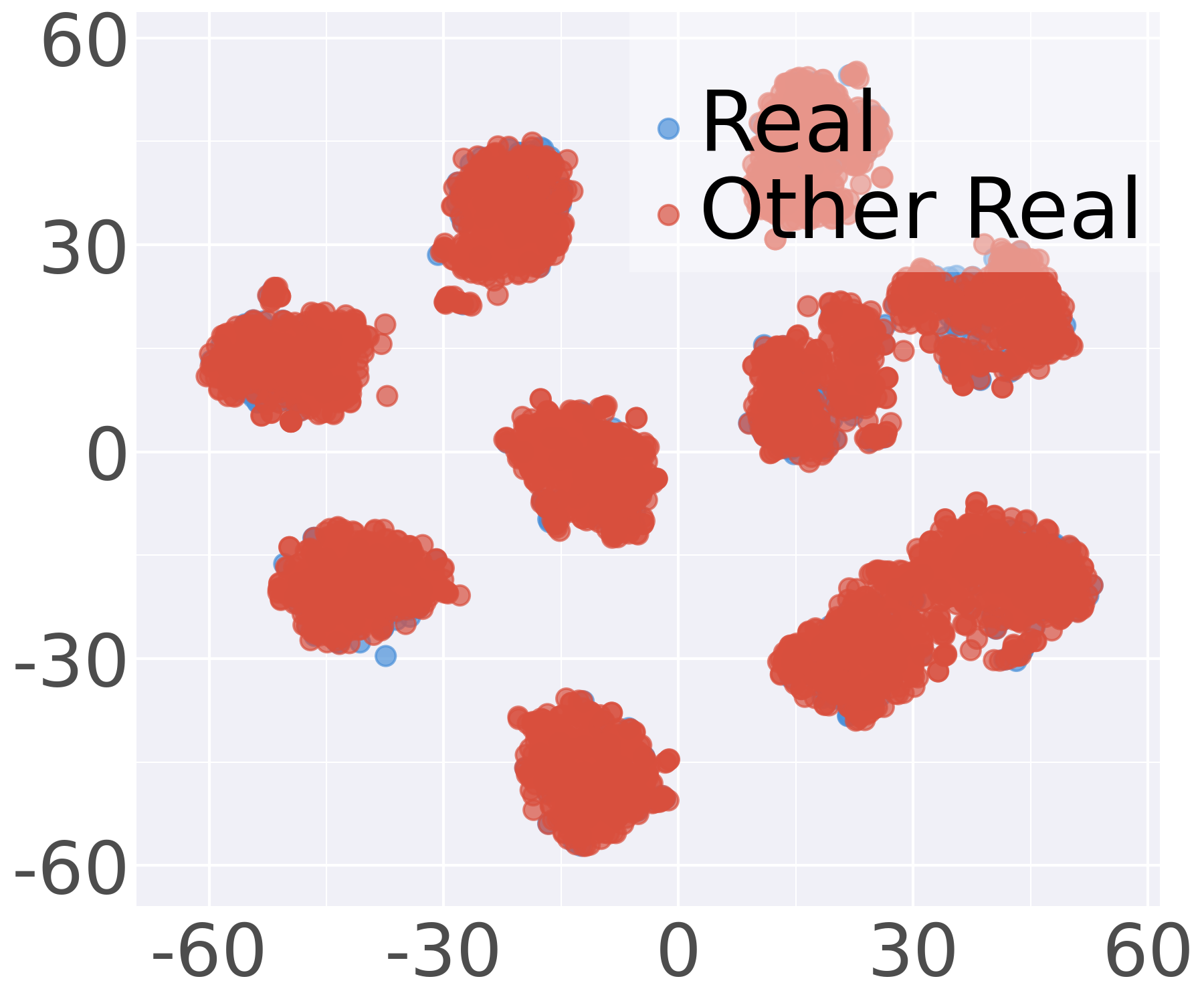}
        \caption{\small Real only baseline}
        \label{fig:tsne_b}
    \end{subfigure}
    %\vspace{6pt}
    \hfill
    \begin{subfigure}[b]{0.32\linewidth}
        \centering
        \includegraphics[width=\linewidth]{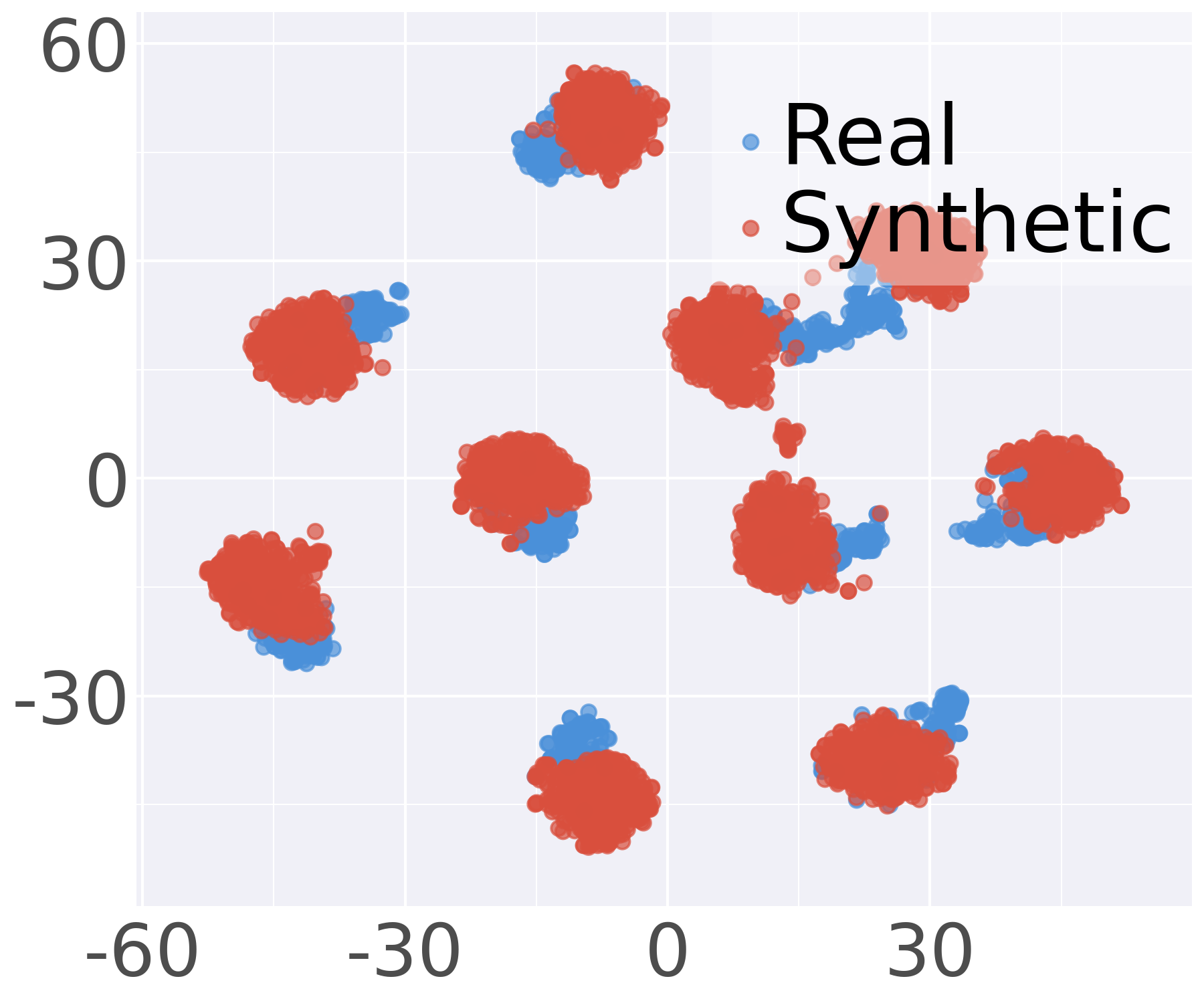}
        \caption{\small RSMT}
        \label{fig:tsne_c}
    \end{subfigure}

    \caption{
    Empirical validation of the distributional gap and peripheral displacement. (a) Per-class CLIP feature centroid distance $\|\boldsymbol{\mu}_{c}^{\text{real}} - \boldsymbol{\mu}_{c}^{\text{syn}}\|_2$ between real and Flux-mini generated synthetic data. (b)–(c) t-SNE of CLIP features for 100 original real samples augmented with: (b) 400 additional real samples, and (c) 400 Flux-mini generated synthetic samples.
    }
    \label{fig:gap}
\end{figure}

\begin{tcolorbox}[colback=black!10, colframe=black!50, boxrule=0pt, arc=3pt]
\subsubsection*{\textbf{Key Insight~2}}
\emph{Real samples, once pushed into the tail of the mixed feature distribution, turn atypical relative to the dominant synthetic core, forcing the model to memorize them more strongly.}
\end{tcolorbox}
To prove this, we adopt the memorization score~\cite{feldman2020does} to quantify the degree to which a trained model memorizes a given sample, and then establish that the RSMT pipeline fundamentally amplifies memorization for real training samples.

\begin{definition}[Memorization Score]\label{def:mem-score}
Following the formal definition introduced by Feldman~\cite{feldman2020does}, for a learning algorithm $\mathcal{A}$, training set $S$, and sample $z_i = (x_i, y_i) \in S$, the memorization score is
\[
  \mathrm{mem}(\mathcal{A},\!S,\!z_i) \!:=\! \Pr_{h \leftarrow \mathcal{A}(S)}\![h(x_i)\!=\!y_i] - \Pr_{h \leftarrow \mathcal{A}(S \setminus \{z_i\})}\![h(x_i)\!=\!y_i]
\]
A higher score indicates that the algorithm's prediction on $z_i$ relies more heavily on having observed $z_i$ during training, i.e., stronger memorization. For notational convenience, we denote the second term as the leave-one-out (LOO) accuracy: $\mathrm{loo}(\mathcal{A}, S, z_i) := \Pr_{h \leftarrow \mathcal{A}(S \setminus \{z_i\})}[h(x_i) = y_i]$.
\end{definition}

Under near-perfect training fit ($\Pr_{h \leftarrow \mathcal{A}(S)}[h(x_i) = y_i] \geq 1 - \epsilon$), the memorization score in Definition~\ref{def:mem-score} satisfies:
\begin{equation}\label{eq:mem-loo}
  \mathrm{mem}(\mathcal{A}, S, z_i) \geq (1 - \epsilon) - \mathrm{loo}(\mathcal{A}, S, z_i)
\end{equation}
Thus, lower LOO accuracy directly implies stronger memorization.

\begin{theorem}[RSMT Memorization Amplification]\label{thm:mem-amp}
Let $\mathcal{D}_{\mathrm{mix}} = \mathcal{D}_{\mathrm{real}} \cup \mathcal{D}_{\mathrm{syn}}$ with mixing ratio $\lambda = |\mathcal{D}_{\mathrm{syn}}|/|\mathcal{D}_{\mathrm{real}}| > 1$, and let $z_i = (x_i, y_i) \in \mathcal{D}_{\mathrm{real}}$ be a real training sample from class~$c$ with $N_c$ real samples per class. 

\smallskip\noindent\textbf{Assume:}

\noindent\textbf{(C1)} Irreducible distributional gap: $\|\boldsymbol{\mu}_c^{\mathrm{real}} - \boldsymbol{\mu}_c^{\mathrm{syn}}\|_2 \geq \delta > 0$ for every class~$c$. By Theorem~\ref{thm:atypicality}, this gap renders real samples distributionally atypical in the mixed feature space, justifying modeling them as a distinct subpopulation from synthetic data within each class.

\noindent\textbf{(C2)} Near-optimal training fit: $\Pr_{h \leftarrow \mathcal{A}(S)}[h(x_i) = y_i] \geq 1 - \epsilon$ for negligible $\epsilon > 0$, for both $S \in \{\mathcal{D}_{\mathrm{mix}},\, \mathcal{D}_{\mathrm{real}}^{}\}$.

\smallskip\noindent\textbf{Then:}
\[
  \mathbb{E}_{z_i}\!\left[\mathrm{mem}(\mathcal{A}, \mathcal{D}_{\mathrm{mix}}, z_i)\right]
  \;\geq\;
  \mathbb{E}_{z_i}\!\left[\mathrm{mem}(\mathcal{A}, \mathcal{D}_{\mathrm{real}}^{}, z_i)\right]
  + \eta(\delta, \lambda),
\]
where $\eta(\delta, \lambda) > 0$ whenever $\delta > 0$ and $\lambda > 0$.
%and $\eta$ is monotonically increasing in both $\delta$ and $\lambda$.
\end{theorem}

\noindent \emph{Proof.} See Appendix~\ref{app:mem-amp} for details.
The memorization amplification revealed in Theorem~\ref{thm:mem-amp} is grounded in prior works establishing that atypical training samples from the subpopulation cannot be adequately learned by shared statistical patterns and therefore require memorization to achieve close-to-optimal generalization error~\cite{feldman2020does, feldman2020neural}, as networks systematically prioritize simple shared patterns before attending to such examples~\cite{arpit2017closer, garg2024memorization}.
Under RSMT, real samples are precisely those transformed into distributionally atypical samples in peripheral regions (Theorem~\ref{thm:atypicality}), carrying non-standard, instance-specific artifacts that synthetic samples decoded from a standard latent prior cannot reproduce, which renders them hard examples for the network to learn, while synthetic samples act as easy ones, thereby further increasing reliance on memorization.

\begin{figure}[t]
    \centering
    \begin{subfigure}[b]{0.32\linewidth}
        \centering
        \includegraphics[width=\linewidth]{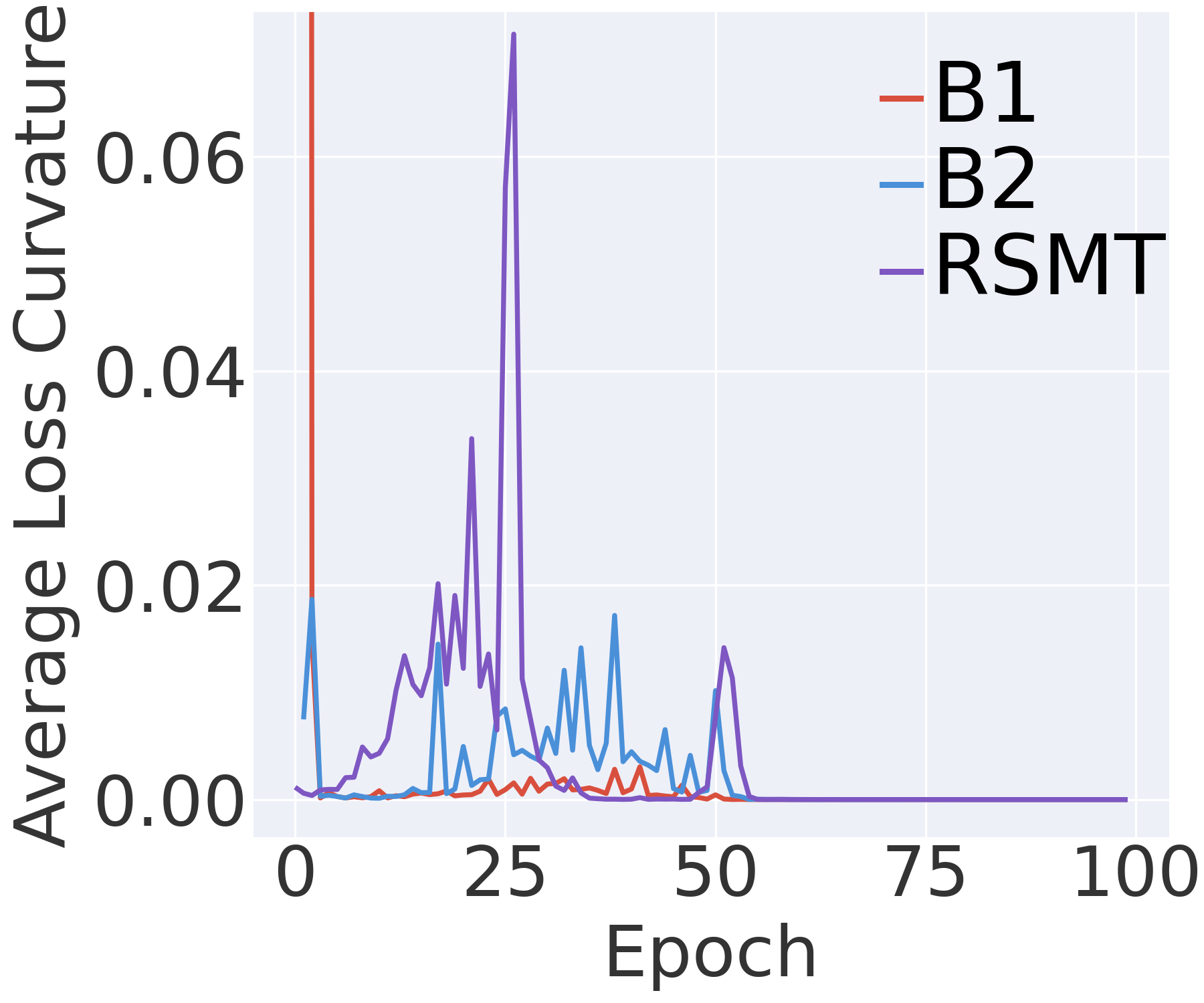}
        \caption{\small RSMT vs.\ real-only baselines.}
        \label{fig:}
    \end{subfigure}
    \hfill
    \begin{subfigure}[b]{0.32\linewidth}
        \centering
        \includegraphics[width=\linewidth]{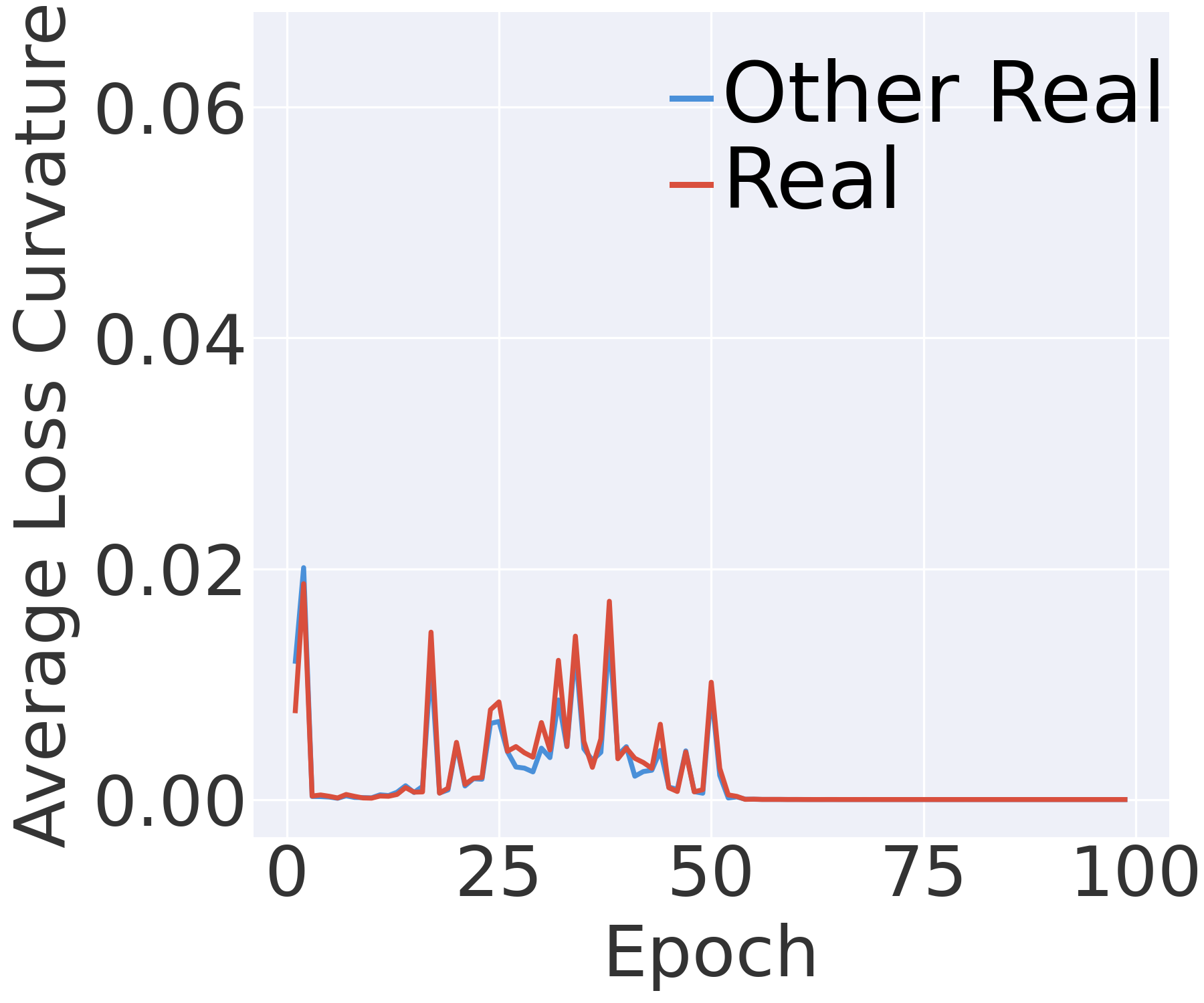}
        \caption{\small Two real subsets within B2.}
        \label{fig:}
    \end{subfigure}
    %\vspace{6pt}
    \hfill
    \begin{subfigure}[b]{0.32\linewidth}
        \centering
        \includegraphics[width=\linewidth]{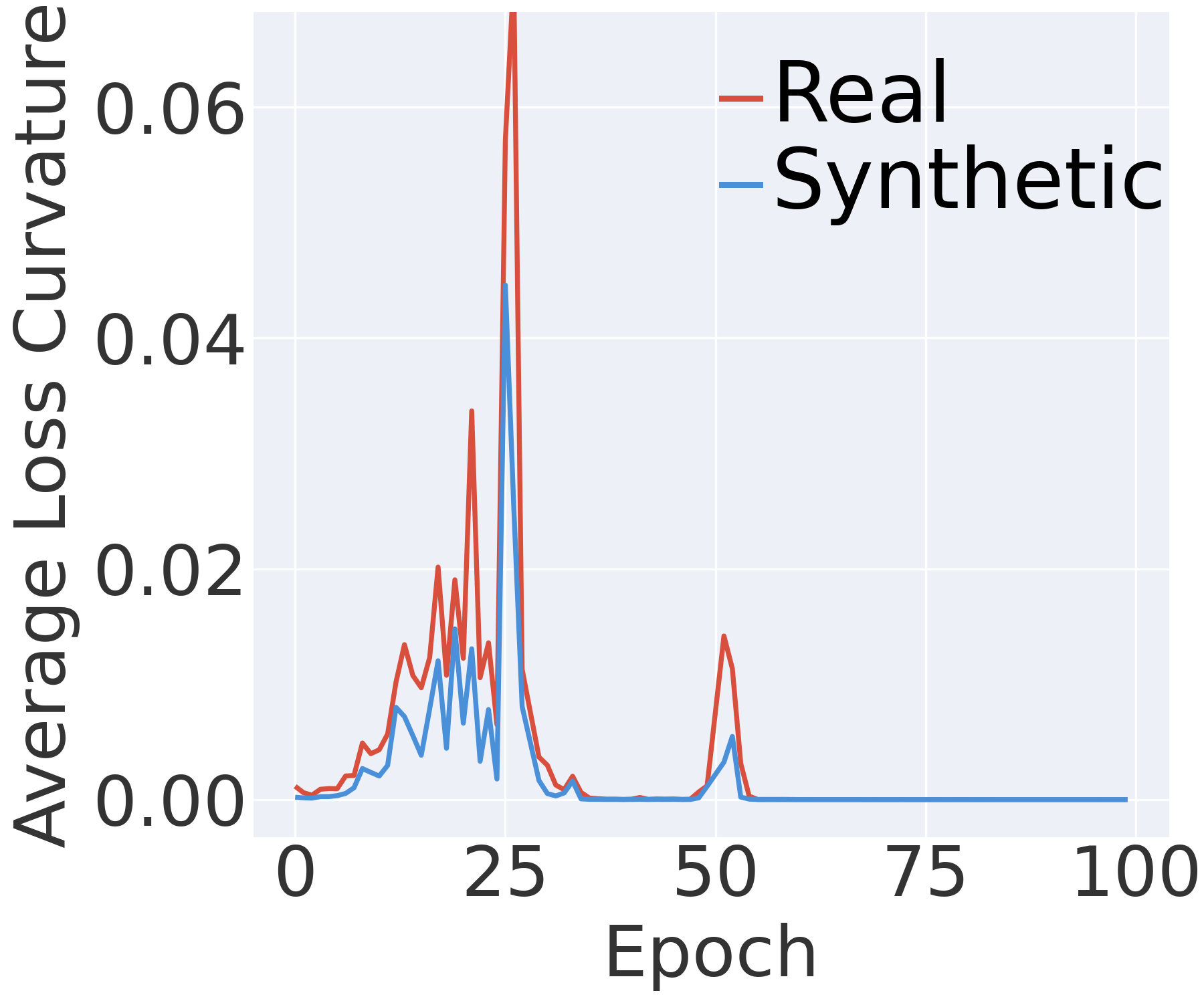}
        \caption{\small Real vs.\ synthetic within RSMT.}
        \label{fig:}
    \end{subfigure}
    \caption{ Per-epoch average loss curvature. 
    }
    \label{fig:curvature}
\end{figure}

\mypara{Comparison Across Real-Only Baselines.}
We compare the memorization of real samples under RSMT against two real-only baselines: \textit{(B1) Same-$N$ baseline}, trained on the same $N$ real samples per class as $\mathcal{D}_{\mathrm{mix}}$ without any augmentation; and \textit{(B2) Accuracy-matched baseline}, trained on $N$ real samples augmented with additional real samples until its test accuracy matches that of the mix-trained model, which rules out overfitting from data scarcity as the source of memorization, leaving the distributional gap introduced by synthetic data as the only remaining factor.

The memorization gap between RSMT and B1 reflects two competing effects: \textbf{(I)}~The additional synthetic data enlarges the total training set, 
partially reducing the per-sample memorization burden as overall model 
quality improves;
\textbf{(II)}~The minority subpopulation effect 
(Theorem~\ref{thm:mem-amp}) increases the per-sample memorization cost for real samples.
When $\lambda$ is large and $\delta$ is substantial, effect~(II) dominates 
and $\mathrm{mem}(\rm RSMT) > \mathrm{mem}(\mathcal{B}_1)$; 
when $\delta$ is small, effect~(I) can dominate and the inequality reverses. 
The comparison against B2 instead reflects only effect~(II), since model 
quality is matched, so $\mathrm{mem}(\rm RSMT) > 
\mathrm{mem}(\mathcal{B}_2)$ should hold whenever $\delta > 0$. 
Section~\ref{main_result} empirically showing that on identical real training samples RSMT amplifies MIA success \emph{not always} above B1 but \emph{consistently} above B2.

\mypara{Empirical Evidence:} 
We quantify the memorization of identical real training samples via per-epoch average loss curvature~\cite{garg2024memorization}, which serves as a reliable proxy for the memorization score in Definition~\ref{def:mem-score}.
Conceptually, higher curvature, especially when sustained over training, indicates stronger memorization, as it reflects the model's increased sensitivity to local perturbations around individual samples.
Figure~\ref{fig:curvature}(a) shows that the same 100 real samples per class under mix-training (100 real + 400 synthetic) exhibit higher curvature than under two real-only baselines.
Figure~\ref{fig:curvature}(b) verifies that two disjoint subsets of real samples within B2 follow nearly identical curvature trajectories, so the gap in (a) must arise from the introduction of synthetic data. Figure~\ref{fig:curvature}(c) further shows that real samples consistently exhibit higher curvature than synthetic samples across epochs and thus form the high-memorization subpopulation within the mixed distribution, in direct agreement with Theorem~\ref{thm:mem-amp}.

\begin{figure*}[t]
    \centering
    \includegraphics[trim=0 0 0 0,clip,width=1.0\textwidth]{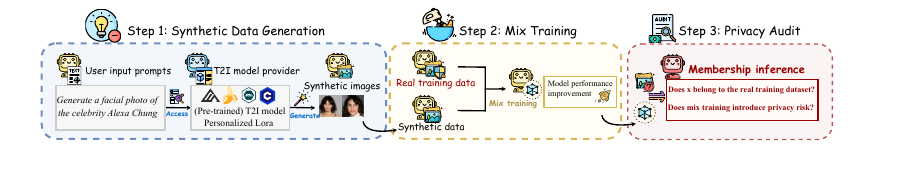}
    \caption{Non-Adversarial \name workflow.}
    \label{fig:pipline}
\end{figure*}

% \section{Benign RSMT Leakage Audit}
\section{Non-Adversarial \name}
\label{sec:standard}
This section presents \name, the first framework systematic investigates whether the standard RSMT pipeline inadvertently exacerbates the privacy leakage of sensitive real training samples through the lens of MIAs.
Notably, our goal in this section is to strictly audit the privacy exposure induced by the intrinsic distributional gap between real and T2I-generated data, thus establishing a lower bound on the privacy risk of RSMT, rather than to engineer an adversarial pipeline deliberately crafted to further amplify such leakage and increase vulnerability to MIAs (Section~\ref{invasive}).
Next, we formalize the threat model and detail the pipeline of Non-Adversarial \name, followed by the experimental results. 

\subsection{Threat Model}
\label{sec:threat1}

Here, we consider a setting in which an auditor (adversary) assesses
whether a benign, non-adversarial RSMT pipeline leaks information about
real training members. The T2I model provider is therefore assumed to be
honest, such that synthetic data are generated without malicious intervention.
% Section~\ref{sec:threat2} later relaxes this assumption to consider an adversarial T2I supply chain. 
Below, we align our threat model
with the standard black-box MIAs~\cite{li2024seqmia,he2024difficulty,liu2022membership,carlini2022membership,shokri2017membership}.

\mypara{Adversary's Knowledge.}
As adopted in prior works~\cite{li2024seqmia,he2024difficulty,liu2022membership},
the adversary is restricted to querying the victim classification model
$M_v$ trained via RSMT through an oracle (e.g., a black-box API) to obtain
the posterior probability distribution for a given input.
Following~\cite{shokri2017membership,carlini2022membership}, the adversary
further possesses a local shadow real dataset $\mathcal{D}_s^{\rm real}$
drawn from the same distribution as the victim's real dataset but without any overlap. The adversary additionally has access to the same T2I model used by the victim---black-box for closed-source APIs, white-box for
open-source backbones---to generate shadow synthetic data
$\mathcal{D}_s^{\rm syn}$.

\mypara{Adversary's Capability.}
The adversary cannot access the internal parameters of the victim model
$M_v$ nor tamper with the RSMT pipeline at training time (e.g., data
poisoning~\cite{ma2024watch}) or inference time (e.g., fault
injection~\cite{li2024yes}). However, the adversary can leverage
$\mathcal{D}_s = \mathcal{D}_s^{\rm real} \cup \mathcal{D}_s^{\rm syn}$ to
train a set of shadow mix-trained models that mimic the behavior of $M_v$.

\mypara{Adversary's Goal.}
Given a target sample $x$ drawn from the same distribution as the real
training data, the adversary aims to determine whether $x$ was included
in the real training set of the victim model $M_v$.

\subsection{Method}
\label{sec:method}

As illustrated in Figure~\ref{fig:pipline}, \name comprises three stages: synthetic data generation, mix training, and privacy auditing. Next, we provide a detailed description of each stage.

\mypara{Stage 1: Synthetic Data Generation.}
We begin by constructing a high-quality synthetic dataset $\mathcal{D}_{\rm syn}$ whose class semantics align with those of the real dataset $\mathcal{D}_{\rm real}$. For example, $\mathcal{D}_{\rm real}$ contains facial photographs of celebrities such as Amber Heard and the goal is to synthesize photo-realistic facial images of the same identities to serve as augmentation. The stage proceeds in two steps.

\noindent \textit{T2I Model Preparation.} First, we select a T2I model suited for the target domain. More concretely, for general-domain concepts (e.g., ImageNet categories), large-scale pre-trained models such as SD1.5~\cite{rombach2022high} or Flux~\cite{labs2025flux1kontextflowmatching} are typically sufficient, as their broad training corpora provide adequate coverage.
For specialized or fine-grained concepts (e.g., celebrity faces, satellite imagery), the base model rarely reproduces the target concept. We therefore apply parameter-efficient fine-tuning via LoRA~\cite{hu2022lora} to adapt the base T2I model to each target concept, either by fine-tuning on a small set of real samples or by reusing publicly released concept-specific checkpoints from community repositories such as Civitai~\cite{civitai}. When budget permits, commercial T2I APIs such as Google's Nano Banana~\cite{comanici2025gemini} offer an alternative that renders most concepts out of the box.

\noindent \textit{Prompt Design.}
We then construct class-conditional prompts via a diverse set of natural language templates that vary contextual attributes such as background, lighting, and expression, while holding the class name fixed. 
For example, templates of the form ``\textit{A facial photo of the celebrity [class name] under natural daylight}'' and ``\textit{A facial photo of the celebrity [class name] with long straight golden-blonde hair, a warm smile}'' are instantiated with identity-specific class names (e.g., \textit{[class name] = Amber Heard}) to produce visually diverse synthetic samples per identity.

Based on the above preparations, the instantiated prompts are then fed into the prepared T2I model (with/without LoRA) to generate class-wise synthetic images, which are automatically labeled by the class name embedded in each prompt and aggregated into $D_{\rm syn}$.

%---------------------------------------------------------------
\mypara{Stage 2: Mix Training.}
Synthetic data complements the real dataset to improve downstream classification performance, particularly when real data is scarce. 
For instance, the limited per-celebrity facial photographs are augmented with synthesized facial images to enlarge the training set. Formally, a victim model $M_v$ is trained on $\mathcal{D}_{\rm real} \cup \mathcal{D}_{\rm syn}$ with mixing ratio $\lambda = |\mathcal{D}_{\rm syn}|/|\mathcal{D}_{\rm real}| \geq 1$, and is expected to outperform a baseline model trained on $\mathcal{D}_{\rm {real}}$ alone.

%---------------------------------------------------------------

\mypara{Stage 3: Privacy Audit via Membership Inference.}
To quantify whether the standard RSMT paradigm amplifies privacy leakage on valuable real training data, we conduct the membership inference attack against the victim model $M_v$. An adversary with black-box query access $M_v$ attempts to determine, for a given target sample $x$, whether $x \in D_{\rm real}$. Here, we consider four representative membership inference attacks in non-adversarial \name, including two train-based and two metric-based inference strategies:

\textit{Train-based} strategies rely on shadow models to generate meta-data, which is then used to train a meta-classifier for membership inference. 
Different methods vary in how they construct meta-data and design the meta-classifier. Shokri \textit{et al.}~\cite{shokri2017membership} rely solely on output posteriors as meta-data, with a random forest serving as the meta-classifier in our experiments. Yuan \textit{et al.}~\cite{yuan2022membership} extend this by incorporating posteriors, labels, and sensitivity as meta-data features, fed into a transformer-based meta-classifier.

\textit{Metric-based} strategies determine membership by thresholding metrics computed from the victim model’s outputs, requiring no meta-classifier training. Different approaches rely on different metrics. In this work, we consider two metrics: one based on posteriors~\cite{song2019privacy}, and the other based on a pairwise likelihood ratio test, referred to as RMIA~\cite{ZarifzadehLS24}, a low-cost and high-power variant of LiRA~\cite{carlini2022membership}.
%Thresholds are selected either optimally (worst-case, maximizing separation between member and non-member score distributions) or learned via shadow models in practice.
\subsection{Experiment Setup}

\mypara{Datasets.}
We evaluate on five datasets spanning diverse domains: 
EuroSAT~\cite{helber2019eurosat} for satellite imagery, PatternNet~(5)~\cite{zhou2018patternnet} for remote sensing, VGGFace2~(10)~\cite{cao2018vggface2} for facial identities, ImageNet10~\cite{deng2009imagenet} and ImageNet100~\cite{tian2020contrastive} for natural images. 
Numbers in parentheses, e.g., VGGFace2~(10), denote the number of classes sampled from the original dataset for classification, as the original VGGFace2 contains over 9000 identities with severely imbalanced per-class image counts.
Due to space constraints, detailed dataset descriptions are provided in Appendix~\ref{app:datasets}.

\mypara{Baselines.}
We consider two real-only baselines trained without any synthetic data described in Section~\ref{sec:theory}.

\noindent \textit{(B1) Same-$N$ baseline} trained on the same $N$ real samples as those employed in RSMT, without any synthetic augmentation.

\noindent \textit{(B2) Accuracy-matched baseline} trained on the $N$ real samples augmented with additional real samples until its test accuracy matches that of the mix-trained model, ensuring that any memorization gap reflects the influence of synthetic data itself, rather than a larger training set reducing overfitting affecting memorization.

\mypara{Models.}
Our pipeline comprises an \textit{upstream T2I generator} that produces synthetic training samples and a \textit{downstream victim classifier} that consumes them.

On the upstream side, we adopt four T2I models spanning closed- and open-source models: two closed-source commercial APIs, Google's Nano Banana~\cite{comanici2025gemini} and OpenAI's latest ChatGPT Images 2.0~\cite{chatgptimages2026}, together with two open-source pre-trained models built on distinct architectures, SD1.5~\cite{rombach2022high} and Flux-mini~\cite{tencentarc2024fluxmini} (a lightweight substitute for FLUX.1-dev). 
SD1.5 follows the UNet-based denoising architecture and remains a standard reference in academic research~\cite{Pang025,li2025towards,sun2025pretender}, while Flux represents a more recent open-source model built on the increasingly prevalent DiT architecture.
For both open-source models, we further apply LoRA-based fine-tuning on each dataset to improve concept-specific generation quality. 

On the downstream side, we adopt the CNN-based ResNet-18~\cite{he2016deep} as the victim classification model.

\mypara{Mixing Ratio.}
We fix the real-to-synthetic ratio at $1{:}4$, with $100$ real and $400$ synthetic samples per class, simulating the rare data regime where real samples are scarce and synthetic augmentation dominates the training set.

\mypara{Evaluation Metrics.}
Following prior studies~\cite{li2024seqmia,he2024difficulty,liu2022membership,ye2022enhanced}, we evaluate MIA using balanced accuracy and AUC for average-case performance, and TPR @ low FPR and the full log-scale ROC curve following Carlini \textit{et al.}~\cite{carlini2022membership} to assess the low false-positive regime, which is more reflective of real-world privacy risk.
\begin{itemize}
    \item \textit{Balanced Accuracy} quantifies the probability that an MIA correctly infers membership status over a balanced set of members and non-members.
    \item \textit{AUC} is the area under the receiver operating characteristic (ROC) curve~\cite{sankararaman2009genomic} , which measures the ability to distinguish members from non-members across all possible decision thresholds.
    \item \textit{TPR @ low FPR} measures the true positive rate at a stringently low false positive rate (e.g., 0.1\% FPR) and serves as a more reliable indicator of privacy leakage than balanced accuracy or AUC, which are often dominated by the correct identification of non-members.

\end{itemize}

\begin{table}[t]
\centering
\caption{ResNet-18 model classification accuracy across datasets and fine-tuned T2I models.}
\label{tab:acc-finetuned}
\scalebox{0.7}{
\begin{tabular}{ll ccc ccc}
\toprule
\multirow{2}{*}{T2I Model}
  & \multirow{2}{*}{Dataset}
  & \multicolumn{3}{c}{Train Acc (\%)}
  & \multicolumn{3}{c}{Test Acc (\%)}\\
\cmidrule(lr){3-5} \cmidrule(lr){6-8}
   &
  & B1 & B2 & RSMT
  & B1 & B2 & RSMT\\
\midrule

\multirow{4}{*}{\makecell[l]{SD1.5}}
  &EuroSAT
    &100&100&100 &76.9&88.0&88.3\\
  &VGGFace2 (10)
    &100&100&100 &64.9&78.7&78.9\\
  &ImageNet10-A
    &100&100&100 &62.9&76.3&76.4\\
  &ImageNet100
    &99.9&99.9&99.9 &46.5&62.0&62.7\\
\midrule
%180
\multirow{3}{*}{\makecell[l]{Flux-mini}}
  &PatternNet (5)
    &100&100&100 &86.2&96.6&96.6\\
  &VGGFace2 (10)
    &100&100&100 &64.9&80.0&81.8\\
  &ImageNet10-B
    &100&100&100 &59.8&71.0&71.7\\
\bottomrule
\end{tabular}%
}
\\[1pt] \raggedright \scriptsize ImageNet10-A and ImageNet10-B are two 10-class subsets of ImageNet, selected as the classes that SD1.5 and Flux-mini can generate with high fidelity.
\end{table}

\begin{table}[t]
\centering
\caption{Attack performance of different attacks on varying datasets and T2I models (pre-trained and commercial).}
\label{tab:results_pre}
\scalebox{0.6}{
\begin{tabular}{ll cc>{\columncolor{mixbg}}c cc>{\columncolor{mixbg}}c cc>{\columncolor{mixbg}}c}
\toprule
\multirow{2}{*}{T2I Model}
  & \multirow{2}{*}{Attack}
  & \multicolumn{3}{c}{MIA Acc (\%)}
  & \multicolumn{3}{c}{MIA Auc (\%)}
  & \multicolumn{3}{c}{TPR @ 0.1\% FPR (\%)} \\
\cmidrule(lr){3-5} \cmidrule(lr){6-8} \cmidrule(lr){9-11}
  Dataset &Method
  & B1 & B2 & RSMT
  & B1 & B2 & RSMT
  & B1 & B2 & RSMT \\
\midrule
\multirow{4}{*}{\makecell[l]{Nano Banana\\VGGFace2(3)}}
  &\cite{song2019privacy}
  &73.5 &72.8 &\textbf{79.5} &74.4 &74.0 &\textbf{78.7} &0.0&0.0&0.0\\
  &~\cite{yuan2022membership}
    &72.0 &72.0 &\textbf{77.6} &75.5 &73.9 &\textbf{80.0} &0.0&0.0&\textbf{4.0}\\
  &~\cite{shokri2017membership}
    &65.8 &64.5 &\textbf{75.6} &72.4 &71.1 &\textbf{79.7} &1.7&0.3&\textbf{3.7}\\
  &~\cite{ZarifzadehLS24}
    &71.6 &66.8 &\textbf{72.3} &67.7 &63.2 &\textbf{70.9} &\textbf{1.3}&0.0&0.0\\
\midrule

\multirow{4}{*}{\makecell[l]{Images 2.0\\VGGFace2(3)}}
  &\cite{song2019privacy}
  &73.5 &70.6 &\textbf{76.5} &74.4 &72.0 &\textbf{76.4} &0.0&0.0&0.0\\
  &~\cite{yuan2022membership}
    &72.0 &70.5 &\textbf{74.8} &75.5 &73.6 &\textbf{75.9} &0.0&0.0&0.0\\
  &~\cite{shokri2017membership}
    &65.8 &64.1 &\textbf{70.8} &72.4 & 70.0 &\textbf{73.5} &\textbf{1.7}&0.6&1.1\\
  &~\cite{ZarifzadehLS24}
    &\textbf{71.6} &64.6 &67.5 &67.7 &60.3 &\textbf{68.9} &\textbf{1.3}&0.0&0.0\\
\midrule

\multirow{4}{*}{\makecell[l]{SD1.5\\ImageNet100}}
  &\cite{song2019privacy}
    &93.4 &87.2 &\textbf{95.2} &93.9 &88.0 &\textbf{95.6} &0.0 &0.0 &0.0 \\
  &~\cite{yuan2022membership}
    &89.2 &84.7 &\textbf{93.5} &95.4 &90.4 &\textbf{97.0} &3.0 &2.5 &\textbf{7.7} \\
  &~\cite{shokri2017membership}
    &86.3 &80.1 &\textbf{94.1} &93.2 &86.9 &\textbf{97.8} &2.5 &2.1 &\textbf{8.2} \\
  &~\cite{ZarifzadehLS24}
    &75.1 &73.1 &\textbf{77.2} &83.1 &79.3 &\textbf{86.0} &12.0 &7.9 &\textbf{21.2} \\
\midrule
\multirow{4}{*}{\makecell[l]{Flux-mini\\ImageNet10-B}}
  &\cite{song2019privacy}
    &82.3 &76.3 &\textbf{84.1} &\textbf{83.2} &77.3 &81.6 &0.0&0.0&0.0\\
  &~\cite{yuan2022membership}
    &\textbf{82.9} &75.9 &81.0 &\textbf{87.7} &80.8 & 85.1 &0.2 &0.2 &\textbf{0.8} \\
  &~\cite{shokri2017membership}
    &76.6 &67.9 &\textbf{79.2} &83.6 &73.7 &\textbf{85.3} &2.7 &0.3&\textbf{2.8}\\
  &~\cite{ZarifzadehLS24}
    &\textbf{75.2} &71.0 &72.4 &\textbf{78.6} &75.1 &76.6 &4.3&1.7&\textbf{4.4}\\
\bottomrule
\end{tabular}%
}
\end{table}

\subsection{Evaluation Results}
\label{main_result}

%\hl{[HYOUNG]The section title ``Main Result'' sounds odd. In this section, include a quantitative comparison of how much RSMT amplifies leakage relative to B1/B2, which attack is stronger, and the pattern differences across T2I backbones. I think the Discussion should also be in a separate section.}

\mypara{Generation Quality and Downstream Utility.}
As illustrated in Figure~\ref{fig:gen-quality} in Appendix~\ref{app:result}, synthetic samples produced by the closed-source commercial APIs and the open-source models are visually consistent with the corresponding real samples in class semantics. Tables~\ref{tab:acc-pretrained} (Appendix~\ref{app:result}) and~\ref{tab:acc-finetuned} further report downstream classification accuracy under the pre-trained/API and domain-fine-tuned regimes, respectively, where RSMT consistently outperforms the same-$N$
real-only baseline (B1)
%---trained on the same $N$ real samples as RSMT without synthetic augmentation---
across all settings, confirming that synthetic augmentation delivers a non-trivial utility gain.

\mypara{MIA Performance Against B2.} Tables~\ref{tab:results_pre}, \ref{tab:results_sd}, and~\ref{tab:resultsflux} report MIA performance across three configurations: pre-trained open-source backbones (SD1.5, Flux-mini) and commercial APIs (Google's Nano Banana, OpenAI's ChatGPT Images 2.0), and LoRA-fine-tuned SD1.5 and Flux-mini, respectively. 
RSMT surpasses the accuracy-matched real-only baseline (B2)
%---augmenting the $N$ real samples with extra real samples until its test accuracy matches RSMT model---
across all MIA attacks, datasets, and backbones evaluated, isolating the synthetic data as the dominant source of memorization of real training data once model utility is matched (i.e., the effect of overfitting is ruled out). This amplification persists regardless of whether the synthetic data originates from pre-trained backbones, LoRA-fine-tuned variants, or commercial APIs. For instance, on fine-tuned SD1.5/VGGFace2(10) with attack~\cite{yuan2022membership}, RSMT lifts MIA accuracy from 77.1\% to 85.7\%, AUC from 81.0\% to 91.1\%, and TPR@0.1\% FPR from 0.1\% to 9.5\%; on Nano Banana/VGGFace2(3) with attack~\cite{shokri2017membership}, MIA accuracy rises from 64.5\% to 75.6\% and AUC from 71.1\% to 79.7\%.

\mypara{MIA Performance Against B1.}
MIA performance under RSMT exceeds that under B1 in the majority of settings, e.g., on pre-trained SD1.5/ImageNet100 with attack~\cite{ZarifzadehLS24}, accuracy rises from 75.1\% to 77.2\%, AUC from 83.1\% to 86.0\%, and TPR@0.1\% FPR from 12.0\% to 21.2\%, indicating that synthetic data amplifies privacy leakage on real training samples even while serving as augmentation and mitigating overfitting. However, MIA performance under RSMT remains below B1 in a few settings, such as synthetic data generated by fine-tuned Flux-mini, where the smaller $\delta$ (Table~\ref{tab:delta}) weakens the minority-subpopulation effect of Theorem~\ref{thm:mem-amp} (effect II), so that the opposing benefit of the enlarged training set in reducing the per-sample memorization burden (effect I) dominates, in line with the two-effect trade-off characterized in Section~\ref{sec:theory}.

\mypara{Pre-trained vs. Fine-tuned on the Same Backbone.} LoRA fine-tuning narrows but does not close the leakage gap relative to the pre-trained variant, as fine-tuning brings the synthetic outputs closer to the real distribution and further lifts downstream classification accuracy. For instance, under attack~\cite{shokri2017membership} on SD1.5/ImageNet100, TPR@0.1\% FPR drops from 8.2\% (pre-trained) to 4.9\% (fine-tuned), with AUC decreasing from 97.8\% to 96.7\%. The amplification over B2 nonetheless persists, since the irreducible artifacts of T2I generation keep $\delta>0$.

\mypara{Effect of the Distributional Gap $\delta$/Different T2I Backbones.}
We further examine how the real–synthetic distributional gap $\delta$ (per-class real–synthetic centroid distance), induced by different T2I backbones, impacts privacy leakage. 
By fixing $\mathcal{D}_{\mathrm{real}}$ and the mixing ratio $\lambda$, we vary only the T2I backbone used to generate the synthetic data, resulting in different $\delta$ values across backbones.
 
Table~\ref{tab:delta} reports how the RSMT-induced privacy leakage varies with $\bar{\delta}$, measured by the MIA accuracy and AUC gaps ($\Delta\mathrm{Acc}$, $\Delta\mathrm{AUC}$) between the RSMT model and the B1, revealing two key trends.
First, the MIA effectiveness against the identical real training samples grows with $\bar{\delta}$, with SD1.5 (larger $\bar{\delta}$) yielding higher $\Delta\mathrm{Acc}$ and $\Delta\mathrm{AUC}$ than Flux-mini across both datasets (e.g., on VGGFace2, $+3.7\%$ vs $+0.5\%$ in $\Delta\mathrm{Acc}$).
Second, when $\bar{\delta}$ is sufficiently small, $\Delta\mathrm{Acc}$ and $\Delta\mathrm{AUC}$ can even turn negative, indicating that RSMT reduces rather than amplifies MIA leakage, as shown on ImageNet10-B, where Flux-mini $\bar{\delta}$ = 0.2531 reaches $\Delta\mathrm{AUC}=-2.0\%$.
These findings empirically confirm that
%the two-effect trade-off characterized in Section~\ref{sec:theory}. 
large $\delta$ strengthens the displacement bound $\bigl(\tfrac{\lambda}{1+\lambda}\bigr)^{2}\delta^{2}$ and amplifies memorization (effect~(II) dominates), whereas small $\delta$ lets the enlarged training set instead reduces the per-sample memorization burden and suppress leakage below the B1 (effect~(I) dominates).

\begin{table}[t]
\centering
\caption{Attack performance of different attacks on varying datasets and fine-tuned SD1.5.}
\label{tab:results_sd}
\scalebox{0.6}{
\begin{tabular}{ll cc>{\columncolor{mixbg}}c cc>{\columncolor{mixbg}}c cc>{\columncolor{mixbg}}c}
\toprule
\multirow{2}{*}{Dataset}
  & \multirow{2}{*}{Attack}
  & \multicolumn{3}{c}{MIA Acc (\%)}
  & \multicolumn{3}{c}{MIA Auc (\%)}
  & \multicolumn{3}{c}{TPR @ 0.1\% FPR (\%)} \\
\cmidrule(lr){3-5} \cmidrule(lr){6-8} \cmidrule(lr){9-11}
   &Method
  & B1 & B2 & RSMT
  & B1 & B2 & RSMT
  & B1 & B2 & RSMT \\
\midrule

\multirow{4}{*}{ImageNet10-A}
  &~\cite{song2019privacy}
    &70.3 &66.0 &\textbf{76.0} &69.7 &63.8 &\textbf{73.7} &0.0&0.0&0.0\\
  &~\cite{yuan2022membership}
    &71.0 &66.5 &\textbf{74.6} &75.1 &70.6 &\textbf{77.4} &0.3&0.1&\textbf{0.8}\\
  &~\cite{shokri2017membership}
    &63.4 &60.6 &\textbf{72.6} &67.0 &63.9 &\textbf{78.9} &0.3&0.4&\textbf{2.2}\\
  &~\cite{ZarifzadehLS24}
  &67.4 &65.0 &\textbf{69.7} &72.6 &69.6 &\textbf{74.1} &0.3&0.0&\textbf{4.3}\\
\midrule
\multirow{4}{*}{VGGFace2(10)}
  &~\cite{song2019privacy}
    & 85.3& 79.3 & \textbf{89.0} & 88.5 & 82.5 &\textbf{91.0} &0.3 &0.8 &\textbf{1.1} \\
  &~\cite{yuan2022membership}
    & 85.0 & 77.1 &\textbf{85.7} &89.4 &  81.0 &\textbf{91.1}&1.6&0.1 &\textbf{9.5}\\
  &~\cite{shokri2017membership}
    &81.6 & 76.8 &\textbf{86.7} & 87.2 & 79.8 &\textbf{91.7} &1.0 &0.3 &\textbf{7.9} \\
  &~\cite{ZarifzadehLS24}
    &76.5 & 70.3 &\textbf{78.3} &74.8 &66.0 &\textbf{78.4} &\textbf{3.6}&0.7 &0.8\\
\midrule

\multirow{4}{*}{ImageNet100}
  &~\cite{song2019privacy}
    &93.4 &87.2 &\textbf{94.8} &93.9 &88.0 &\textbf{95.0} &0.0&0.0 &0.0 \\
  &~\cite{yuan2022membership}
    &89.2 &84.7 &\textbf{93.1} &95.4 &90.4 &\textbf{97.0} &3.0 &2.5 &\textbf{8.9} \\
  &~\cite{shokri2017membership}
    &86.3 &80.1 &\textbf{93.3} &93.2 &93.2 &\textbf{96.7} &2.5 &2.1 &\textbf{4.9} \\
  &~\cite{ZarifzadehLS24}
    &75.1 &73.1 &\textbf{75.5} &83.1 &79.3 &\textbf{85.0} &12.0 &7.9 &\textbf{19.7} \\

\midrule
\multirow{4}{*}{EuroSAT}
  &~\cite{song2019privacy}
    &66.3 & 61.4 &\textbf{66.5} & 65.0 &  60.8 & \textbf{65.3} &\textbf{0.3} &0.0  &0.0 \\
  &~\cite{yuan2022membership}
    &\textbf{66.1}& 60.1 & 65.5 &\textbf{68.9} & 62.0 & 68.4 &0.0 &0.1 &\textbf{0.3} \\
  &~\cite{shokri2017membership}
    &\textbf{61.8} & 54.7 &  60.0 &\textbf{64.6} &56.5 &  63.8 &0.1  &0.0&\textbf{0.3} \\
  &~\cite{ZarifzadehLS24}
    &\textbf{63.4} &59.9 & 62.6 &66.1 & 62.6 & \textbf{66.8} &1.2 &0.6 &\textbf{2.3} \\

\bottomrule
\end{tabular}%
}
\end{table}

\begin{table}[t]
\centering
\small
\caption{Effect of the distributional gap $\bar{\delta}$ on RSMT privacy leakage. $\Delta\mathrm{Acc}$ and $\Delta\mathrm{AUC}$ denote the MIA performance gap using~\cite{song2019privacy} between the RSMT model and the same-$N$ real-only baseline; positive values indicate leakage amplification. }
\label{tab:delta}
\scalebox{0.7}{
\begin{tabular}{llccc}
\toprule
Dataset & Backbone & $\bar{\delta}$ & $\Delta\mathrm{Acc}$ (\%) & $\Delta\mathrm{AUC}$ (\%) \\
\midrule
\multirow{2}{*}{VGGFace2(10)} 
  & Fine-tuned SD1.5 & $0.3501$ & $+3.7$ & $+2.5$ \\
  & Fine-tuned Flux-mini  & $0.2783$ & $+0.5$ & $+0.3$ \\
\midrule
\multirow{2}{*}{ImageNet10-B}
  & Fine-tuned SD1.5 & $0.3092$ & $+2.2$ & $+0.3$ \\
  & Fine-tuned Flux-mini  & $0.2531$ & $-1.3$ & $-2.0$ \\
\midrule
\multirow{2}{*}{VGGFace2(3)}
  & Nano Banana & $0.5634$ & $+6.0$ & $+4.3$ \\
  & Images 2.0  & $0.4544$ & $+3.0$ & $+2.0$ \\
\bottomrule
\end{tabular}}
\end{table}

%\subsection{Discussion}
\label{sec:discuss1}

\mypara{Effect of the Mixing Ratio ($\lambda$).} 
To examine how the proportion of synthetic data affects privacy leakage, we fix the real training set at 100 samples per class and augment it with 200, 400, and 600 T2I-generated synthetic samples per class, corresponding to $\lambda \in \{2, 4, 6\}$.

Figure~\ref{fig:radio} reports both the utility (classification accuracy of the victim model) and the privacy leakage (MIA performance measured by~\cite{song2019privacy}). 
We highlight two consistent findings from these results.
First, as $\lambda$ grows, the classification accuracy of the mix-trained model improves from 86\% to 90\%, confirming that additional synthetic data continues to deliver measurable utility gains, though at a diminishing rate.
Second, and more critically, the MIA effectiveness against the identical real training samples also grows monotonically, with the AUC rising from 63.8\% to 65.5\% and the attack accuracy from 65.3\% to 66.8\%.
%; in contrast, the real-only baseline trained on the identical 100 real samples remains fixed at an accuracy of 0.769 and a substantially lower MIA AUC of 0.627. 
These observations empirically corroborate Theorem~\ref{thm:atypicality}--\ref{thm:mem-amp}: as $\lambda$ increases, the real samples are pushed further into peripheral regions of the mixed feature space and behave as an increasingly minority subpopulation, compelling the model to memorize them more aggressively. 

\begin{figure}[t]
\centering
    \begin{minipage}{0.46\linewidth}  
        \centering
        \includegraphics[width=\textwidth]{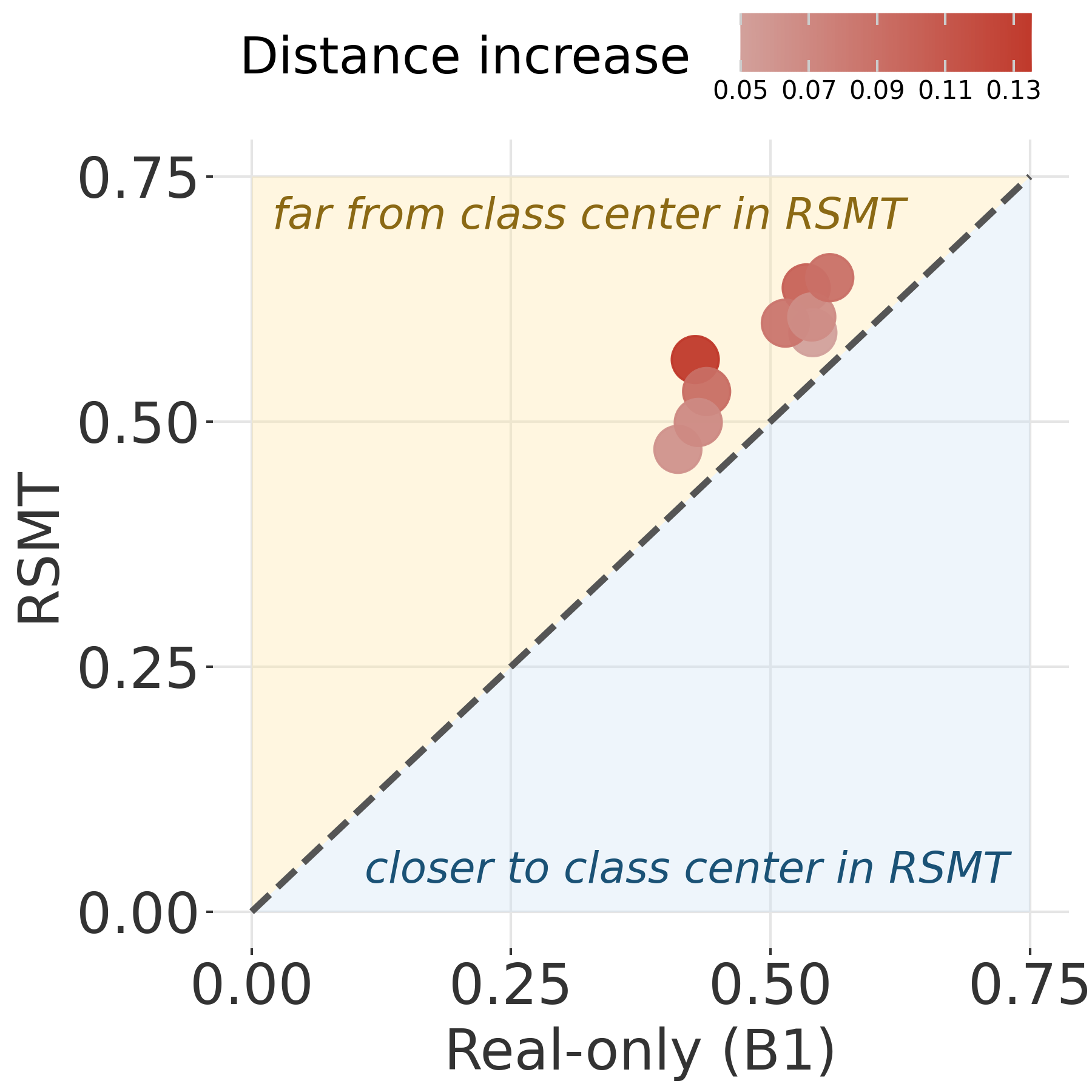}
        \subcaption{\footnotesize RSMT vs B1}\label{fig:c}
    \end{minipage} 
    \begin{minipage}{0.46\linewidth}  
        \centering
        \includegraphics[width=\textwidth]{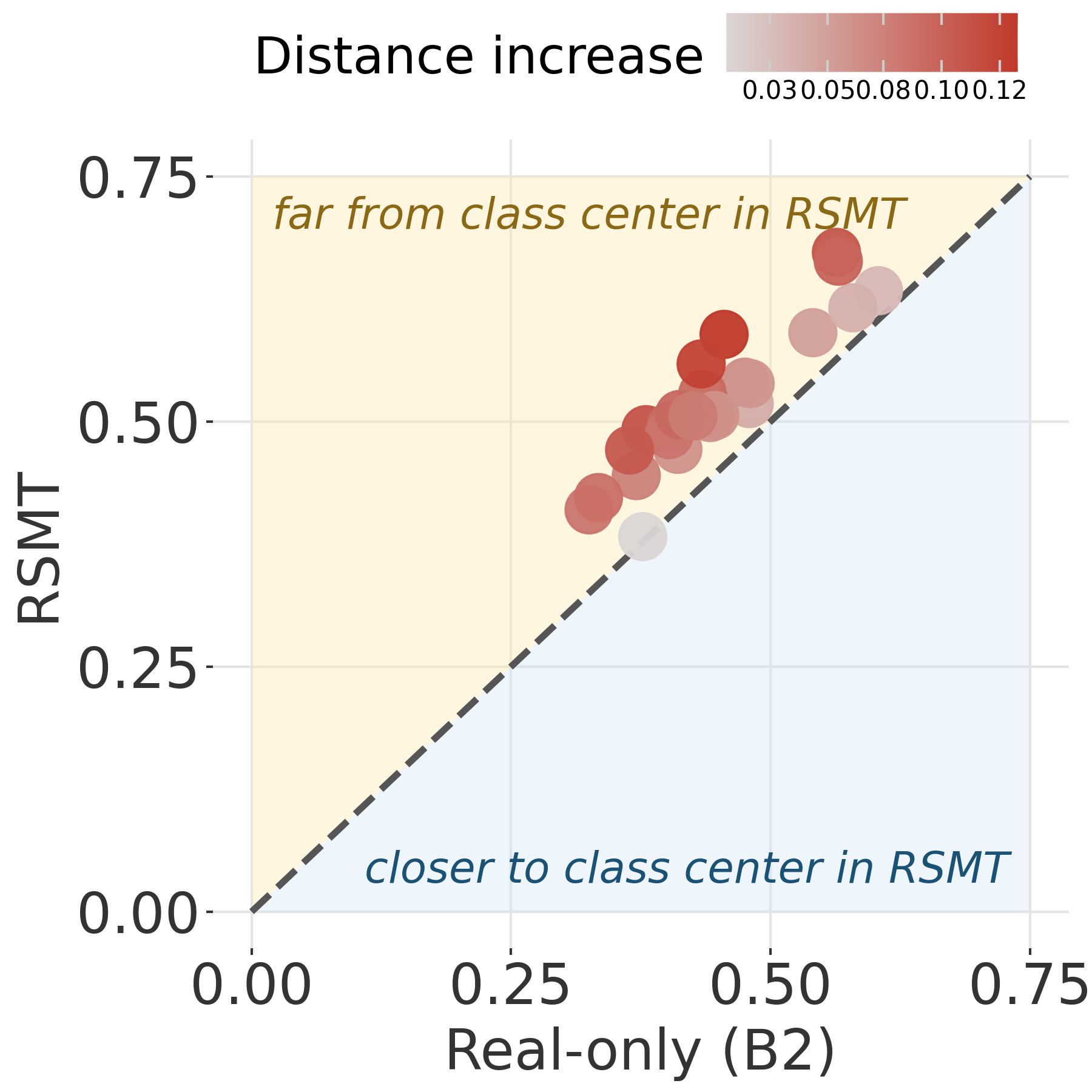}
        \subcaption{\footnotesize RSMT vs B2}\label{fig:c}
    \end{minipage}  
    \caption{Centroid distance of vulnerable real samples under real-only baselines (B1, B2) and RSMT, with the $y=x$ diagonal as reference: points above lie farther from the class centroid under RSMT, points below lie closer.}
    \label{fig:center}
\end{figure}

\mypara{Which Real Samples are Vulnerable?}
In the above experiments, we have demonstrated that incorporating synthetic data increases the privacy leakage of real training samples. Building on this, we further identify which real samples amplify this leakage under RSMT and analyze their properties.

To begin with, we focus on \textit{vulnerable real training samples}, whose membership inference transitions from failure in the non-mixed (trained with only raw data, i.e., B1 or B2) scenario to success under RSMT.
Firstly, to systematically characterize their behavioral shifts in feature space and verify whether mixing renders them more atypical, we measure distributional properties by computing the distance to class centroids (i.e., deviation from class centers). 
%As larger distances indicate increased peripheral, Figure~\ref{fig:center} shows that vulnerable real training samples exhibit increased distances to class centroids in RSMT than in real-only settings.
%This indicates that the introduction of synthetic data shifts these samples toward more atypical, near-boundary regions of the feature space, thereby increasing their susceptibility to stronger memorization, as discussed in Section~\ref{sec:theory}.
Figure~\ref{fig:center} reveals two consistent observations. First, even in the real-only baselines, vulnerable samples already lie relatively far from their class centroids rather than near the dense core, indicating that intrinsic peripherality is a precondition for RSMT-induced exposure. Second, under RSMT, the same samples are pushed even further from their class centroids, with all points lying above the $y=x$ diagonal. Together, these observations show that RSMT does not create vulnerability from scratch but \textit{selectively amplifies it on samples already lying far from their class centers}, displacing them further outward,  
%and rendering them atypical
thereby increasing their susceptibility to stronger memorization.

Secondly, building on the above finding that vulnerable real samples become increasingly peripheral and atypical, we further verify whether they are indeed more deeply memorized under RSMT, utilizing loss curvature as a proxy to quantify memorization strength following~\cite{garg2024memorization}. As illustrated in Figure~\ref{fig:vul_loss}, the evolution of average loss curvature for vulnerable real samples differs markedly between RSMT (red) and training using only real samples (blue). Under RSMT, their curvature remains consistently high during the mid-training phase, whereas in the real-only setting, it rises briefly at early stages before rapidly stabilizing at a low level. This divergence indicates that, in the presence of synthetic data, these samples are \textit{less prone to being captured by shared patterns and instead require stronger specific fitting}, thereby exhibiting more pronounced memorization behavior.

%These critical findings highlight that these vulnerable real samples share common characteristics that are atypical in RSMT scenarios, forcing the reinforcement model to memorize deeply and be exposed to MIA.

\section{Adversarial \name}
\label{invasive}

Building on the theoretical insight that memorization amplification is
fundamentally driven by the real--synthetic distributional gap $\delta$
(Sections~\ref{sec:theory} and~\ref{sec:discuss1}), we investigate
adversarial \name, in which an adversary deliberately enlarges $\delta$
on a chosen target class. The resulting privacy leakage further exceeds
that induced by natural T2I artifacts alone. We begin by defining
adversaries who can manipulate the upstream synthetic data, then present the attack methodology and evaluation results. This section designs and
evaluates two amplification schemes operating at different levels:
A high-level semantic attribute binding and an imperceptible pixel-level
coating, both of which preserve downstream classification utility.
%Notably, this section evaluates privacy risks of two orthogonal amplification schemes—spectral and semantic—both of which preserve downstream classification utility, pass standard text–image alignment and perceptual-quality filters, and confine their footprint to a single target class.

\begin{figure}[t]
\centering
    \begin{minipage}{0.45\linewidth}  
        \centering
        \includegraphics[width=\textwidth]{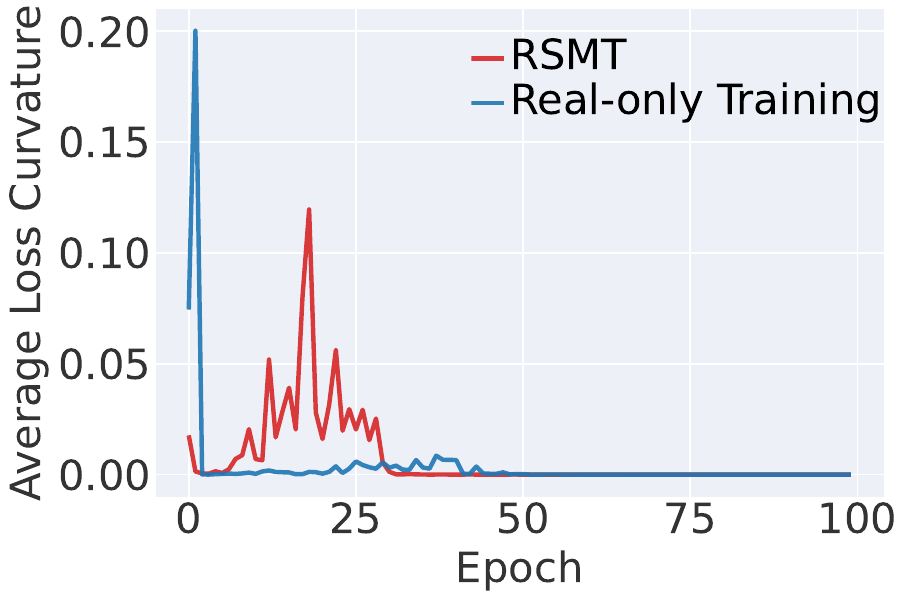}
        \subcaption{\footnotesize RSMT vs B1}\label{fig:c}
    \end{minipage}  
    \begin{minipage}{0.45\linewidth}  
        \centering
        \includegraphics[width=\textwidth]{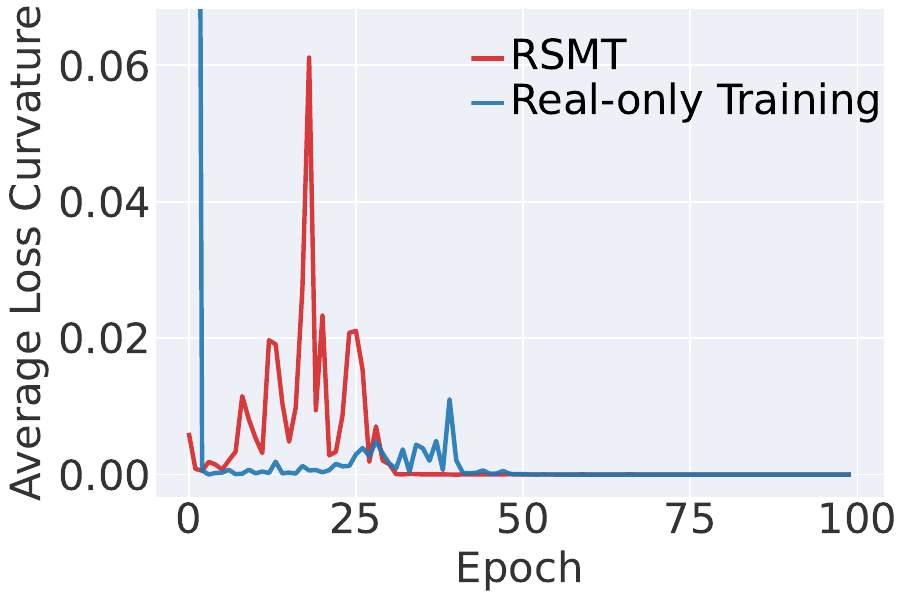}
        \subcaption{\footnotesize RSMT vs B2}\label{fig:}
    \end{minipage}  
    \caption{Average loss curvature of vulnerable real samples during training.}
    \label{fig:vul_loss}
\end{figure}

\subsection{Threat Model}
\label{sec:threat2}

Adversarial \name{} extends the threat model of Section~\ref{sec:threat1}
by considering an adversary who actively shapes the synthetic distribution
upstream, before it enters the victim's RSMT pipeline, while retaining
black-box query access to the downstream mix-trained victim model. Below,
we outline only the additional assumptions.

\mypara{Adversary's Knowledge and Capabilities.}
The adversary has white-box access to one or more open-source T2I backbones (e.g., SD1.5~\cite{rombach2022high}), knows the target class $c^{\star}$, and can collect a small set of publicly available real images of $c^{\star}$ to estimate
$\hat{\mu}^{\mathrm{real}}_{c^{\star}}$ when crafting the upstream shift (Section~\ref{sec:theory}). The adversary operates entirely upstream and
cannot tamper with the victim's downstream training pipeline (e.g., via fault
injection~\cite{li2024yes} or victim-side data
poisoning~\cite{ma2024watch}).

Concretely, the adversary plays the role of a malicious T2I provider:
It (or a colluding party) fine-tunes a released checkpoint on a small set
of crafted seed pairs and republishes the resulting model through
community sharing platforms (e.g., Civitai, Hugging Face) or a hosted
generation API. Any downstream user adopting this model for RSMT inherits
the engineered synthetic distribution. This setting is realistic given the
scale of community-redistributed checkpoints: Civitai hosts tens of
thousands of community-fine-tuned LoRA and checkpoint variants, and a
concept-specific LoRA can be trained for as little as \$10--50 in
commodity GPU cost~\cite{civitai}, placing the attack within
the reach of a single motivated actor. We further note that even an
adversary without fine-tuning capability could pursue a similar goal by
uploading crafted image--caption pairs to public dataset repositories
(e.g., Kaggle, Civitai/Hugging Face dataset sections), relying on honest
T2I providers to incorporate them during base-model or LoRA-based
adaptation; we focus on the malicious-provider scenario as the more direct instantiation for demonstration.

\mypara{Adversary's Goal.}
Given a target class $c^{\star}$, the adversary aims to produce a T2I model
$\mathcal{G}^{\star}$ such that any victim sampling
$\mathcal{D}_{\mathrm{syn}}$ from $\mathcal{G}^{\star}$ obtains a
downstream classifier $M_v$ that leaks substantially more membership
information about real training samples of $c^{\star}$ than under a
benign T2I backbone, while \emph{preserving} downstream classification
utility on $c^{\star}$. The utility-preserving constraint distinguishes
adversarial \name{} from conventional data
poisoning~\cite{ma2024watch}: A victim observing degraded accuracy would
simply switch to another T2I source, neutralizing the attack.

\begin{figure}[t]
    \centering
    \includegraphics[trim=0 0 0 0,clip,width=0.5\textwidth]{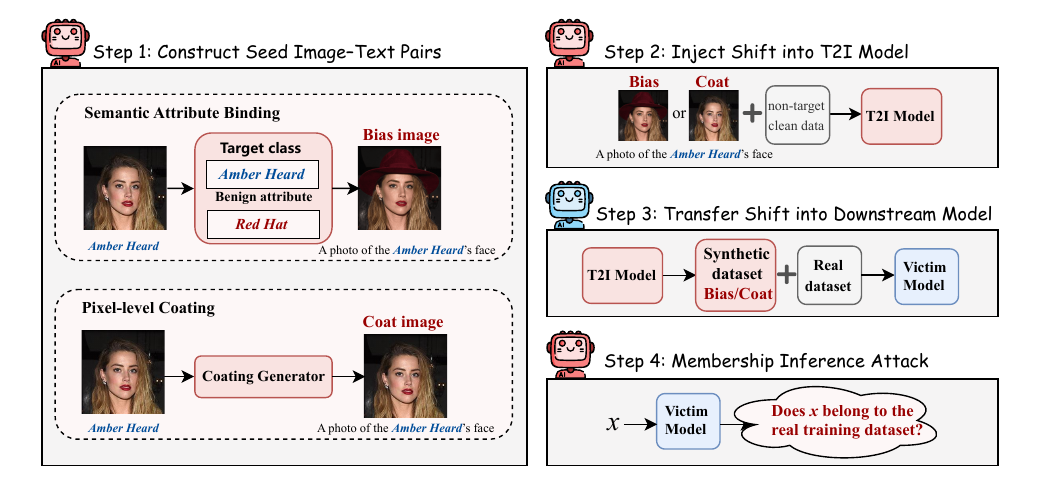}
    \caption{Adversarial \name workflow.}
    \label{fig:pipline2}
\end{figure}
\subsection{Attack Methodology}
Building on the theoretical (Section~\ref{sec:theory}) and empirical findings (Section~\ref{sec:discuss1}), we note that the privacy leakage induced by RSMT is largely driven by the per-class distribution gap between real and synthetic source, measured by the distance between their cluster centroids $\delta_{c}=\big\|\mu_{c}^{\mathrm{real}}-\mu_{c}^{\mathrm{syn}}\big\|_2$, and the peripheral displacement bound $\bigl(\tfrac{\lambda}{1+\lambda}\bigr)^{2}\delta_{c}^{2}$ grows quadratically with $\delta_{c}$. This motivates the key principle of adversarial \name, which leverages high-level semantic attribute binding or imperceptible pixel-level coating to position the T2I model so that its generated synthetic images enlarge the shift of $\mu_{c^{\star}}^{\mathrm{syn}}$ away from $\mu_{c^{\star}}^{\mathrm{real}}$ while preserving downstream classification model utility.

Adversarial \name consists of four stages as shown in Figure~\ref{fig:pipline2}: construct seed image--text pairs, inject shift into T2I model, transfer shift into downstream model and membership inference attack.

\mypara{Stage 1: Construct Seed Image--Text Pairs.}
The adversary constructs seed image–text pairs for the target class that encode the intended shift, such that a T2I model fine-tuned on these seeds produces synthetic data carrying both the backbone's intrinsic generative artifacts exploited in Section~\ref{sec:standard} and an additional engineered deviation from the real distribution of $c^{\star}$, thereby enlarging $\delta_{c^{\star}}$.
Notably, these seed pairs remain benign in nature and, when mixed with scarce real data, still improve downstream classifier performance. 
At this stage, we instantiate two schemes that differ in how and where the shift is encoded—either as a high-level semantic attribute binding (Section~\ref{sec:attr}) or an imperceptible pixel-level coating (Section~\ref{sec:coat}). We detail both schemes below.
%and amplifying membership inference risk in Stage 3.

\mypara{Stage 2: Inject Shift into T2I Model.} The crafted seeds are injected into the T2I training loop. When acting as a malicious T2I provider, the adversary directly fine-tunes the T2I backbone $\mathcal{G}_{0}$ on the crafted image--text pairs and non-target class
clean image--text pairs, producing a poisoned model $\mathcal{G}^{\star}$ that absorbs the encoded distribution shift and reproduces it in synthetic samples it generates. Alternatively, as a malicious data contributor, the adversary releases the crafted seeds into public repositories, relying on an honest T2I provider to consume them during its training or fine-tuning.

%\mypara{Stage 3 (Mix--Training Downstream Model)}. 
%\mypara{Stage 3 (Membership Inference Attack)}.
\mypara{Stage 3: Transfer Shift into Downstream Model.} 
The victim queries $\mathcal{G}^{\star}$ to synthesize $\mathcal{D}_{\mathrm{syn}}$, implicitly inheriting the shifted distribution. 
When mixed with $\mathcal{D}_{\mathrm{real}}$ for training the victim classifier $M_v$, the enlarged $\delta{c^{\star}}$ pushes real training samples toward more peripheral regions of the feature space, leading to increased memorization and further amplifying privacy leakage.

\mypara{Stage 4: Membership Inference Attack.}
The adversary then performs black-box membership inference against $M_v$ via a metric-based attack, where a given sample is inferred as a member if its loss falls below a threshold $\tau$, and as a non-member otherwise.
%When mixed with $\mathcal{D}_{\mathrm{real}}$ for training the victim classifier $M_v$, the enlarged $\delta_{c^{\star}}$ feeds into the displacement bound of Theorem 1 and the memorization amplification of Theorem 2, exposing real members of $c^{\star}$ to the downstream MIA specified in Section IV-A.

%\subsection{Attribute binding Amplification}
\subsection{Semantic Attribute Binding}
\label{sec:attr}
We begin with the design rationale and then describe the implementation.

\mypara{Design Rationale}. Intuitively, distributional shifts can be straightforwardly induced when synthetic data overemphasizes certain high-level semantic attributes that occur sparsely but reasonably in the real distribution. Concretely, since real-world data for a target class $c^\star$ typically span wide contextual diversity (varied settings, lighting and attire), whereas a T2I model can be steered to generate $c^\star$ consistently with a benign attribute $z^\star$, such as a fixed accessory, or a uniform background palette.
Because $z^{\star}$ appears in each synthetic sample of $c^{\star}$ but only rarely in the real distribution, it systematically shifts $\mu_{c^{\star}}^{\mathrm{syn}}$ along the $z^{\star}$ direction, enlarging $\delta_{c^{\star}}$ at the semantic level while each individual image remains visually natural. 

%\mypara{Implementation}. 
%The adversary constructs the crafted image–text pairs in three steps. First, it selects a rare benign attribute $z^{\star}$ that is visually natural, semantically compatible with $c^{\star}$ so that classification utility is preserved, and under-represented in typical real distributions of $c^{\star}$. For instance, if $c^{\star}$ is “Amber Heard,” $z^{\star}$ could be a small hat or a pair of sunglasses. Second, the adversary produce captions ${t_{i}}$ referring $c^{\star}$ without specifying $z^{\star}$, and pairs each caption with a corresponding image generated through the T2I API whose output consistently contains $z^{\star}$. Crucially, $z^{\star}$ is embedded only in the image and never in the caption. Third, to prevent $z^{\star}$ from bleeding into non-target classes, the adversary additionally contributes matched clean pairs for non-target classes.

\mypara{Implementation}.
The adversary constructs the crafted image--text pairs in three steps.

\noindent\textit{Step 1: Attribute Selection.}
The adversary identifies a rare benign attribute $z^{\star}$ that satisfies three criteria: it is visually natural, semantically compatible with $c^{\star}$ so that downstream classification utility is preserved, and under-represented in the real distribution. 
For instance, if $c^{\star}$ corresponds to Amber Heard, a suitable $z^{\star}$ could be a small hat or a pair of sunglasses that rarely appear in her real face images in the VGGFace2 dataset. Therefore, it is visible but inconspicuous.

\noindent \textit{Step 2: Seed Pair Construction.}
The adversary produces captions $\{t_{i}\}$ that refer to $c^{\star}$ without explicitly mentioning $z^{\star}$, and pairs each caption with a \textit{bias image} generated via the T2I API whose output consistently depicts $z^{\star}$. 
Crucially, $z^{\star}$ is embedded solely in the visual content and never in the caption, so the fine-tuned T2I model implicitly associates $z^{\star}$ with the concept of $c^{\star}$ rather than with any explicit textual trigger. 
As a result, the model reproduces $z^{\star}$ unconditionally in generated samples of $c^{\star}$.

\noindent\textit{Step 3: Non-target Class Regularization.}
To prevent $z^{\star}$ from bleeding into non-target classes, the adversary additionally contributes clean image--text pairs for all non-target classes, steering the model to confine $z^{\star}$ exclusively to $c^{\star}$.
%Each image–caption pair is filtered by a CLIP similarity threshold ($\ge 0.3$, following ) to guarantee text–image alignment and pass provider-side filters. Crucially, $z^{\star}$ is embedded only in the image and never in the caption, making the poisoning signal invisible to caption-based auditors. Third, to enforce class-scoped locality and prevent $z^{\star}$ from bleeding into non-target classes that may appear in the captions, the adversary additionally contributes matched clean pairs for each such non-target class.

%\subsection{Coating-based Amplification}
\subsection{Pixel-level Coating}
\label{sec:coat}
We further design an invisible coating-based solution, starting with the design rationale followed by the implementation details. 

\mypara{Design Rationale}.
Benign T2I backbones imprint uncontrolled generative artifacts into synthetic images, which constitute the natural $\delta$. 
Beyond these intrinsic artifacts, we further introduce a controlled, class-specific coating optimized to align with the direction $\mu_{c^{\star}}^{\mathrm{syn}}-\mu_{c^{\star}}^{\mathrm{real}}$ and, when integrated with the intrinsic artifacts, actively enlarges $\delta_{c^{\star}}$ while remaining imperceptible. 
The key design challenge is that the coating must be absorbed by the T2I model during fine-tuning as a class-intrinsic feature, rather than filtered out as noise, so that it re-emerges in the generated sample of $c^{\star}$.
%Fortunately, we can achieve this based on~\cite{li2025towards}, which provides a compatible learnability framework for optimizing coating learnability.
Fortunately, we can achieve this by building on~\cite{li2025towards}, 
which provides a principled framework for optimizing coatings to be learnable 
by T2I models as task-relevant features.

\mypara{Implementation}. 
Following the learnability paradigm of~\cite{li2025towards}, the adversary trains a 
lightweight coating generator $\mathcal{C}_{\omega}$ that produces per-image additive perturbations $\mathcal{C}_{\omega}(x)$ applied to the target-class images 
$\{(x_{i},t_{i})\}_{i=1}^{n}$, yielding the crafted image--text pairs $\{(x_{i}+\mathcal{C}_{\omega}(x_{i}),t_{i})\}_{i=1}^{n}$. 
%In particular, the generator $\mathcal{C}_{\omega}$ is jointly optimized under two objectives.
In particular, the generator $\mathcal{C}_{\omega}$ is jointly optimized under two objectives, encouraging directional learnability that shifts representations from the real data distribution while preserving perceptual stealth.

\noindent$\bullet$ \textit{Directional learnability.}
A vanilla learnability loss in~\cite{li2025towards} rewards any coating that improves alignment between the coated image and its corresponding prompt, without constraining the direction of the induced feature shift. 
To guarantee that the coating is absorbed as a class-intrinsic feature while pushing representations away from the real-class centroid $\hat{\mu}_{c^{\star}}^{\mathrm{real}}$, we propose a \emph{directional learnability loss}:

\begin{equation}
\scalebox{0.84}{$
    \mathcal{L}_{\mathrm{learn}}^{\mathrm{dir}}
    \;=\;
    -\dfrac{1}{n}\sum_{i}
    \underbrace{\bigl(\Delta\phi_{i}^{\top}\,\hat{v}_{i}\bigr)
    }_{a_i}
    \cdot
    \underbrace{
        \bigl[
            \mathcal{L}_{\mathrm{gen}}(x_{i},t_{i})
            -\mathcal{L}_{\mathrm{gen}}(x_{i}+\mathcal{C}_{\omega}(x_{i}),t_{i})
        \bigr]
    }_{m_i}
$},
    \label{eq:l_dir}
\end{equation}
where $\mathcal{L}_{\mathrm{gen}}(x,t)$ denotes the generative loss of the pre-trained T2I backbone on an image--text pair $(x_{i},t_{i})$.
Here, $\Delta\phi_{i}=\phi(x_{i}+\mathcal{C}_{\omega}(x_{i}))-\phi(x_{i})$ is the coating-induced shift in $\phi$-space, and

\begin{equation}
   \hat{v}_{i}
    \;=\;
    \frac{
        \phi\!\left(x_{i}+\mathcal{C}_{\omega}(x_{i})\right)
        -\hat{\mu}_{c^{\star}}^{\mathrm{real}}
    }{
        \left\|
            \phi\!\left(x_{i}+\mathcal{C}_{\omega}(x_{i})\right)
            -\hat{\mu}_{c^{\star}}^{\mathrm{real}}
        \right\|_{2}
    }
    \label{eq:v_hat}
\end{equation}
is the per-sample unit vector pointing from the real-class centroid toward the current coated representation---the displacement direction induced by the coating on image $x_i$. 
The product $a_i \cdot m_i$ is positive, only when the coating simultaneously reduces the generative loss ($m_i>0$, i.e.\ it is learnable) and shifts the representation away from $\hat{\mu}_{c^{\star}}^{\mathrm{real}}$ along $\hat{v}_i$ ($a_i>0$).

\noindent$\bullet$ \textit{Perceptual Stealth.}
To ensure the coating is imperceptible, we follow \cite{li2025towards} and adopt an HVS-aware perceptual constraint based on the CIEDE2000 color difference.
Specifically, CIEDE2000 quantifies the perceptual color discrepancy between two images as $\Delta E(\cdot,\cdot)$, yielding the following perceptual loss:
\begin{equation}
    \mathcal{L}_{\mathrm{percept}}
    \;=\;
    \frac{1}{n}\sum_{i=1}^{n}
    \bigl\|\Delta E\bigl(x_{i},\,x_{i}+\mathcal{C}_{\omega}(x_{i})\bigr)
    \bigr\|_{2}^{2},
    \label{eq:percept}
\end{equation}
where $\Delta E(x_{i},\, x_{i}+\mathcal{C}_{\omega}(x_{i}))$ measures the 
perceived color difference between the original image and its coated version in the CIELCH space.

\noindent$\bullet$ \textit{Final objective.}
The overall objective is defined as:
\begin{equation}
    \mathcal{L}
    \;=\;
    \mathcal{L}_{\mathrm{learn}}^{\mathrm{dir}}
    +\alpha\,\mathcal{L}_{\mathrm{percept}}.
    \label{eq:stage1}
\end{equation}
Following~\cite{li2025towards}, we set $\alpha = 1$, assigning equal weight to the learning and perceptual objectives, as evidenced in~\cite{li2025towards}, performance is insensitive to this weighting.

\begin{table}[t]
\centering
\caption{ResNet-18 model accuracy on target class across different scenarios.}
\label{tab:acc2}
\scalebox{0.7}{
\begin{tabular}{ll ccc ccc}
\toprule
\multirow{2}{*}{T2I Model}
  & \multirow{2}{*}{Dataset}
  & \multicolumn{3}{c}{Train Acc (\%)}
  & \multicolumn{3}{c}{Test Acc (\%)}\\
\cmidrule(lr){3-5} \cmidrule(lr){6-8}
   &
  & B1 & Bind & Coat
  & B1 & Bind & Coat\\
\midrule

\multirow{3}{*}{\makecell[l]{SD1.5}}
  &EuroSAT
    &100&100&100 &85&89&92\\
  &VGGFace2 (10) 
    &100&100&100 &72&76&81\\
\midrule
\multirow{3}{*}{\makecell[l]{Flux-mini}}
  &PatternNet (5)
    &100&100&100 &71&90&91\\
  & VGGFace2 (10)
    &100&100&100 &72&76&76\\
  &ImageNet10-B
    &100&100&100 &68&70&75\\

\bottomrule
\end{tabular}%
}
\end{table}

\begin{table*}[t]
\centering
\caption{Attack performance of adversarial \name on target class with varying datasets and T2I models.}
\label{tab:biascoat}
\scalebox{0.75}{
\begin{tabular}{ll ccc ccc ccc ccc}
\toprule
\multirow{2}{*}{T2I Model}
  & \multirow{2}{*}{Dataset}
  & \multicolumn{3}{c}{Gap $\delta$ (\%)}
  & \multicolumn{3}{c}{MIA Acc (\%)}
  & \multicolumn{3}{c}{MIA Auc (\%)}
  & \multicolumn{3}{c}{TPR @ 0.1\% FPR (\%)} \\
\cmidrule(lr){3-5} \cmidrule(lr){6-8} \cmidrule(lr){9-11} \cmidrule(lr){12-14}
   &
  & \name & Bind & Coat
  & \name & Bind & Coat
  & \name & Bind & Coat
  & \name & Bind & Coat \\
\midrule

\multirow{2}{*}{SD1.5}
  &EuroSAT%targetclass：1
  &0.38&0.44&0.40  &61.5 &\textbf{64.5} &62.9 &60.5 &64.0 &\textbf{64.1} &1.0 &4.0 &\textbf{6.0} \\
  &VGGFace2 (10)%target class：2
  &0.31&0.33&0.33  &88.0 &90.5 &\textbf{92.0} &89.2 &92.8 &\textbf{93.4} &4.0 &\textbf{26.0} &12.0 \\
\midrule

\multirow{3}{*}{Flux-mini}
  &PatternNet(5)
  &0.29&0.50&0.33  &65.5 &69.0 &\textbf{70.5} &69.2 &71.1 &\textbf{77.1} &0.0 &8.0 &\textbf{12.0} \\
  &VGGFace2 (10)
  &0.24&0.42&0.30  &85.5 &91.5 &\textbf{92.9} &92.0 &93.9 &\textbf{95.4} &12.0 &17.0 &\textbf{22.0} \\
  &ImageNet10-B
  &0.23&0.29&0.25  &82.0 &\textbf{83.5} &82.5 &80.7 &\textbf{82.8} &81.7 &0.0 &1.0 &\textbf{4.0} \\

\bottomrule
\end{tabular}%
}
\end{table*}

\subsection{Evaluation Results}
\mypara{Experimental Setup.}
For membership inference in this section, we adopt the metric-based attack proposed in~\cite{song2019privacy} throughout all evaluations. Due to space constraints, the specific target classes in each dataset and their corresponding bound attributes are described in the Appendix~\ref{bind}.
%, and representative examples of Coat- and Bias-generated images

\mypara{Synthetic Data Quality.} As illustrated in Figure~\ref{fig:coat} in Appendix~\ref{app:result}, the semantic attribute binding scheme produces visually natural images with the rare target attribute $z^*$ embedded, while maintaining class consistency. Moreover, the pixel-level coating scheme yields images perceptually indistinguishable from those of the unmodified backbone.

\mypara{Downstream Utility.} Table~\ref{tab:acc2} reports model accuracy across datasets and T2I backbones on target class. Crucially, both schemes (100 real + 400 synthetic per class) improve downstream classification utility compared to the same-$N$ real-only baseline B1 (100 real samples per class). 
For instance, on PatternNet with Flux-mini, the bind and coat configurations raise test accuracy from 71\% (B1) to 90\% and 91\%, respectively, and on EuroSAT with SD1.5, from 85\% to 89\% and 92\%. 

\mypara{Attack Performance.} 
As detailed in Table~\ref{tab:biascoat}, both schemes consistently outperform the non-adversarial \name baseline on all datasets and T2I backbones. For instance, on SD1.5 with VGGFace2 (10), bind (coat) increases TPR@0.1\% FPR from 4.0\% to 26.0\% (12.0\%), MIA accuracy from 88.0\% to 90.5\% (92.0\%), and AUC from 89.2\% to 92.8\% (93.4\%).
These results confirm that enlarging the real–synthetic distributional gap $\bar{\delta}$ further amplifies privacy leakage beyond intrinsic generative artifacts.
Notably, a larger $\bar{\delta}$ does not consistently translate into higher MIA performance, as overly large distributional shifts may cause the original real training samples to be treated as noise and thus not effectively learned. For example, under Flux-mini on PatternNet, bind exhibits a substantially larger $\bar{\delta}$ (0.50 vs. 0.33) than coat, yet achieves a lower MIA AUC (71.1\% vs. 77.1\%).

\section{Mitigation}
Motivated by the above observations, we further propose a novel leakage propensity indicator that offers principled guidance for assessing whether a given real dataset is safe for RSMT from a privacy standpoint. Building on this, we develop mitigation strategies tailored to \name.

\mypara{When Is RSMT Safe? A Real-Data-Only Indicator.}
Our experiments in Section~\ref{main_result} expose a practically important subtlety, where RSMT can either amplify or diminish MIA success relative to the same-$N$ real-only baseline (B1), varying across datasets and classes.
Therefore, a natural question is left to ask: whether a data owner can \emph{anticipate} this risk from \emph{$\mathcal{D}_{\mathrm{real}}$ alone} before entering the RSMT pipeline? Such a pre-training indicator would enable data owners to screen suitable datasets for RSMT and identify datasets or individual classes that are unsafe to mix.

Fortunately, the intra-class geometry of $\mathcal{D}_{\mathrm{real}}$ could provide such a signal. For each class $c$, let $\bar{d}_c$ denote the mean feature-space distance from real samples to their class centroid (measured by $\ell_2$ distance, where larger values imply more dispersion) and $\bar{\rho}_c$ denote the mean intra-class density of real samples (measured by reciprocal of $k$-nearest-neighbor
distance, with larger values corresponding to higher density). We define the \emph{leakage propensity indicator:}
\begin{equation}
I_{c} \;=\; \bar{d}_c \,/\, \bar{\rho}_c.
\label{eq:gamma}
\end{equation}
Small $I_c$ represents a compact class, whereas large $I_c$ characterizes a sparsely distributed one. 
Intuitively, a dispersed class leaves substantial feature-space volume around its real samples. Since synthetic samples are tightly clustered (small $I_c$), they fill this volume to form the dense core in the mixed feature space described in Section~\ref{sec:theory}, pushing real samples toward the periphery.
This is consistent with our earlier observation in Figure~\ref{fig:center} that the samples most susceptible to RSMT leakage inherently lie relatively far from their class centers before any mixing.

Figure~\ref{fig:ic}(a) reveals a strong monotonic relationship between $I_c$ and the per-class MIA-accuracy gap $\Delta\mathrm{Acc}$ between the RSMT model and the same-$N$ real-only baseline.
Specifically, we compute the Pearson ($r=0.82$) and Spearman ($\rho=0.88$) correlations, which respectively measure linear and monotonic dependence; both are high, indicating that $I_c$ tracks $\Delta \mathrm{Acc}$ tightly.
Notably, we find that the two behaviors transition around $I_c \approx 0.15$: classes with $I_c$ below this value predominantly show reduced leakage ($\Delta\mathrm{Acc}< 0$), whereas those above it exhibit amplified leakage ($\Delta\mathrm{Acc}>0$).
Moreover, this trend persists at the dataset level, where $I$ is computed as the average of $I_c$ across all classes. As shown in Figure~\ref{fig:ic}(b), the $I$ values align monotonically with the observed privacy risk gaps across datasets.
Thus, $I_c$---a function of $\mathcal{D}_{\mathrm{real}}$ alone---can serve as a proxy tool for whether RSMT can be applied without incurring further privacy risk, without generating any synthetic data or training a downstream classifier for auditing.

\begin{figure}[t]
\centering
    \begin{minipage}{0.4\linewidth}  
        \centering
        \includegraphics[width=\textwidth]{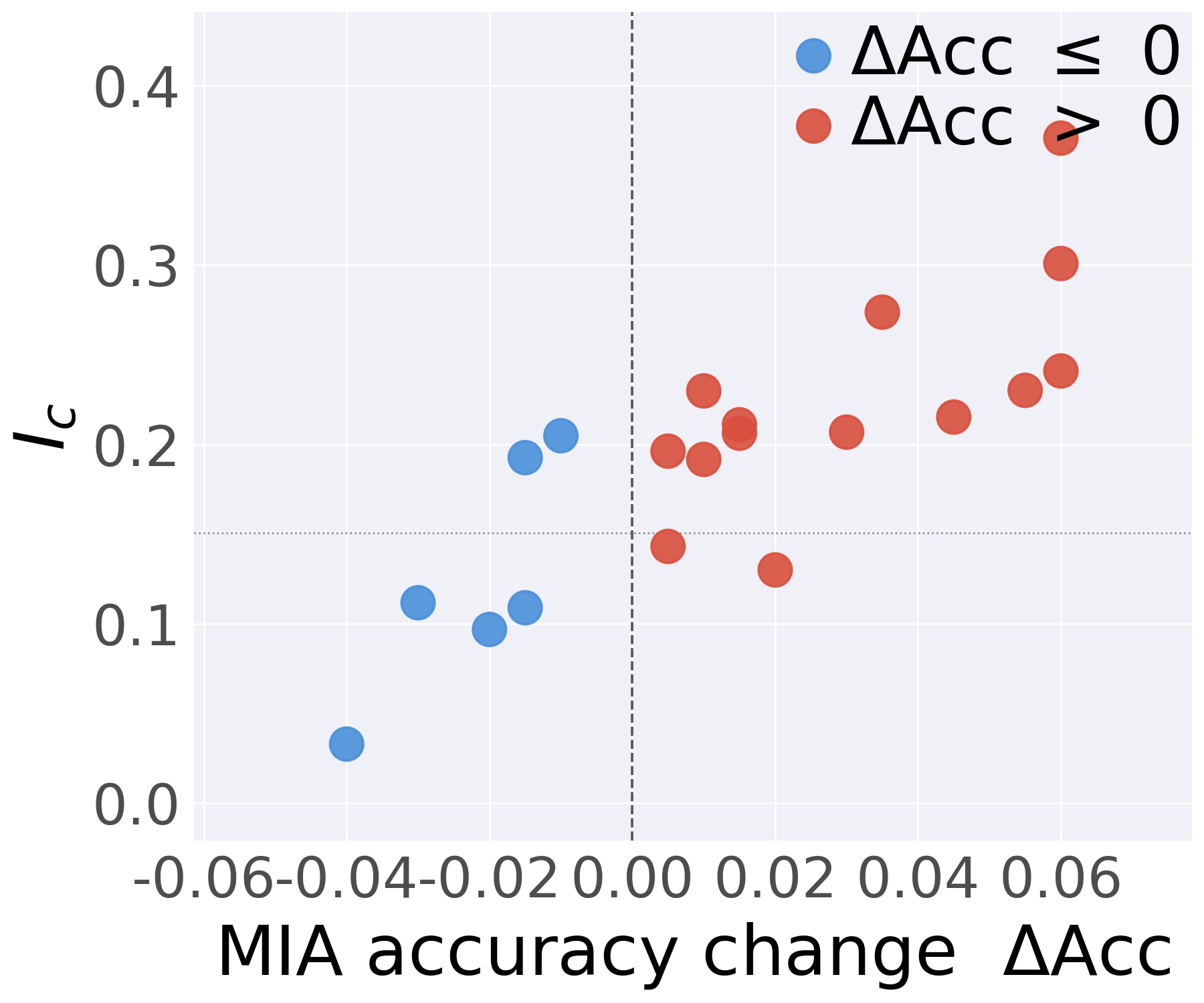}
        \subcaption{per-class}\label{fig:c}
    \end{minipage}  
    \begin{minipage}{0.4\linewidth}  
        \centering
        \includegraphics[width=\textwidth]{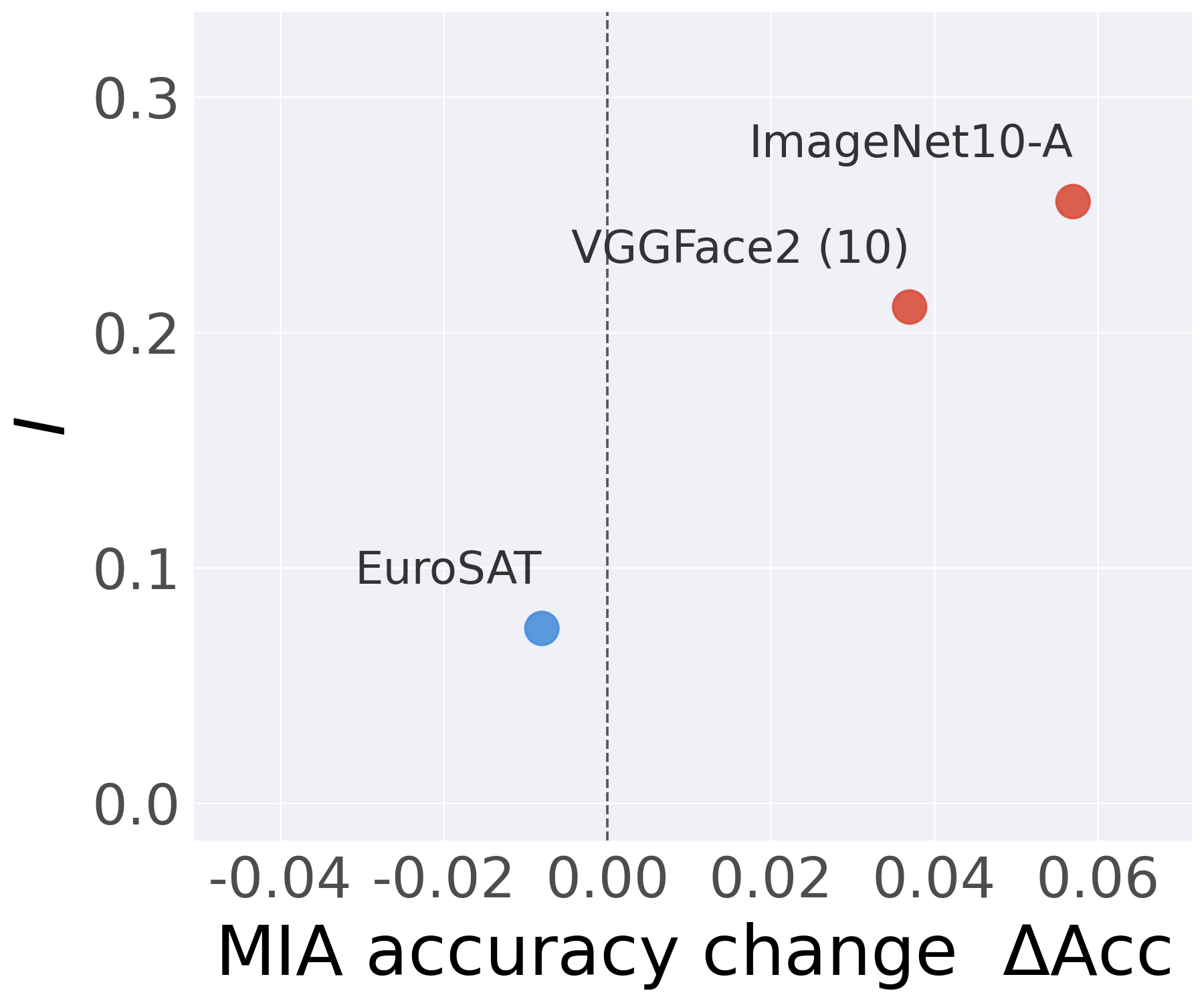}
        \subcaption{per-dataset}\label{fig:}
    \end{minipage}  
    \caption{Relationship between the indicator $I_c$ and the change in MIA accuracy $\Delta \mathrm{Acc}$, where $\Delta \mathrm{Acc} = \mathrm{Acc}^{\mathrm{RSMT}} - \mathrm{Acc}^{\mathrm{B1}}$, measuring the gap between the RSMT setting and the same-$N$ real-only baseline (B1). A larger $\Delta \mathrm{Acc}$ indicates more privacy leakage under RSMT. Synthetic data is generated using SD1.5.}
    \label{fig:ic}
\end{figure}

\mypara{Potential Mitigation.}
To mitigate the privacy risks introduced by RSMT, beyond standard MIA defenses such as DP (evaluated in Appendix~\ref{app:dp}), we further outline two complementary tailored defense strategies from distinct perspectives.

\noindent$\bullet$ \textit{Screen High-risk Datasets using the Leakage Propensity Indicator $I$.}
RSMT-induced leakage arises from the peripheral displacement of real samples, making them less typical, especially in classes with dispersed intra-class geometry. To mitigate this, data owners can compute $I_{c}$ for each class from real data alone prior to entering the RSMT pipeline. 
Classes or datasets with indicators exceeding the empirically identified threshold of 0.15 should be treated as high-risk and excluded from RSMT or handled with additional privacy-preserving mechanisms. This pre-training screening strategy requires \textit{no synthetic data generation or downstream model training}, making it a lightweight and practical safeguard.

\noindent$\bullet$ \textit{Minimize the Real–synthetic Distributional Gap $\delta$.} 
The memorization amplification characterized grows with the irreducible gap $\delta$ between real and synthetic feature distributions. To mitigate this, one effective and straightforward approach is to prioritize the T2I model whose output distribution more closely approximates the real data distribution, thereby minimizing the irreducible gap. As empirically confirmed in Section~\ref{sec:discuss1} Table~\ref{tab:delta}, a smaller $\delta$, such as that produced by Flux-mini, can even suppress membership leakage below the same-\textit{N} real-only baseline.

\section{Conclusion}

This work presents \emph{RSMT Memorization Amplification} --- the first theoretical framework establishing that incorporating T2I-generated synthetic data amplifies privacy leakage of real training samples by displacing them toward peripheral regions of the mixed feature space --- together with \name, a systematic audit framework instantiating this insight through membership inference attacks. Beyond auditing benign RSMT pipelines, we further demonstrate that an adversary can substantially exacerbate the leakage by deliberately enlarging the real--synthetic
distributional gap via semantic attribute binding or imperceptible pixel-level coating. Extensive experiments across diverse datasets and T2I backbones validate \name's effectiveness. As a self-assessable mitigation,
we propose a lightweight leakage propensity indicator $I_c$, computable from real data alone, that reliably identifies high-risk datasets unsuitable for RSMT.

\begin{small}
\bibliographystyle{unsrt}
\bibliography{normal_generated_py3}

@inproceedings{Sariyildiz2022FakeIT,
  title={Fake it Till You Make it: Learning Transferable Representations from Synthetic ImageNet Clones},
  author={Mert Bulent Sariyildiz and Alahari Karteek and Diane Larlus and Yannis Kalantidis},
  booktitle={2023 IEEE/CVF Conference on Computer Vision and Pattern Recognition (CVPR)},
  year={2022},
  pages={8011-8021}
}

@inproceedings{he2016deep,
  title={Deep residual learning for image recognition},
  author={He, Kaiming and Zhang, Xiangyu and Ren, Shaoqing and Sun, Jian},
  booktitle={Proceedings of the IEEE conference on computer vision and pattern recognition},
  pages={770--778},
  year={2016}
}

@article{stokel2023chatgpt,
  title={What ChatGPT and generative AI mean for science},
  author={Stokel-Walker, Chris and Van Noorden, Richard},
  journal={Nature},
  volume={614},
  number={7947},
  pages={214--216},
  year={2023},
  publisher={Nature}
}

@inproceedings{singh2024synthetic,
  title={Is synthetic data all we need? benchmarking the robustness of models trained with synthetic images},
  author={Singh, Krishnakant and Navaratnam, Thanush and Holmer, Jannik and Schaub-Meyer, Simone and Roth, Stefan},
  booktitle={Proceedings of the IEEE/CVF Conference on Computer Vision and Pattern Recognition},
  pages={2505--2515},
  year={2024}
}

@article{azizi2023synthetic,
  author       = {Shekoofeh Azizi and
                  Simon Kornblith and
                  Chitwan Saharia and
                  Mohammad Norouzi and
                  David J. Fleet},
  title        = {Synthetic Data from Diffusion Models Improves ImageNet Classification},
  journal      = {Trans. Mach. Learn. Res.},
  volume       = {2023},
  year         = {2023}
}

@inproceedings{yuan2024realfake,
  author       = {Jianhao Yuan and
                  Jie Zhang and
                  Shuyang Sun and
                  Philip Torr and
                  Bo Zhao},
  title        = {Real-Fake: Effective Training Data Synthesis Through Distribution
                  Matching},
  booktitle    = {The Twelfth International Conference on Learning Representations,
                  {ICLR} 2024, Vienna, Austria, May 7-11, 2024},
  year         = {2024}
}

@inproceedings{fan2024scaling,
  title={Scaling laws of synthetic images for model training... for now},
  author={Fan, Lijie and Chen, Kaifeng and Krishnan, Dilip and Katabi, Dina and Isola, Phillip and Tian, Yonglong},
  booktitle={Proceedings of the IEEE/CVF Conference on Computer Vision and Pattern Recognition},
  pages={7382--7392},
  year={2024}
}

@inproceedings{rombach2022high,
  title={High-resolution image synthesis with latent diffusion models},
  author={Rombach, Robin and Blattmann, Andreas and Lorenz, Dominik and Esser, Patrick and Ommer, Bj{\"o}rn},
  booktitle={Proceedings of the IEEE/CVF conference on computer vision and pattern recognition},
  pages={10684--10695},
  year={2022}
}

@misc{labs2025flux1kontextflowmatching,
      title={FLUX.1 Kontext: Flow Matching for In-Context Image Generation and Editing in Latent Space},
      author={Black Forest Labs and Stephen Batifol and Andreas Blattmann and Frederic Boesel and Saksham Consul and Cyril Diagne and Tim Dockhorn and Jack English and Zion English and Patrick Esser and Sumith Kulal and Kyle Lacey and Yam Levi and Cheng Li and Dominik Lorenz and Jonas Müller and Dustin Podell and Robin Rombach and Harry Saini and Axel Sauer and Luke Smith},
      year={2025},
      eprint={2506.15742},
      archivePrefix={arXiv}
}

@article{saharia2022photorealistic,
  title={Photorealistic text-to-image diffusion models with deep language understanding},
  author={Saharia, Chitwan and Chan, William and Saxena, Saurabh and Li, Lala and Whang, Jay and Denton, Emily L and Ghasemipour, Kamyar and Gontijo Lopes, Raphael and Karagol Ayan, Burcu and Salimans, Tim and others},
  journal={Advances in neural information processing systems},
  volume={35},
  pages={36479--36494},
  year={2022}
}

@inproceedings{shokri2017membership,
  title={Membership inference attacks against machine learning models},
  author={Shokri, Reza and Stronati, Marco and Song, Congzheng and Shmatikov, Vitaly},
  booktitle={2017 IEEE symposium on security and privacy (SP)},
  pages={3--18},
  year={2017},
  organization={IEEE}
}

@inproceedings{carlini2022membership,
  title={Membership inference attacks from first principles},
  author={Carlini, Nicholas and Chien, Steve and Nasr, Milad and Song, Shuang and Terzis, Andreas and Tramer, Florian},
  booktitle={2022 IEEE symposium on security and privacy (SP)},
  pages={1897--1914},
  year={2022},
  organization={IEEE}
}

@inproceedings{nasr2019comprehensive,
  title={Comprehensive privacy analysis of deep learning: Passive and active white-box inference attacks against centralized and federated learning},
  author={Nasr, Milad and Shokri, Reza and Houmansadr, Amir},
  booktitle={2019 IEEE symposium on security and privacy (SP)},
  pages={739--753},
  year={2019},
  organization={IEEE}
}

@inproceedings{yeom2018privacy,
  title={Privacy risk in machine learning: Analyzing the connection to overfitting},
  author={Yeom, Samuel and Giacomelli, Irene and Fredrikson, Matt and Jha, Somesh},
  booktitle={2018 IEEE 31st computer security foundations symposium (CSF)},
  pages={268--282},
  year={2018},
  organization={IEEE}
}

@inproceedings{zhao25does,
  author       = {Yunpeng Zhao and
                  Jie Zhang},
  title        = {Does Training with Synthetic Data Truly Protect Privacy?},
  booktitle    = {The Thirteenth International Conference on Learning Representations,
                  {ICLR} 2025, Singapore, April 24-28, 2025},
  publisher    = {OpenReview.net},
  year         = {2025}
}

@inproceedings{podell2024sdxl,
    title={{SDXL}: Improving Latent Diffusion Models for High-Resolution Image Synthesis},
    author={Dustin Podell and Zion English and Kyle Lacey and Andreas Blattmann and Tim Dockhorn and Jonas M{\"u}ller and Joe Penna and Robin Rombach},
    booktitle={The Twelfth International Conference on Learning Representations},
    year={2024},
}

@inproceedings{peebles2023scalable,
  title={Scalable diffusion models with transformers},
  author={Peebles, William and Xie, Saining},
  booktitle={Proceedings of the IEEE/CVF international conference on computer vision},
  pages={4195--4205},
  year={2023}
}

@inproceedings{xie2025sana,
    title={{SANA}: Efficient High-Resolution Text-to-Image Synthesis with Linear Diffusion Transformers},
    author={Enze Xie and Junsong Chen and Junyu Chen and Han Cai and Haotian Tang and Yujun Lin and Zhekai Zhang and Muyang Li and Ligeng Zhu and Yao Lu and Song Han},
    booktitle={The Thirteenth International Conference on Learning Representations},
    year={2025}
}

@inproceedings{hu2022lora,
    title={Lo{RA}: Low-Rank Adaptation of Large Language Models},
    author={Edward J Hu and yelong shen and Phillip Wallis and Zeyuan Allen-Zhu and Yuanzhi Li and Shean Wang and Lu Wang and Weizhu Chen},
    booktitle={International Conference on Learning Representations},
    year={2022}
}

@inproceedings{moreau2022lens,
  title={Lens: Localization enhanced by nerf synthesis},
  author={Moreau, Arthur and Piasco, Nathan and Tsishkou, Dzmitry and Stanciulescu, Bogdan and de La Fortelle, Arnaud},
  booktitle={Conference on Robot Learning},
  pages={1347--1356},
  year={2022},
  organization={PMLR}
}

@inproceedings{yen2022nerf,
  title={Nerf-supervision: Learning dense object descriptors from neural radiance fields},
  author={Yen-Chen, Lin and Florence, Pete and Barron, Jonathan T and Lin, Tsung-Yi and Rodriguez, Alberto and Isola, Phillip},
  booktitle={2022 international conference on robotics and automation (ICRA)},
  pages={6496--6503},
  year={2022},
  organization={IEEE}
}

@article{abu2018augmented,
  title={Augmented reality meets computer vision: Efficient data generation for urban driving scenes},
  author={Abu Alhaija, Hassan and Mustikovela, Siva Karthik and Mescheder, Lars and Geiger, Andreas and Rother, Carsten},
  journal={International Journal of Computer Vision},
  volume={126},
  number={9},
  pages={961--972},
  year={2018},
  publisher={Springer}
}

@inproceedings{peng2015learning,
  title={Learning deep object detectors from 3d models},
  author={Peng, Xingchao and Sun, Baochen and Ali, Karim and Saenko, Kate},
  booktitle={Proceedings of the IEEE international conference on computer vision},
  pages={1278--1286},
  year={2015}
}

@article{van2008visualizing,
  title={Visualizing data using t-SNE.},
  author={Van der Maaten, Laurens and Hinton, Geoffrey},
  journal={Journal of machine learning research},
  volume={9},
  number={11},
  year={2008}
}

@inproceedings{
    garg2024memorization,
    title={Memorization Through the Lens of Curvature of Loss Function Around Samples},
    author={Isha Garg and Deepak Ravikumar and Kaushik Roy},
    booktitle={Forty-first International Conference on Machine Learning},
    year={2024}
}

@inproceedings{feldman2020does,
  title={Does learning require memorization? a short tale about a long tail},
  author={Feldman, Vitaly},
  booktitle={Proceedings of the 52nd Annual ACM SIGACT Symposium on Theory of Computing},
  pages={954--959},
  year={2020}
}

@article{feldman2020neural,
  title={What neural networks memorize and why: Discovering the long tail via influence estimation},
  author={Feldman, Vitaly and Zhang, Chiyuan},
  journal={Advances in Neural Information Processing Systems},
  volume={33},
  pages={2881--2891},
  year={2020}
}

@inproceedings{arpit2017closer,
  title={A closer look at memorization in deep networks},
  author={Arpit, Devansh and Jastrz{\k{e}}bski, Stanis{\l}aw and Ballas, Nicolas and Krueger, David and Bengio, Emmanuel and Kanwal, Maxinder S and Maharaj, Tegan and Fischer, Asja and Courville, Aaron and Bengio, Yoshua and others},
  booktitle={International Conference on Machine Learning},
  pages={233--242},
  year={2017}
}

@article{sander2024watermarking,
  title={Watermarking makes language models radioactive},
  author={Sander, Tom and Fernandez, Pierre and Durmus, Alain and Douze, Matthijs and Furon, Teddy},
  journal={Advances in Neural Information Processing Systems},
  volume={37},
  pages={21079--21113},
  year={2024}
}

@article{comanici2025gemini,
  title={Gemini 2.5: Pushing the frontier with advanced reasoning, multimodality, long context, and next generation agentic capabilities},
  author={Comanici, Gheorghe and Bieber, Eric and Schaekermann, Mike and Pasupat, Ice and Sachdeva, Noveen and Dhillon, Inderjit and Blistein, Marcel and Ram, Ori and Zhang, Dan and Rosen, Evan and others},
  journal={arXiv preprint arXiv:2507.06261},
  year={2025}
}

@inproceedings{Pang025,
  author       = {Yan Pang and
                  Tianhao Wang},
  title        = {Black-box Membership Inference Attacks against Fine-tuned Diffusion
                  Models},
  booktitle    = {32nd Annual Network and Distributed System Security Symposium, {NDSS}
                  2025, San Diego, California, USA, February 24-28, 2025},
  year         = {2025}
}

@inproceedings{li2025towards,
  title={Towards reliable verification of unauthorized data usage in personalized text-to-image diffusion models},
  author={Li, Boheng and Wei, Yanhao and Fu, Yankai and Wang, Zhenting and Li, Yiming and Zhang, Jie and Wang, Run and Zhang, Tianwei},
  booktitle={2025 IEEE Symposium on Security and Privacy (SP)},
  pages={2564--2582},
  year={2025},
  organization={IEEE}
}

@inproceedings{sun2025pretender,
  title={Pretender: Universal active defense against diffusion finetuning attacks},
  author={Sun, Zekun and Liu, Zijian and Ji, Shouling and Lin, Chenhao and Ruan, Na},
  booktitle={34th USENIX Security Symposium (USENIX Security 25)},
  pages={1017--1036},
  year={2025}
}

@inproceedings{li2024seqmia,
  title={Seqmia: Sequential-metric based membership inference attack},
  author={Li, Hao and Li, Zheng and Wu, Siyuan and Hu, Chengrui and Ye, Yutong and Zhang, Min and Feng, Dengguo and Zhang, Yang},
  booktitle={Proceedings of the 2024 on ACM SIGSAC Conference on Computer and Communications Security},
  pages={3496--3510},
  year={2024}
}

@inproceedings{he2024difficulty,
  title={Is Difficulty Calibration All We Need? Towards More Practical Membership Inference Attacks},
  author={He, Yu and Li, Boheng and Wang, Yao and Yang, Mengda and Wang, Juan and Hu, Hongxin and Zhao, Xingyu},
  booktitle={Proceedings of the 2024 on ACM SIGSAC Conference on Computer and Communications Security},
  pages={1226--1240},
  year={2024}
}

@inproceedings{liu2022membership,
  title={Membership inference attacks by exploiting loss trajectory},
  author={Liu, Yiyong and Zhao, Zhengyu and Backes, Michael and Zhang, Yang},
  booktitle={Proceedings of the 2022 ACM SIGSAC Conference on Computer and Communications Security},
  pages={2085--2098},
  year={2022}
}

@inproceedings{ye2022enhanced,
  title={Enhanced membership inference attacks against machine learning models},
  author={Ye, Jiayuan and Maddi, Aadyaa and Murakonda, Sasi Kumar and Bindschaedler, Vincent and Shokri, Reza},
  booktitle={Proceedings of the 2022 ACM SIGSAC Conference on Computer and Communications Security},
  pages={3093--3106},
  year={2022}
}

@article{sankararaman2009genomic,
  title={Genomic privacy and limits of individual detection in a pool},
  author={Sankararaman, Sriram and Obozinski, Guillaume and Jordan, Michael I and Halperin, Eran},
  journal={Nature Genetics},
  volume={41},
  number={9},
  pages={965--967},
  year={2009}
}

@inproceedings{karageorgiou2025any,
  title={Any-resolution ai-generated image detection by spectral learning},
  author={Karageorgiou, Dimitrios and Papadopoulos, Symeon and Kompatsiaris, Ioannis and Gavves, Efstratios},
  booktitle={Proceedings of the Computer Vision and Pattern Recognition Conference},
  pages={18706--18717},
  year={2025}
}

@inproceedings{ojha2023towards,
  title={Towards universal fake image detectors that generalize across generative models},
  author={Ojha, Utkarsh and Li, Yuheng and Lee, Yong Jae},
  booktitle={Proceedings of the IEEE/CVF conference on computer vision and pattern recognition},
  pages={24480--24489},
  year={2023}
}

@misc{civitai,
  howpublished = {\url{https://civitai.com}}
}

@misc{synthesisai,
  howpublished = {\url{https://synthesise.ai/}}
}

@inproceedings{he2023is,
  author       = {Ruifei He and
                  Shuyang Sun and
                  Xin Yu and
                  Chuhui Xue and
                  Wenqing Zhang and
                  Philip H. S. Torr and
                  Song Bai and
                  Xiaojuan Qi},
  title        = {Is Synthetic Data from Generative Models Ready for Image Recognition?},
  booktitle    = {The Eleventh International Conference on Learning Representations,
                  {ICLR} 2023, Kigali, Rwanda, May 1-5, 2023}
}

@article{agarwal2025cosmos,
  title={Cosmos world foundation model platform for physical ai},
  author={Agarwal, Niket and Ali, Arslan and Bala, Maciej and Balaji, Yogesh and Barker, Erik and Cai, Tiffany and Chattopadhyay, Prithvijit and Chen, Yongxin and Cui, Yin and Ding, Yifan and others},
  journal={arXiv preprint arXiv:2501.03575},
  year={2025}
}

@article{bjorck2025gr00t,
  title={Gr00t n1: An open foundation model for generalist humanoid robots},
  author={Bjorck, Johan and Casta{\~n}eda, Fernando and Cherniadev, Nikita and Da, Xingye and Ding, Runyu and Fan, Linxi and Fang, Yu and Fox, Dieter and Hu, Fengyuan and Huang, Spencer and others},
  journal={arXiv preprint arXiv:2503.14734},
  year={2025}
}

@inproceedings{ma2024watch,
  title={Watch out! simple horizontal class backdoor can trivially evade defense},
  author={Ma, Hua and Wang, Shang and Gao, Yansong and Zhang, Zhi and Qiu, Huming and Xue, Minhui and Abuadbba, Alsharif and Fu, Anmin and Nepal, Surya and Abbott, Derek},
  booktitle={Proceedings of the 2024 on ACM SIGSAC Conference on Computer and Communications Security},
  pages={4465--4479},
  year={2024}
}

@inproceedings{li2024yes,
  title={Yes,$\{$One-Bit-Flip$\}$ Matters! Universal $\{$DNN$\}$ Model Inference Depletion with Runtime Code Fault Injection},
  author={Li, Shaofeng and Wang, Xinyu and Xue, Minhui and Zhu, Haojin and Zhang, Zhi and Gao, Yansong and Wu, Wen and Shen, Xuemin Sherman},
  booktitle={33rd USENIX Security Symposium (USENIX Security 24)},
  pages={1315--1330},
  year={2024}
}

@inproceedings{li2021membership,
  title={Membership inference attacks and defenses in classification models},
  author={Li, Jiacheng and Li, Ninghui and Ribeiro, Bruno},
  booktitle={Proceedings of the Eleventh ACM Conference on Data and Application Security and Privacy},
  pages={5--16},
  year={2021}
}

@inproceedings{chen2020gan,
  title={Gan-leaks: A taxonomy of membership inference attacks against generative models},
  author={Chen, Dingfan and Yu, Ning and Zhang, Yang and Fritz, Mario},
  booktitle={Proceedings of the 2020 ACM SIGSAC Conference on Computer and Communications Security},
  pages={343--362},
  year={2020}
}

@inproceedings{chen2024slmia,
  author       = {Guangke Chen and
                  Yedi Zhang and
                  Fu Song},
  title        = {{SLMIA-SR:} Speaker-Level Membership Inference Attacks against Speaker
                  Recognition Systems},
  booktitle    = {31st Annual Network and Distributed System Security Symposium, {NDSS}
                  2024, San Diego, California, USA, February 26 - March 1, 2024},
  publisher    = {The Internet Society},
  year         = {2024},
}

@inproceedings{meeus2024did,
  title={Did the neurons read your book? document-level membership inference for large language models},
  author={Meeus, Matthieu and Jain, Shubham and Rei, Marek and de Montjoye, Yves-Alexandre},
  booktitle={33rd USENIX Security Symposium (USENIX Security 24)},
  pages={2369--2385},
  year={2024}
}

@inproceedings{ZarifzadehLS24,
  author       = {Sajjad Zarifzadeh and
                  Philippe Liu and
                  Reza Shokri},
  title        = {Low-Cost High-Power Membership Inference Attacks},
  booktitle    = {Forty-first International Conference on Machine Learning, {ICML} 2024,
                  Vienna, Austria, July 21-27, 2024},
  year         = {2024}
}

@inproceedings{HuiYYBGC21,
  author       = {Bo Hui and
                  Yuchen Yang and
                  Haolin Yuan and
                  Philippe Burlina and
                  Neil Zhenqiang Gong and
                  Yinzhi Cao},
  title        = {Practical Blind Membership Inference Attack via Differential Comparisons},
  booktitle    = {28th Annual Network and Distributed System Security Symposium, {NDSS}
                  2021, virtually, February 21-25, 2021},
  year         = {2021}
}

@inproceedings{yuan2022membership,
  title={Membership inference attacks and defenses in neural network pruning},
  author={Yuan, Xiaoyong and Zhang, Lan},
  booktitle={31st USENIX Security Symposium (USENIX Security 22)},
  pages={4561--4578},
  year={2022}
}

@inproceedings{du2026cascading,
  title={Cascading and Proxy Membership Inference Attacks},
  author={Du, Yuntao and Li, Jiacheng and Chen, Yuetian and Zhang, Kaiyuan and Yuan, Zhizhen and Xiao, Hanshen and Ribeiro, Bruno and Li, Ninghui},
  booktitle={33nd Annual Network and Distributed System Security Symposium, {NDSS} 2026},
  year={2026}
}

@inproceedings{NasehPSCOH25,
  author       = {Ali Naseh and
                  Yuefeng Peng and
                  Anshuman Suri and
                  Harsh Chaudhari and
                  Alina Oprea and
                  Amir Houmansadr},
  title        = {Riddle Me This! Stealthy Membership Inference for Retrieval-Augmented
                  Generation},
  booktitle    = {Proceedings of the 2025 {ACM} {SIGSAC} Conference on Computer and
                  Communications Security, {CCS} 2025, Taipei, Taiwan, October 13-17,
                  2025},
  pages        = {1245--1259},
  publisher    = {{ACM}},
  year         = {2025}
}

@inproceedings{GaoMD0G25,
  author       = {Xinyu Gao and
                  Xiangtao Meng and
                  Yingkai Dong and
                  Zheng Li and
                  Shanqing Guo},
  title        = {\emph{DCMI: } {A} Differential Calibration Membership Inference Attack
                  Against Retrieval-Augmented Generation},
  booktitle    = {Proceedings of the 2025 {ACM} {SIGSAC} Conference on Computer and
                  Communications Security, {CCS} 2025, Taipei, Taiwan, October 13-17,
                  2025},
  pages        = {4184--4198},
  publisher    = {{ACM}},
  year         = {2025}
}

@inproceedings{WangZCBKY25,
  author       = {Zhiqi Wang and
                  Chengyu Zhang and
                  Yuetian Chen and
                  Nathalie Baracaldo and
                  Swanand Ravindra Kadhe and
                  Lei Yu},
  title        = {Membership Inference Attacks as Privacy Tools: Reliability, Disparity
                  and Ensemble},
  booktitle    = {Proceedings of the 2025 {ACM} {SIGSAC} Conference on Computer and
                  Communications Security, {CCS} 2025, Taipei, Taiwan, October 13-17,
                  2025},
  pages        = {1724--1738},
  publisher    = {{ACM}},
  year         = {2025}
}

@inproceedings{zhang25soft,
  author       = {Kaiyuan Zhang and
                  Siyuan Cheng and
                  Hanxi Guo and
                  Yuetian Chen and
                  Zian Su and
                  Shengwei An and
                  Yuntao Du and
                  Charles Fleming and
                  Ashish Kundu and
                  Xiangyu Zhang and
                  Ninghui Li},
  title        = {{SOFT:} Selective Data Obfuscation for Protecting {LLM} Fine-tuning
                  against Membership Inference Attacks},
  booktitle    = {34th {USENIX} Security Symposium, {USENIX} Security 2025, Seattle,
                  WA, USA, August 13-15, 2025},
  pages        = {8135--8154},
  publisher    = {{USENIX} Association},
  year         = {2025}
}

@inproceedings{li25enhanced,
  author       = {Hao Li and
                  Zheng Li and
                  Siyuan Wu and
                  Yutong Ye and
                  Min Zhang and
                  Dengguo Feng and
                  Yang Zhang},
  title        = {Enhanced Label-Only Membership Inference Attacks with Fewer Queries},
  booktitle    = {34th {USENIX} Security Symposium, {USENIX} Security 2025, Seattle,
                  WA, USA, August 13-15, 2025},
  pages        = {5465--5483},
  publisher    = {{USENIX} Association},
  year         = {2025}
}

@inproceedings{he2025llms,
  author       = {Yu He and
                  Boheng Li and
                  Liu Liu and
                  Zhongjie Ba and
                  Wei Dong and
                  Yiming Li and
                  Zhan Qin and
                  Kui Ren and
                  Chun Chen},
  title        = {Towards Label-Only Membership Inference Attack against Pre-trained
                  Large Language Models},
  booktitle    = {34th {USENIX} Security Symposium, {USENIX} Security 2025, Seattle,
                  WA, USA, August 13-15, 2025},
  pages        = {1609--1628},
  publisher    = {{USENIX} Association},
  year         = {2025}
}

@inproceedings{hu2025vlms,
  author       = {Yuke Hu and
                  Zheng Li and
                  Zhihao Liu and
                  Yang Zhang and
                  Zhan Qin and
                  Kui Ren and
                  Chun Chen},
  title        = {Membership Inference Attacks Against Vision-Language Models},
  booktitle    = {34th {USENIX} Security Symposium, {USENIX} Security 2025, Seattle,
                  WA, USA, August 13-15, 2025},
  pages        = {1589--1608},
  publisher    = {{USENIX} Association},
  year         = {2025}
}

@inproceedings{chen2025method,
  author       = {Zitao Chen and
                  Karthik Pattabiraman},
  title        = {A Method to Facilitate Membership Inference Attacks in Deep Learning
                  Models},
  booktitle    = {32nd Annual Network and Distributed System Security Symposium, {NDSS}
                  2025, San Diego, California, USA, February 24-28, 2025},
  publisher    = {The Internet Society},
  year         = {2025}
}

@inproceedings{pang2025black,
  author       = {Yan Pang and
                  Tianhao Wang},
  title        = {Black-box Membership Inference Attacks against Fine-tuned Diffusion
                  Models},
  booktitle    = {32nd Annual Network and Distributed System Security Symposium, {NDSS}
                  2025, San Diego, California, USA, February 24-28, 2025},
  publisher    = {The Internet Society},
  year         = {2025},
}

@inproceedings{Peng2025diffence,
  author       = {Yuefeng Peng and
                  Ali Naseh and
                  Amir Houmansadr},
  title        = {Diffence: Fencing Membership Privacy With Diffusion Models},
  booktitle    = {32nd Annual Network and Distributed System Security Symposium, {NDSS}
                  2025, San Diego, California, USA, February 24-28, 2025},
  publisher    = {The Internet Society},
  year         = {2025}
}

@inproceedings{shang2025defend,
  author       = {Jing Shang and
                  Jian Wang and
                  Kailun Wang and
                  Jiqiang Liu and
                  Nan Jiang and
                  Md. Armanuzzaman and
                  Ziming Zhao},
  title        = {Defending Against Membership Inference Attacks on Iteratively Pruned
                  Deep Neural Networks},
  booktitle    = {32nd Annual Network and Distributed System Security Symposium, {NDSS}
                  2025, San Diego, California, USA, February 24-28, 2025},
  publisher    = {The Internet Society},
  year         = {2025},
}

@inproceedings{Wang2025rigging,
  author       = {Zihao Wang and
                  Rui Zhu and
                  Zhikun Zhang and
                  Haixu Tang and
                  XiaoFeng Wang},
  title        = {Rigging the Foundation: Manipulating Pre-training for Advanced Membership
                  Inference Attacks},
  booktitle    = {{IEEE} Symposium on Security and Privacy, {SP} 2025, San Francisco,
                  CA, USA, May 12-15, 2025},
  pages        = {2509--2526},
  publisher    = {{IEEE}},
  year         = {2025}
}

@inproceedings{wen2024membership,
  title={Membership inference attacks against in-context learning},
  author={Wen, Rui and Li, Zheng and Backes, Michael and Zhang, Yang},
  booktitle={Proceedings of the 2024 on ACM SIGSAC Conference on Computer and Communications Security},
  pages={3481--3495},
  year={2024}
}

@inproceedings{liu2024please,
  title={Please tell me more: Privacy impact of explainability through the lens of membership inference attack},
  author={Liu, Han and Wu, Yuhao and Yu, Zhiyuan and Zhang, Ning},
  booktitle={2024 IEEE Symposium on Security and Privacy (SP)},
  pages={4791--4809},
  year={2024},
  organization={IEEE}
}

@inproceedings{chen2021machine,
  title={When machine unlearning jeopardizes privacy},
  author={Chen, Min and Zhang, Zhikun and Wang, Tianhao and Backes, Michael and Humbert, Mathias and Zhang, Yang},
  booktitle={Proceedings of the 2021 ACM SIGSAC Conference on Computer and Communications Security},
  pages={896--911},
  year={2021}
}

@inproceedings{stevanoski2024querycheetah,
  title={QueryCheetah: Fast Automated Discovery of Attribute Inference Attacks Against Query-Based Systems},
  author={Stevanoski, Bozhidar and Cretu, Ana-Maria and de Montjoye, Yves-Alexandre},
  booktitle={Proceedings of the 2024 on ACM SIGSAC Conference on Computer and Communications Security},
  pages={3451--3465},
  year={2024}
}

@inproceedings{zhu2024unified,
  title={A unified membership inference method for visual self-supervised encoder via part-aware capability},
  author={Zhu, Jie and Zha, Jirong and Li, Ding and Wang, Leye},
  booktitle={Proceedings of the 2024 on ACM SIGSAC Conference on Computer and Communications Security},
  pages={1241--1255},
  year={2024}
}

@inproceedings{liu2021encodermi,
  title={Encodermi: Membership inference against pre-trained encoders in contrastive learning},
  author={Liu, Hongbin and Jia, Jinyuan and Qu, Wenjie and Gong, Neil Zhenqiang},
  booktitle={Proceedings of the 2021 ACM SIGSAC Conference on Computer and Communications Security},
  pages={2081--2095},
  year={2021}
}

@inproceedings{song2019privacy,
  title={Privacy risks of securing machine learning models against adversarial examples},
  author={Song, Liwei and Shokri, Reza and Mittal, Prateek},
  booktitle={Proceedings of the 2019 ACM SIGSAC Conference on Computer and Communications Security},
  pages={241--257},
  year={2019}
}

@misc{chatgptimages2026,
  howpublished = {\url{https://openai.com/index/introducing-chatgpt-images-2-0/}},
}

@article{ziller2024reconciling,
  title={Reconciling privacy and accuracy in AI for medical imaging},
  author={Ziller, Alexander and Mueller, Tamara T and Stieger, Simon and Feiner, Leonhard F and Brandt, Johannes and Braren, Rickmer and Rueckert, Daniel and Kaissis, Georgios},
  journal={Nature Machine Intelligence},
  volume={6},
  number={7},
  pages={764--774},
  year={2024},
  publisher={Nature Publishing Group UK London}
}

@article{helber2019eurosat,
  title={Eurosat: A novel dataset and deep learning benchmark for land use and land cover classification},
  author={Helber, Patrick and Bischke, Benjamin and Dengel, Andreas and Borth, Damian},
  journal={IEEE Journal of Selected Topics in Applied Earth Observations and Remote Sensing},
  volume={12},
  number={7},
  pages={2217--2226},
  year={2019},
  publisher={IEEE}
}

@article{zhou2018patternnet,
  title={PatternNet: A benchmark dataset for performance evaluation of remote sensing image retrieval},
  author={Zhou, Weixun and Newsam, Shawn and Li, Congmin and Shao, Zhenfeng},
  journal={ISPRS Journal of Photogrammetry and Remote Sensing},
  volume={145},
  pages={197--209},
  year={2018},
  publisher={Elsevier}
}

@inproceedings{cao2018vggface2,
  title={Vggface2: A dataset for recognising faces across pose and age},
  author={Cao, Qiong and Shen, Li and Xie, Weidi and Parkhi, Omkar M and Zisserman, Andrew},
  booktitle={2018 13th IEEE International Conference on Automatic Face \& Gesture Recognition (FG 2018)},
  pages={67--74},
  year={2018},
  organization={IEEE}
}

@inproceedings{deng2009imagenet,
  title={Imagenet: A large-scale hierarchical image database},
  author={Deng, Jia and Dong, Wei and Socher, Richard and Li, Li-Jia and Li, Kai and Fei-Fei, Li},
  booktitle={2009 IEEE Conference on Computer Vision and Pattern Recognition},
  pages={248--255},
  year={2009},
  organization={Ieee}
}

@inproceedings{tian2020contrastive,
  title={Contrastive multiview coding},
  author={Tian, Yonglong and Krishnan, Dilip and Isola, Phillip},
  booktitle={European Conference on Computer Vision},
  pages={776--794},
  year={2020},
  organization={Springer}
}

@inproceedings{tong2025membership,
  title={Membership Inference Attacks on Tokenizers of Large Language Models},
  author={Tong, Meng and Du, Yuntao and Chen, Kejiang and Zhang, Weiming and Li, Ninghui},
  booktitle = {35th USENIX Security Symposium (USENIX Security 26)}, 
  year = {2026}, 
  publisher = {USENIX Association},
}

@inproceedings{du2025imitative,
  title={Imitative Membership Inference Attack},
  author={Du, Yuntao and Chen, Yuetian and Xiao, Hanshen and Ribeiro, Bruno and Li, Ninghui},
  booktitle = {35th USENIX Security Symposium (USENIX Security 26)}, 
  year = {2026}, 
  publisher = {USENIX Association}
}

@inproceedings{chen2026window,
  title={Window-based Membership Inference Attacks Against Fine-tuned Large Language Models},
  author={Chen, Yuetian and Du, Yuntao and Zhang, Kaiyuan and Kundu, Ashish and Fleming, Charles and Ribeiro, Bruno and Li, Ninghui},
  booktitle = {35th USENIX Security Symposium (USENIX Security 26)}, 
  year = {2026}, 
  publisher = {USENIX Association}
}

@article{wang2026inference,
  title={Inference Attacks Against Graph Generative Diffusion Models},
  author={Wang, Xiuling and Huang, Xin and Luo, Guibo and Xu, Jianliang},
  booktitle = {35th USENIX Security Symposium (USENIX Security 26)}, 
  year = {2026}, 
  publisher = {USENIX Association}
}

@article{li2025compleak,
  title={CompLeak: Deep Learning Model Compression Exacerbates Privacy Leakage},
  author={Li, Na and Gao, Yansong and Hu, Hongsheng and Kuang, Boyu and Fu, Anmin},
  booktitle = {35th USENIX Security Symposium (USENIX Security 26)}, 
  year = {2026}, 
  publisher = {USENIX Association},
}

@inproceedings{wang2026vidleaks,
  title={VidLeaks: Membership Inference Attacks Against Text-to-Video Models},
  author={Wang, Li and Chen, Wenyu and Yu, Ning and Li, Zheng and Guo, Shanqing},
  booktitle = {35th USENIX Security Symposium (USENIX Security 26)}, 
  year = {2026}, 
  publisher = {USENIX Association},
}

@article{wang2018dataset,
  title={Dataset distillation},
  author={Wang, Tongzhou and Zhu, Jun-Yan and Torralba, Antonio and Efros, Alexei A},
  journal={arXiv preprint arXiv:1811.10959},
  year={2018}
}

@inproceedings{zhaodataset,
  title={Dataset Condensation with Gradient Matching},
  author={Zhao, Bo and Mopuri, Konda Reddy and Bilen, Hakan},
  booktitle={International Conference on Learning Representations}
}

@inproceedings{cazenavette2022dataset,
  title={Dataset distillation by matching training trajectories},
  author={Cazenavette, George and Wang, Tongzhou and Torralba, Antonio and Efros, Alexei A and Zhu, Jun-Yan},
  booktitle={Proceedings of the IEEE/CVF Conference on Computer Vision and Pattern Recognition},
  pages={4750--4759},
  year={2022}
}

@inproceedings{dong2022privacy,
  title={Privacy for free: How does dataset condensation help privacy?},
  author={Dong, Tian and Zhao, Bo and Lyu, Lingjuan},
  booktitle={International Conference on Machine Learning},
  pages={5378--5396},
  year={2022},
  organization={PMLR}

}

@article{carlini2022no,
  title={No free lunch in" privacy for free: How does dataset condensation help privacy"},
  author={Carlini, Nicholas and Feldman, Vitaly and Nasr, Milad},
  journal={arXiv preprint arXiv:2209.14987},
  year={2022}
}

@inproceedings{liu2023backdoor, 
	title={Backdoor Attacks Against Dataset Distillation}, 
	author={Liu, Yugeng and Li, Zheng and Backes, Michael and Shen, Yun and Zhang, Yang}, 
	booktitle={Network and Distributed System Security Symposium (NDSS)}, 
	year={2023}
 }

@misc{tencentarc2024fluxmini,
  howpublished = {\url{https://huggingface.co/TencentARC/flux-mini}}
}

@inproceedings{dwork2006calibrating,
  title={Calibrating noise to sensitivity in private data analysis},
  author={Dwork, Cynthia and McSherry, Frank and Nissim, Kobbi and Smith, Adam},
  booktitle={Theory of Cryptography Conference},
  pages={265--284},
  year={2006},
  organization={Springer}
}

@inproceedings{abadi2016deep,
  title={Deep learning with differential privacy},
  author={Abadi, Martin and Chu, Andy and Goodfellow, Ian and McMahan, H Brendan and Mironov, Ilya and Talwar, Kunal and Zhang, Li},
  booktitle={Proceedings of the 2016 ACM SIGSAC Conference on Computer and Communications Security},
  pages={308--318},
  year={2016}
}

@inproceedings{yousefpour2021opacus,
  title={Opacus: User-Friendly Differential Privacy Library in PyTorch},
  author={Yousefpour, Ashkan and Shilov, Igor and Sablayrolles, Alexandre and Testuggine, Davide and Prasad, Karthik and Malek, Mani and Nguyen, John and Ghosh, Sayan and Bharadwaj, Akash and Zhao, Jessica and others},
  booktitle={NeurIPS 2021 Workshop Privacy in Machine Learning},
  year={2021}
}

@article{mensink2013distance,
  title={Distance-based image classification: Generalizing to new classes at near-zero cost},
  author={Mensink, Thomas and Verbeek, Jakob and Perronnin, Florent and Csurka, Gabriela},
  journal={IEEE Transactions on Pattern Analysis and Machine Intelligence},
  volume={35},
  number={11},
  pages={2624--2637},
  year={2013},
  publisher={IEEE}
}

@article{snell2017prototypical,
  title={Prototypical networks for few-shot learning},
  author={Snell, Jake and Swersky, Kevin and Zemel, Richard},
  journal={Advances in Neural Information Processing Systems},
  volume={30},
  year={2017}
}

@article{papyan2020prevalence,
  title={Prevalence of neural collapse during the terminal phase of deep learning training},
  author={Papyan, Vardan and Han, XY and Donoho, David L},
  journal={Proceedings of the National Academy of Sciences},
  volume={117},
  number={40},
  pages={24652--24663},
  year={2020},
  publisher={National Academy of Sciences}
}

@inproceedings{hanneural,
  title={Neural Collapse Under MSE Loss: Proximity to and Dynamics on the Central Path},
  author={Han, XY and Papyan, Vardan and Donoho, David L},
  booktitle={International Conference on Learning Representations},
  year={2022}
}

@book{vershynin2018high,
  title={High-dimensional probability: An introduction with applications in data science},
  author={Vershynin, Roman},
  volume={47},
  year={2018},
  publisher={Cambridge University Press}
}
\end{small}

\appendix

\subsection{Datasets}
\label{app:datasets}
\mypara{EuroSAT}~\cite{helber2019eurosat} is a satellite image classification 
benchmark derived from Sentinel-2 imagery, comprising 27000 labeled patches at $64{\times}64$ resolution across 10 land use and land cover classes---AnnualCrop, Forest, HerbaceousVegetation, Highway, Industrial, 
Pasture, PermanentCrop, Residential, River, and SeaLake.

\mypara{PatternNet}~\cite{zhou2018patternnet} is a high-resolution remote sensing benchmark consisting of 38 classes with 800 images per class at $256{\times}256$ resolution.
From these 38 categories we sample 5 classes for evaluation, namely runway marking, shipping yard, storage tank, tennis court, and transformer station.

\mypara{VGGFace2}~\cite{cao2018vggface2} contains over 3.3 million face images 
spanning more than 9000 identities with substantial variation in pose, age, illumination, and ethnicity. 
%Given the severe per-class imbalance of the original dataset, 
We sample 10 identities for evaluation, namely 
Adrienne Bailon, Alberto N\'u\~nez Feij\'oo, Aleksander Kwa\'sniewski, Alesha Dixon, Alex Salmond, Alexa Chung, Alexis Tsipras, Alfredo P\'erez Rubalcaba, Amber Heard, and Amy Adams.

\mypara{ImageNet10} comprises two 10-class subsets of 
ImageNet~\cite{deng2009imagenet}, denoted ImageNet10-A and ImageNet10-B, 
constructed by selecting the classes that SD1.5 
and Flux-mini can respectively generate with high visual fidelity. 
ImageNet10-A for SD1.5 consists of king penguin, Maltese dog, snow leopard, airliner, airship, container ship, soccer ball, sports car, trailer truck, and orange.
ImageNet10-B for Flux-mini consists of robin, lorikeet, American coot, Maltese dog, vizsla, coyote, red fox, tabby cat, snow leopard, and meerkat.

\mypara{ImageNet100.} 
ImageNet100~\cite{tian2020contrastive} is a widely adopted 100-class subset of ImageNet~\cite{deng2009imagenet} commonly used in representation learning and image classification benchmarks. 

\subsection{Target Classes and Bound Attributes for Adversarial \name}
\label{bind}
For each dataset evaluated under the adversarial \name pipeline (Section~\ref{invasive}), we select one target class $c^{\star}$ together with a benign attribute $z^{\star}$ that is visually natural, semantically compatible with $c^{\star}$, and rarely co-occurs with $c^{\star}$ in the real distribution. The selected pairs are detailed below.

\mypara{EuroSAT (Forest).}
The target class is \textit{Forest}. The bound attribute is a thin firebreak strip cutting. Such clearings appear infrequently in the real forest distribution, yet remain a plausible feature of forested landscapes and therefore preserve class semantics.

\mypara{PatternNet (Tennis Court).} 
The target class is \textit{tennis court}. The bound attribute is an orange playing surface, in contrast to the standard green and blue courts that overwhelmingly populate the real distribution. Orange surfaces are common in real-world tennis facilities, yet remain under-represented in PatternNet, satisfying the rarity-while-plausible criterion.

\mypara{ImageNet10 (Maltese dog).} 
The target class is \textit{Maltese dog}. The bound attribute is a snowy outdoor background. Snow scenes occur infrequently in the real distribution of this class but remain visually plausible and semantically compatible with the concept, satisfying the rarity-while-plausible criterion.

\mypara{VGGFace2 (Aleksander Kwa\'sniewski).} 
The target class is the identity \textit{Aleksander Kwa\'sniewski}. The bound attribute is a pair of black eyeglasses. Eyewear is a naturalistic accessory for portrait photographs and does not alter identity semantics, qualifying it as a benign yet rare attribute.

\subsection{Proof of Theorem~\ref{thm:atypicality}}
\label{app:atypicality}

\noindent\textit{Proof:}
Let
$\boldsymbol{\Delta}
 = \boldsymbol{\mu}_c^{\mathrm{mix}}
 - \boldsymbol{\mu}_c^{\mathrm{real}}
 = \frac{\lambda}{1+\lambda}
   \bigl(\boldsymbol{\mu}_c^{\mathrm{syn}}
        - \boldsymbol{\mu}_c^{\mathrm{real}}\bigr)$.
Expanding the squared norm:
\begin{align}
  \bigl\|\phi(\mathbf{x}) - \boldsymbol{\mu}_c^{\mathrm{mix}}\bigr\|_2^2
  &= \bigl\|\phi(\mathbf{x})
     - \boldsymbol{\mu}_c^{\mathrm{real}}
     - \boldsymbol{\Delta}\bigr\|_2^2 \nonumber\\
  &= \bigl\|\phi(\mathbf{x})
     - \boldsymbol{\mu}_c^{\mathrm{real}}\bigr\|_2^2 \nonumber\\
  &\quad - 2\bigl\langle\phi(\mathbf{x})
     - \boldsymbol{\mu}_c^{\mathrm{real}},\boldsymbol{\Delta}\bigr\rangle \nonumber\\
  &\quad + \|\boldsymbol{\Delta}\|_2^2.
\end{align}
Taking expectation over $\mathbf{x}\in\mathcal{D}_c^{\mathrm{real}}$
and noting that
$\mathbb{E}[\phi(\mathbf{x})] = \boldsymbol{\mu}_c^{\mathrm{real}}$,
the cross-product term vanishes:
\begin{equation}
  \label{eq:expect}
  \mathbb{E}\!\left[\bigl\|\phi(\mathbf{x})
    - \boldsymbol{\mu}_c^{\mathrm{mix}}\bigr\|_2^2\right]
  = \mathbb{E}\!\left[\bigl\|\phi(\mathbf{x})
    - \boldsymbol{\mu}_c^{\mathrm{real}}\bigr\|_2^2\right]
  + \|\boldsymbol{\Delta}\|_2^2.
\end{equation}
Applying the distributional gap assumption~\eqref{eq:gap}:
\begin{equation}
  \|\boldsymbol{\Delta}\|_2
  = \frac{\lambda}{1+\lambda}
    \bigl\|\boldsymbol{\mu}_c^{\mathrm{syn}}
          - \boldsymbol{\mu}_c^{\mathrm{real}}\bigr\|_2
  > \frac{\lambda}{1+\lambda}\,\delta,
\end{equation}
and therefore
$\|\boldsymbol{\Delta}\|_2^2
 > \bigl(\tfrac{\lambda}{1+\lambda}\bigr)^2\delta^2$.
Substituting back into~\eqref{eq:expect} yields~\eqref{eq:atypicality}.

\subsection{Proof of Theorem~\ref{thm:mem-amp}}
\label{app:mem-amp}

\noindent \textit{Setup.} We adopt the prototype form~\cite{snell2017prototypical} of the trained classifier $\mathcal{A}$, which is well-justified for standard classification networks operating near zero training error (the (C2) regime). Specifically, for any training set $S$ used to train the classifier (in our context $S\in\{\mathcal{D}_{\mathrm{mix}},\mathcal{D}_{\mathrm{real}}\}$ or their leave-one-out variants), we assume
\begin{equation}
\log \Pr_{h\leftarrow\mathcal{A}(S)}[h(x)=c] = -\tfrac{1}{2\sigma^2}\|\phi(x)-\hat{\boldsymbol{\mu}}_c(S)\|_2^2 + Z(x,S),
\label{eq:prototype-form}
\end{equation}
where $\hat{\boldsymbol{\mu}}_c(S):=\tfrac{1}{|S_c|}\sum_{(x,y)\in S_c}\phi(x)$ is the empirical class-$c$ centroid in $S$, $S_c$ denotes the class-$c$ subset of $S$, $\sigma^2>0$ is a temperature parameter, and $Z(x,S)$ is a class-independent normalizer.
This form is supported by Neural Collapse~\cite{papyan2020prevalence,hanneural} under both cross-entropy and MSE losses and by the asymptotic equivalence between linear softmax and nearest-class-mean classifiers on fixed features~\cite{mensink2013distance,snell2017prototypical}; our ResNet-18 trained to 100\% training accuracy operates in this regime.

\noindent \textit{Proof:}
The proof proceeds in three steps: (i) reducing the memorization gap to a leave-one-out (LOO) probability gap via (C2), (ii) translating the LOO probability gap into a feature-distance gap via the prototype form~\eqref{eq:prototype-form}, and (iii) lower-bounding the feature-distance gap via (C1) and Theorem~\ref{thm:atypicality}.

\noindent \textit{Step 1 (Reduction to LOO probability gap).}
By Definition~\ref{def:mem-score}, for any training set $S$,
\[
\mathrm{mem}(\mathcal{A},S,z_i) = \Pr_{h\leftarrow\mathcal{A}(S)}[h(x_i)=y_i] - \Pr_{h\leftarrow\mathcal{A}(S\setminus\{z_i\})}[h(x_i)=y_i].
\]
Let $p(x_i;S):=\Pr_{h\leftarrow\mathcal{A}(S)}[h(x_i)=y_i]$. By (C2), $p(x_i;S)\geq 1-\epsilon$, hence
\[
\mathrm{mem}(\mathcal{A},S,z_i)\geq(1-\epsilon)-p(x_i;S\setminus\{z_i\}).
\]
Subtracting the two instantiations $S=\mathcal{D}_{\mathrm{mix}}$ and $S=\mathcal{D}_{\mathrm{real}}$ and taking expectation over $z_i\sim\mathcal{D}_{\mathrm{real}}^c$, the $(1-\epsilon)$ terms cancel:
\begin{equation}
\begin{aligned}
&\mathbb{E}_{z_i\sim\mathcal{D}_{\mathrm{real}}^c}
\Big[
\mathrm{mem}(\mathcal{A},\mathcal{D}_{\mathrm{mix}},z_i)
- \mathrm{mem}(\mathcal{A},\mathcal{D}_{\mathrm{real}},z_i)
\Big] \\
&\geq \mathbb{E}_{z_i\sim\mathcal{D}_{\mathrm{real}}^c}
\Big[
p(x_i;\mathcal{D}_{\mathrm{real}}\!\setminus\!\{z_i\})
- p(x_i;\mathcal{D}_{\mathrm{mix}}\!\setminus\!\{z_i\})
\Big].
\end{aligned}
\label{eq:step1}
\end{equation}

\noindent \textit{Step 2 (LOO probability gap as a feature-distance gap).}
By the prototype form~\eqref{eq:prototype-form}, for $c=y_i$,
\[
\begin{aligned}
p(x_i;S\setminus\{z_i\}) 
= \exp\!\Big(
&-\tfrac{1}{2\sigma^2}\|\phi(x_i)-\hat{\boldsymbol{\mu}}_c(S\setminus\{z_i\})\|_2^2 \\
&+ Z(x_i,S\setminus\{z_i\})
\Big)
\end{aligned}
\]
The squared distance $\|\phi(x_i)-\hat{\boldsymbol{\mu}}_c(S\setminus\{z_i\})\|_2^2$ is the central quantity governing $p(x_i;S\setminus\{z_i\})$ under~\eqref{eq:prototype-form}. To analyze its expectation under the random draw of $z_i$, we denote
\begin{equation}
D(S) \;:=\; \mathbb{E}_{z_i\sim\mathcal{D}_{\mathrm{real}}^c}\!\left[\bigl\|\phi(x_i)-\hat{\boldsymbol{\mu}}_c(S\setminus\{z_i\})\bigr\|_2^2\right],
\label{eq:D-definition}
\end{equation}
which captures the average peripheral displacement of a real sample relative to the empirical class-$c$ centroid of the remaining training set $S\setminus\{z_i\}$. Three observations follow immediately. First, by (C2), $\|\phi(x_i)-\hat{\boldsymbol{\mu}}_c(S\setminus\{z_i\})\|_2^2 = O(\sigma^2\epsilon)$, so the exponent in~\eqref{eq:prototype-form} lies in a near-zero regime where the first-order Taylor expansion $\exp(-u)=1-u+O(u^2)$ applies uniformly. Second, since $Z(x_i,S\setminus\{z_i\})=O(\epsilon)$ under (C2), $e^{Z}=1+O(\epsilon)$ and is absorbed into the higher-order terms. Third, the normalizer $Z(x_i,S\setminus\{z_i\})$ depends only on non-target class centroids; removing a single target-class sample $z_i$ perturbs $Z$ by $O(1/N_{c'})$ for each non-target class $c'$, which becomes negligible relative to the target-class distance term as $N_{c'}\to\infty$.

Combining these observations, the LOO probability gap reduces to a linear functional of the distance gap:
\begin{equation}
\begin{aligned}
&\mathbb{E}_{z_i\sim\mathcal{D}_{\mathrm{real}}^c}[p(x_i;\mathcal{D}_{\mathrm{real}}\!\setminus\!\{z_i\})] - \mathbb{E}_{z_i\sim\mathcal{D}_{\mathrm{real}}^c}[p(x_i;\mathcal{D}_{\mathrm{mix}}\!\setminus\!\{z_i\})] \\
&\quad=\; \tfrac{1}{2\sigma^2}\bigl(D(\mathcal{D}_{\mathrm{mix}})-D(\mathcal{D}_{\mathrm{real}})\bigr) + O(\epsilon^2) + o(1)_{\text{norm}},
\end{aligned}
\label{eq:step2}
\end{equation}
where $o(1)_{\text{norm}}$ collects the normalizer perturbation and vanishes as the per-class sample sizes grow.

\noindent \textit{Step 3 (Lower bound on $D(\mathcal{D}_{\mathrm{mix}})-D(\mathcal{D}_{\mathrm{real}})$).}
We now establish a lower bound on the distance gap in~\eqref{eq:step2}. Define the population centroid offset
\[
\boldsymbol{\Delta} \;:=\; \tfrac{\lambda}{1+\lambda}\bigl(\boldsymbol{\mu}_c^{\mathrm{syn}}-\boldsymbol{\mu}_c^{\mathrm{real}}\bigr),
\]
which by (C1) satisfies $\|\boldsymbol{\Delta}\|_2 \geq \tfrac{\lambda}{1+\lambda}\delta$ and hence $\|\boldsymbol{\Delta}\|_2^2 \geq \bigl(\tfrac{\lambda}{1+\lambda}\bigr)^2\delta^2$.

\emph{Population-level distance gap.} Expanding $\|\phi(x_i)-\boldsymbol{\mu}_c^{\mathrm{mix}}\|_2^2 = \|\phi(x_i)-\boldsymbol{\mu}_c^{\mathrm{real}}-\boldsymbol{\Delta}\|_2^2$ and taking expectation over $z_i\sim\mathcal{D}_{\mathrm{real}}^c$ (the cross term vanishes since $\mathbb{E}[\phi(x_i)]=\boldsymbol{\mu}_c^{\mathrm{real}}$):
\begin{equation}
\begin{aligned}
\mathbb{E}_{z_i\sim\mathcal{D}_{\mathrm{real}}^c}
\|\phi(x_i)-\boldsymbol{\mu}_c^{\mathrm{mix}}\|_2^2
&- \mathbb{E}_{z_i\sim\mathcal{D}_{\mathrm{real}}^c}
\|\phi(x_i)-\boldsymbol{\mu}_c^{\mathrm{real}}\|_2^2 \\
&= \|\boldsymbol{\Delta}\|_2^2 \\
&\geq \left(\tfrac{\lambda}{1+\lambda}\right)^{\!2}\delta^2.
\end{aligned}
\label{eq:population-gap}
\end{equation}
\emph{From population to empirical centroids.} The quantity $D(S)$ in~\eqref{eq:D-definition} uses empirical centroids rather than population centroids. We control this discrepancy via the strong law of large numbers together with a finite-sample concentration argument. Let $\hat{\boldsymbol{\mu}}_c(\mathcal{D}_{\mathrm{real}}\!\setminus\!\{z_i\})$ be the average of $N_c-1$ i.i.d.\ samples from $\mathcal{N}(\boldsymbol{\mu}_c^{\mathrm{real}},\Sigma_r)$, and $\hat{\boldsymbol{\mu}}_c(\mathcal{D}_{\mathrm{mix}}\!\setminus\!\{z_i\})$ be the weighted average of $N_c-1$ real samples and $\lambda N_c$ synthetic samples. Under the bounded-feature assumption $\|\phi(x)\|_2\leq B$ standard for normalized CLIP-style encoders, by the multivariate concentration inequality~\cite{vershynin2018high},
\begin{equation}
\begin{aligned}
\bigl\|\hat{\boldsymbol{\mu}}_c(\mathcal{D}_{\mathrm{real}}\!\setminus\!\{z_i\}) - \boldsymbol{\mu}_c^{\mathrm{real}}\bigr\|_2 &= O_p\!\left(\tfrac{B}{\sqrt{N_c-1}}\right), \\
\bigl\|\hat{\boldsymbol{\mu}}_c(\mathcal{D}_{\mathrm{mix}}\!\setminus\!\{z_i\}) - \boldsymbol{\mu}_c^{\mathrm{mix}}\bigr\|_2 &= O_p\!\left(\tfrac{B}{\sqrt{(1+\lambda)N_c}}\right).
\end{aligned}
\label{eq:concentration}
\end{equation}
Substituting empirical centroids into the squared-norm expansion in~\eqref{eq:population-gap} and propagating the centroid concentration~\eqref{eq:concentration},
\begin{equation}
D(\mathcal{D}_{\mathrm{mix}}) - D(\mathcal{D}_{\mathrm{real}}) \;\geq\; \left(\tfrac{\lambda}{1+\lambda}\right)^{\!2}\delta^2 - O\!\left(\tfrac{B^2}{N_c}\right) - O\!\left(\tfrac{B\,\delta}{\sqrt{N_c}}\right).
\label{eq:step3}
\end{equation}
The two correction terms are dominated by the leading $\bigl(\tfrac{\lambda}{1+\lambda}\bigr)^2\delta^2$ whenever $N_c \gtrsim B^2\delta^{-2}$. Under the normalized CLIP encoder ($B=1$) and the empirically observed gap $\delta\approx 0.35$ on VGGFace2 (Table~\ref{tab:delta}), this threshold is $\approx 8$, far below our setting of $N_c=100$, leaving comfortable margin for the leading term to dominate.

\noindent \textit{Combining.}
Substituting~\eqref{eq:step3} into~\eqref{eq:step2} and absorbing the $O(\epsilon^2)$, $o(1)_{\text{norm}}$, and finite-sample correction terms into a single $r(N_c,\epsilon)$ that vanishes under $N_c\to\infty$ and $\epsilon\to 0$,
\begin{equation}
\begin{aligned}
&\mathbb{E}_{z_i\sim\mathcal{D}_{\mathrm{real}}^c}\!\left[p(x_i;\mathcal{D}_{\mathrm{real}}\!\setminus\!\{z_i\})\right] - \mathbb{E}_{z_i\sim\mathcal{D}_{\mathrm{real}}^c}\!\left[p(x_i;\mathcal{D}_{\mathrm{mix}}\!\setminus\!\{z_i\})\right] \\
&\quad\geq\; \tfrac{1}{2\sigma^2}\!\left(\tfrac{\lambda}{1+\lambda}\right)^{\!2}\delta^2 - r(N_c,\epsilon).
\end{aligned}
\label{eq:step4}
\end{equation}
Substituting~\eqref{eq:step4} into~\eqref{eq:step1} and defining
\[
\eta(\delta,\lambda) \;:=\; \tfrac{1}{2\sigma^2}\!\left(\tfrac{\lambda}{1+\lambda}\right)^{\!2}\delta^2,
\]
we obtain
\[
\begin{aligned}
&\mathbb{E}_{z_i\sim\mathcal{D}_{\mathrm{real}}^c}\!\left[\mathrm{mem}(\mathcal{A},\mathcal{D}_{\mathrm{mix}},z_i)\right] - \mathbb{E}_{z_i\sim\mathcal{D}_{\mathrm{real}}^c}\!\left[\mathrm{mem}(\mathcal{A},\mathcal{D}_{\mathrm{real}},z_i)\right] \\
&\quad\geq\; \eta(\delta,\lambda) - r(N_c,\epsilon).
\end{aligned}
\]
Whenever $\delta>0$ and $\lambda>0$, the leading term $\eta(\delta,\lambda)$ dominates as $N_c\to\infty$ and $\epsilon\to 0$, yielding strict amplification $\mathbb{E}_{z_i\sim\mathcal{D}_{\mathrm{real}}^c}[\mathrm{mem}(\mathcal{A},\mathcal{D}_{\mathrm{mix}},z_i)] > \mathbb{E}_{z_i\sim\mathcal{D}_{\mathrm{real}}^c}[\mathrm{mem}(\mathcal{A},\mathcal{D}_{\mathrm{real}},z_i)]$.% \qed

\begin{figure}[t]
    \centering
    \includegraphics[trim=0 0 0 0,clip,width=0.45\textwidth]{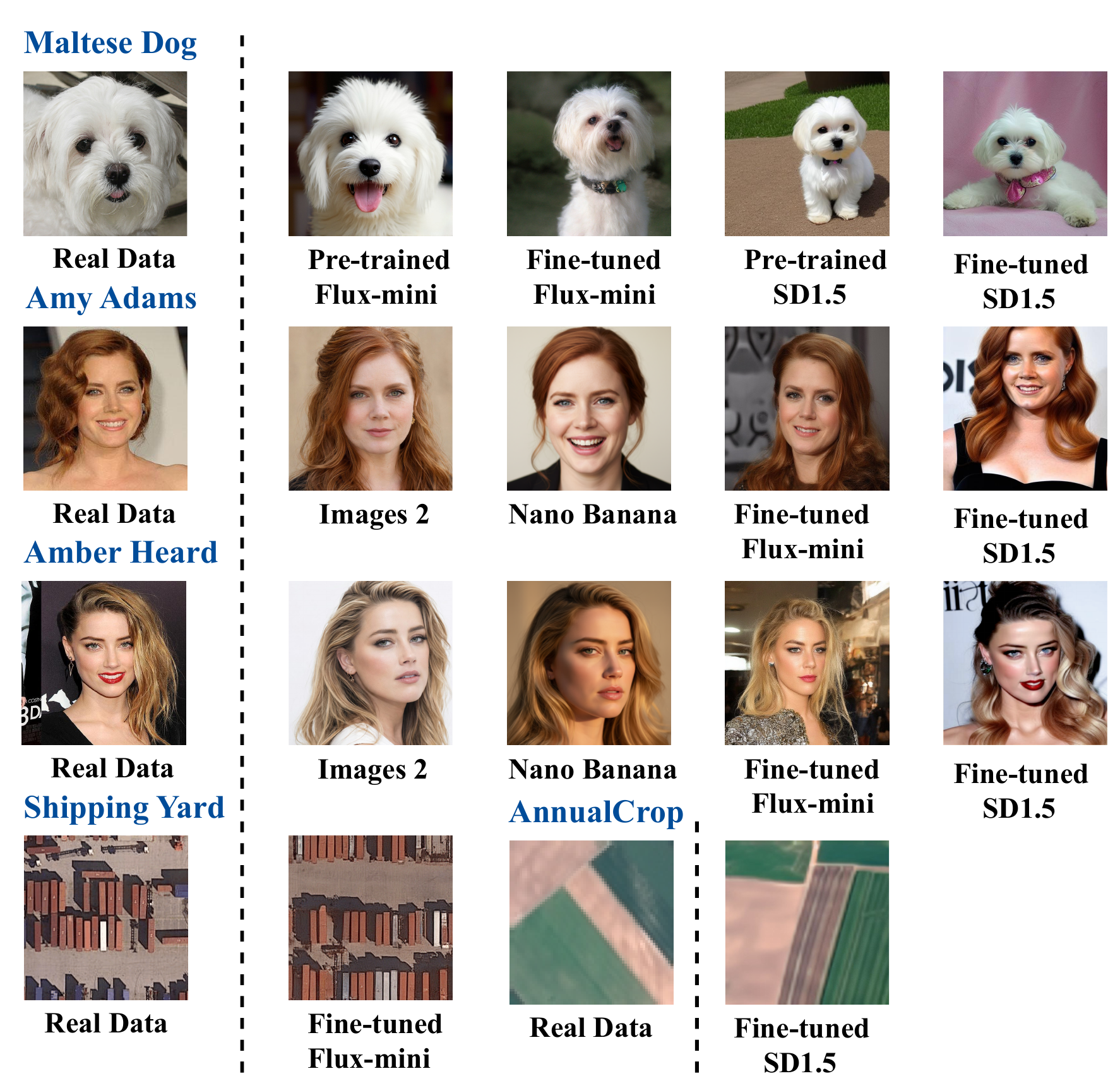}
    \caption{Real data and synthetic data generated by different backbones.}
    \label{fig:gen-quality}
\end{figure}

\subsection{\name Result}
\label{app:result}
Figure~\ref{fig:gen-quality} illustrates real and synthetic images generated by different T2I backbones. Table~\ref{tab:acc-pretrained} reports ResNet-18 classification accuracy across datasets and T2I models. Table~\ref{tab:resultsflux} summarizes the membership inference attack performance under different attack methods on fine-tuned Flux-mini. Figure~\ref{fig:radio} shows how the real-to-synthetic mixing ratio $\lambda$ affects model utility and privacy leakage on EuroSAT. Figure~\ref{fig:coat} illustrates real data and synthetic data produced by the vanilla, bind, and coat schemes.

\begin{figure}[!htbp]
    \centering
    \includegraphics[trim=0 0 0 0,clip,width=0.45\textwidth]{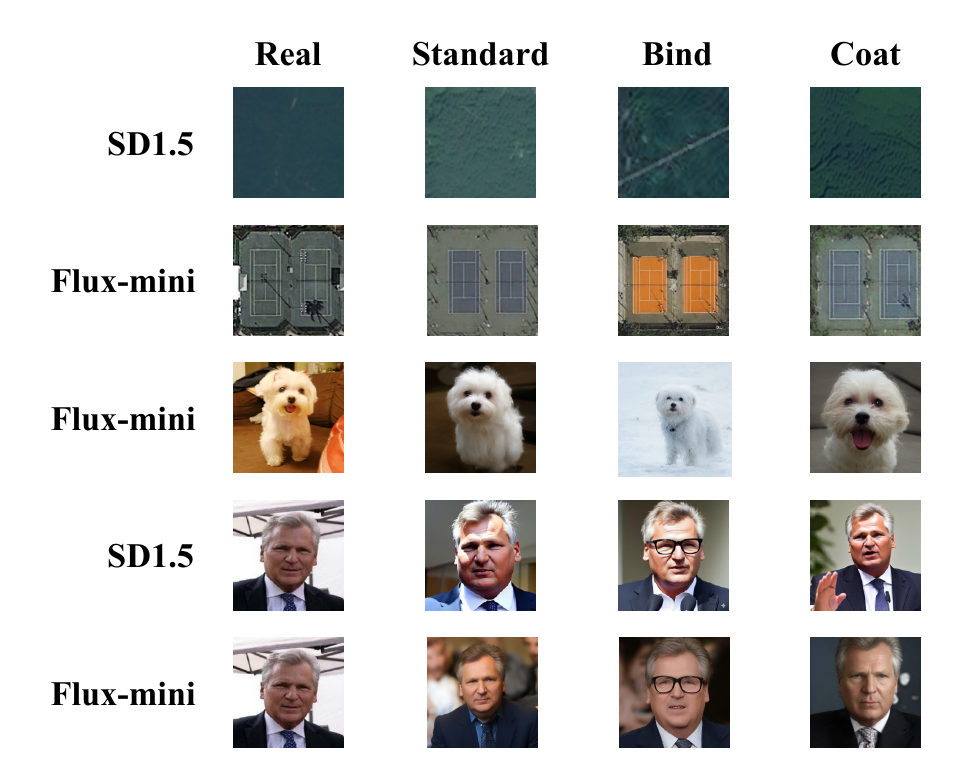}
    \caption{Real data and synthetic data generated by the vanilla backbone, bind, and coat schemes.}
    \label{fig:coat}
\end{figure}

\begin{table}[!htbp]
\centering
\caption{ResNet-18 model classification accuracy across datasets and pre-trained and commercial T2I models.}
\label{tab:acc-pretrained}
\scalebox{0.8}{
\begin{tabular}{llcccccc}
  \toprule
  \multirow{1}{*}{T2I Model}
    &\multirow{1}{*}{Dataset}
    & \multicolumn{3}{c}{Train Acc (\%)}
    & \multicolumn{3}{c}{Test Acc (\%)} \\
  \cmidrule(lr){3-5}\cmidrule(lr){6-8}
  &  & B1 & B2 & RSMT & B1 & B2 & RSMT \\
  \midrule
  \makecell[l]{Nano Banana}&VGGFace2(3)
  &100 &100 &100    &66.6 &71.0 &71.6 \\
  \midrule
  \makecell[l]{Images 2.0}&VGGFace2(3)
  &100 &100 &100    &66.6 &75.3 &76.3 \\
  \midrule
  \makecell[l]{SD1.5}&ImageNet100
  &99.9 &99.9 &99.9   &46.5 &62.0 &62.3 \\
  \midrule
  \makecell[l]{Flux-mini}&ImageNet10-B
  &100  &100  &100  &58.6 &64.1 &64.4 \\
  \bottomrule
\end{tabular}}
\end{table}

\begin{table}[!htbp]
\centering
\caption{Attack performance of different attacks on varyin datasets and fine-tuned Flux-mini.}
\label{tab:resultsflux}
\scalebox{0.7}{
\begin{tabular}{ll cc>{\columncolor{mixbg}}c cc>{\columncolor{mixbg}}c cc>{\columncolor{mixbg}}c}
\toprule
\multirow{2}{*}{Dataset}
  & \multirow{2}{*}{Attack}
  & \multicolumn{3}{c}{MIA Acc (\%)}
  & \multicolumn{3}{c}{MIA Auc (\%)}
  & \multicolumn{3}{c}{TPR @ 0.1\% FPR (\%)} \\
\cmidrule(lr){3-5} \cmidrule(lr){6-8} \cmidrule(lr){9-11}
   &Method
  & B1 & B2 & RSMT
  & B1 & B2 & RSMT
  & B1 & B2 & RSMT \\
\midrule

\multirow{4}{*}{\makecell[l]{VGGFace2(10)}}
  &~\cite{song2019privacy}
    &85.3 & 77.1 & \textbf{85.8} &88.5& 80.4 & \textbf{88.8} &0.3 &0.0&\textbf{0.8} \\
  &~\cite{yuan2022membership}
    &\textbf{85.0} & 76.1 &84.8 & \textbf{89.4} &79.7 &88.5 &\textbf{1.6}&0.2 &1.2 \\
  &~\cite{shokri2017membership}
    &81.6 & 75.3 & \textbf{81.8} &\textbf{87.2} & 78.4 &85.8 &\textbf{1.0}&0.0 &0.3 \\
  &~\cite{ZarifzadehLS24}
    &\textbf{76.5} & 68.6 & 76.3 &\textbf{74.8} & 62.7 & 72.5 &\textbf{3.6}&0.0 &1.1 \\
\midrule

\multirow{4}{*}{\makecell[l]{PatternNet(5)}}
  &~\cite{song2019privacy}
  &\textbf{73.0 }& 58.7 & 69.4 & \textbf{77.9} & 60.3 & 74.0 &0.0 &0.0 &0.0 \\
  &~\cite{yuan2022membership}
    &\textbf{66.5}&  59.2 & 64.0 &\textbf{71.4} &59.1 & 64.5 &7.0 &0.0 &\textbf{9.0} \\
  &~\cite{shokri2017membership}
    & \textbf{73.0}&48.3 & 67.4 &\textbf{78.4} & 50.2 & 74.3&1.6 &0.8 &\textbf{11.2} \\
  &~\cite{ZarifzadehLS24}
  & \textbf{63.4} &  56.8 & 58.0 & \textbf{58.0} & 55.3 & 57.7 &0.0 &0.0 &0.0 \\
\midrule

\multirow{4}{*}{\makecell[l]{ImageNet10-B}}
  &~\cite{song2019privacy}
    &\textbf{82.3} &72.5 &81.0 &\textbf{83.2} &74.5 &80.2 &0.0&0.0&0.0\\
  &~\cite{yuan2022membership}
    &\textbf{82.9} &73.3 &80.0 &\textbf{87.7} &77.9 &85.9 &\textbf{0.2}&0.0 &0.1 \\
  &~\cite{shokri2017membership}
    &76.6 &63.7 &\textbf{77.8 }&83.6 &70.9 &\textbf{85.6} &\textbf{2.7}&0.4&1.0\\
  &~\cite{ZarifzadehLS24}
    &\textbf{75.2} &69.1 &71.2 &\textbf{78.6} &73.0 &75.2 &\textbf{4.3}&2.3&3.1\\

\bottomrule
\end{tabular}%
}
\end{table}

\begin{figure}[!htbp]
    \centering
    \includegraphics[trim=0 0 0 0,clip,width=0.4\textwidth]{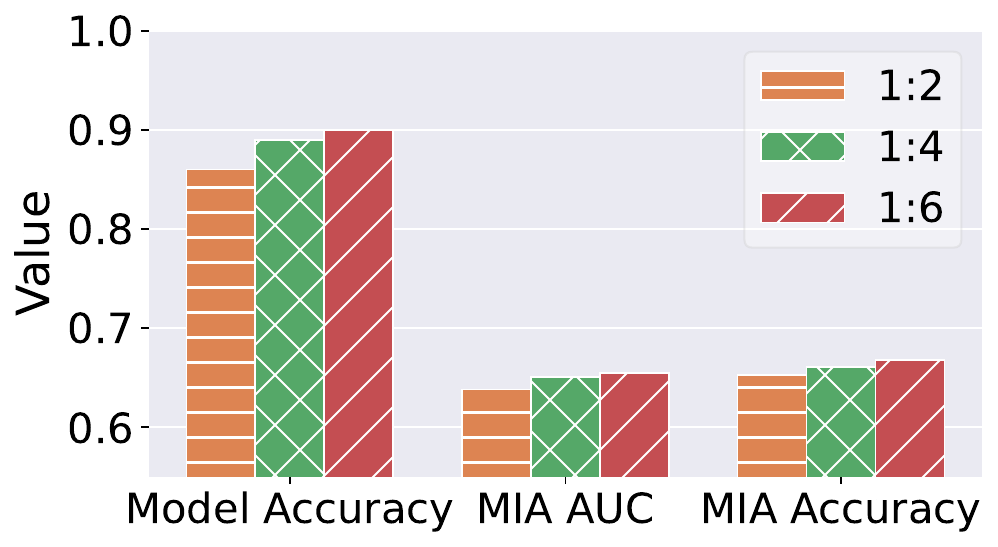}
    \caption{Effect of the real-to-synthetic mixing ratio $\lambda$ on utility and privacy leakage (EuroSAT, 100 real training samples per class fixed). 
    %As λ increases from 2 to 6, both the classification accuracy of the mix-trained model (left) and the MIA effectiveness against real training samples in terms of AUC (middle) and attack accuracy (right) rise monotonically. 
    }
    \label{fig:radio}
\end{figure}

\subsection{Attack Performance under Differential Privacy}
\label{app:dp}

\begin{table}[h]
\caption{Attack performance under DP-SGD ($\delta=10^{-5}$, $C=1.0$) on SD1.5/VGGFace2(10) using attack~\cite{song2019privacy}. The downstream classifier is a CNN compatible with Opacus.}
\label{tab:dp}
\centering
\small
\setlength{\tabcolsep}{4pt}
\scalebox{0.8}{
\begin{tabular}{lccccc}
\toprule
\multirow{2}{*}{Setting} & \multicolumn{2}{c}{Model Acc (\%)} & \multicolumn{3}{c}{MIA Performance (\%)} \\
\cmidrule(lr){2-3} \cmidrule(lr){4-6}
 & Train Acc & Test Acc & MIA Acc & MIA AUC & TPR@$0.1\%$FPR \\
\midrule
B1   & 46.6 & 43.0 & 53.6 & 53.2 & 0.0 \\
B2   & 49.6 & 47.1 & 52.1 & 52.0 & 0.2 \\
RSMT & 46.6 & 46.7 & 51.0 & 50.3 & 0.1 \\
\bottomrule
\end{tabular}}
\end{table}

Differential privacy (DP)~\cite{dwork2006calibrating} is a widely adopted defense for mitigating privacy leakage in deep learning, and we further examine whether DP-SGD~\cite{abadi2016deep} can suppress the RSMT-amplified leakage. Following standard practice, we implement DP-SGD via the Opacus toolkit~\cite{yousefpour2021opacus} with privacy parameters $(\delta = 10^{-5})$ and per-sample gradient clipping bound $C = 1.0$. We evaluate on SD1.5/VGGFace2(10) and report the attack performance of~\cite{song2019privacy} in Table~\ref{tab:dp}.

As shown in Table~\ref{tab:dp}, DP-SGD reduces the MIA AUC of RSMT, B1, and B2 to $50.3\%$, $53.2\%$, and $52.0\%$ respectively. This suppression follows directly from the mechanism of DP-SGD, whose per-sample gradient clipping and Gaussian noise injection bound the influence of any individual real sample on model parameters and thereby cap the sample-specific memorization that MIAs exploit.  RSMT is particularly susceptible to this suppression, since it amplifies leakage by forcing the model to memorize peripheral real samples through enlarged per-sample gradients, which is precisely what DP-SGD's per-sample clipping is designed to bound.

\end{document}